\documentclass[a4paper,18pt]{article}
\usepackage[top=2cm, bottom=1.5cm]{geometry} 
\usepackage[utf8]{inputenc}
\usepackage[american]{babel}  
\usepackage{amsmath,amssymb}
\usepackage{amsthm}
\usepackage{mathtools}
\usepackage{hyperref}
\usepackage{enumerate}
\usepackage{paralist}
\usepackage{graphics}
\usepackage[commentColor=black]{algpseudocodex}
\usepackage{algorithm}
\usepackage{float}
\usepackage{todonotes}
\usepackage[many]{tcolorbox}

\newtcolorbox{header}{
  fontupper = \it\color{black}, 
  boxrule = 0.5pt,
  colframe = black,
  colback = white,
  rounded corners,
  arc = 5pt   
}

\newcommand{\almosttextwidth}{0.95\textwidth}

\newcommand{\From}{~\textbf{from}~ }
\newcommand{\To}{~\textbf{to}~ }
 
\newcommand{\fun}[1]{\textsc{#1}} 
\newcommand{\arr}[1]{\texttt{#1}} 
\newcommand{\var}[1]{\texttt{#1}} 

\usepackage{tikz}
\usetikzlibrary{arrows}
\newtheorem{lemma}{Lemma}
\newtheorem{definition}{Definition}
\newtheorem{theorem}{Theorem}
\newtheorem{proposition}{Proposition}
\newtheorem{corollary}{Corollary}
\newtheorem{remark}{Remark}
\newtheorem*{notation}{Notation}
\newtheorem*{convention}{Convention}
\newtheorem{example}{Example}
\newtheorem{claim}{Claim}
\newtheorem{openpb}{Open problem}

\newtheorem{fact}{Fact}

\def\lc{\left\lceil}   
\def\rc{\right\rceil}
\def\lf{\left\lfloor}   
\def\rf{\right\rfloor}

\newcommand{\modop}{\;\mathtt{mod}\;}
\newcommand{\divop}{\;\mathtt{div}\;}

\newcommand{\xorop}{\mathtt{xor}}
\newcommand{\orop}{\mathtt{or}}
\newcommand{\andop}{\mathtt{and}}
\newcommand{\notop}{\mathtt{not}}

\newcommand{\subop}{\dot{-}}

\newcommand{\op}{\mathtt{op}}

\DeclareMathOperator{\linppop}{\ensuremath{\textsc{LinPP}\!}.\ensuremath{\!\mathtt{op}}}
\DeclareMathOperator{\cstpop}{\ensuremath{\textsc{CstP}\!}.\ensuremath{\!\mathtt{op}}}
\DeclareMathOperator{\cstpopprime}{\ensuremath{\textsc{CstP}2}.\ensuremath{\!\mathtt{op}}}

\DeclareMathOperator{\cdlin}{\ensuremath{\textsc{Cst}}\ensuremath{\textsc{Delay}_{\mathtt{lin}}}}

\DeclareMathOperator{\cpp}{\ensuremath{\fun{CstPP}}}

\newcommand{\OP}{\mathtt{OP}}

\begin{document}

\title{Constant Time with Minimal Preprocessing,\\ a Robust and Extensive Complexity Class\footnote{In Honour of Yuri Gurevich at the occasion of his 85th Birthday.}
}

\date{\today}

\author{Étienne Grandjean\footnote{Corresponding author, \texttt{etienne.grandjean@unicaen.fr}.} \\
\small{Univ. Caen Normandie, ENSICAEN, CNRS, Normandie Univ. GREYC UMR 6072, F-14000, France}\\ 
Louis Jachiet\footnote{\texttt{louis.jachiet@telecom-paris.fr}.} \\
\small{LTCI, Télécom Paris, Institut polytechnique de Paris, F-91120, France}.
}

\maketitle

\tableofcontents

\paragraph{Abstract:} 
Although the notion of computation with preprocessing is commonly used in computer science, its computational complexity has been little studied.  
In this paper, we study the class $\cpp$ of operations $\mathtt{op}: \mathbb{N}^k\to\mathbb{N}$, of any fixed arity $k\ge 1$, satisfying the following property: 
for each fixed integer $d\ge 1$, there exists an algorithm for a RAM machine which, for any input integer 
$N\ge 2$, 

• pre-computes some tables in $O(N)$ time,
 
• then reads $k$ operands $x_1,\ldots,x_k<N^d$ and computes $\op(x_1,\dots,x_k)$ in \emph{constant time}.  

\medskip
We show that the $\cpp$ class is robust and extensive and satisfies several closure properties.  
It is invariant depending on whether the set of primitive operations of the RAM is $\{+\}$, or $\{+,-,\times,\mathtt{div},\mathtt{mod}\}$, or any set of operations in $\cpp$ provided it includes $+$.  
We prove that the $\cpp$ class is closed under composition 
and, for fast growing functions, is closed under inverse.  
We also show that in the definition of $\cpp$ the constant-time procedure can be reduced to a single return instruction.
Finally, we establish that linear preprocessing time is not essential in the definition of
the $\cpp$ class: this class is not modified if the preprocessing time is increased to $O(N^c)$, for any fixed $c>1$, or conversely, is reduced to $N^{\varepsilon}$, for any positive $\varepsilon<1$ (provided the set of primitive operation includes $+$, $\mathtt{div}$ and $\mathtt{mod}$). 
To complete the picture, we demonstrate that the $\cpp$ class degenerates if the preprocessing time reduces to $N^{o(1)}$.  

\smallskip
It is significant that the robustness of $\cpp$ entails the robustness of more classical complexity classes on RAMs such as $\textsc{LinTime}$, and more generally $\textsc{Time}(T(N))$ for usual time functions~$T$, or even the enumeration complexity class $\cdlin$.


\section{Introduction}
The notion of computation with preprocessing is in common use in
computer science.  The goal is to answer questions asked by users as
quickly as possible.  The price to pay is a computation of auxiliary
data, an off-line phase which may be quite long because it is carried
out once beforehand. This has long been the case in database systems
where a set of indexes are preprocessed in order to instantly answer the most
frequently asked queries~\cite{Date81,UllmanW97}.  In
artificial intelligence, what we call knowledge compilation serves a
comparable purpose as Darwiche and Marquis~\cite{DarwicheM02,
  Marquis15} write: ``The key motivation behind knowledge compilation
is to push as much of the computational overhead into the off-line
phase, which is amortized over all on-line queries.''

\medskip
Since reading an input of size $N$ requires time $N$, it is well accepted that the minimum time complexity class of $\mathtt{Input}\mapsto\mathtt{Output}$ problems is linear time 
$O(N)$ which turns out to be a robust complexity class~\cite{Grandjean96,GrandjeanSchwentick02,Schwentick97}.  
For problems computable with preprocessing, it seems clear that the minimum preprocessing time is also $O(N)$ to read the input and the minimum computation time is $O(1)$ if the question asked by the user and its response are of constant size~\cite{BaganDGO08}.

\medskip
In the most basic version, the preprocessing consists of constructing
an array~$T$ of $O(N)$ integers in linear time and the on-line phase
consists of directly accessing an element~$T[a]$ for an index $a$
given by a user.  As an example, using addition as the only primitive
operation (``primitive" here means ``of unit time cost'') one can
construct, for each input integer $N>0$, an array
$\mathtt{SquareRoot}[0..N]$ such that $\mathtt{SquareRoot}[x]\coloneqq
\lf \sqrt{x} \rf$, in $O(N)$ time.  This allows a user to obtain $\lf
\sqrt{a} \rf$ by direct access, for any integer $a\le N$ that she
needs.

\medskip
It seems natural and useful to ask whether the following generalization is possible.
For any fixed integer $d\ge 1$, is there an algorithm, using addition as the only primitive operation, such that, for any input integer $N$, the algorithm computes some tables in time $O(N)$ and then, by using these tables and after reading any integer $x<N^d$, the algorithm computes~$\lf \sqrt{x} \rf$ in \emph{constant time}?
In~\cite{GrandjeanJachiet22} we answered this question positively.
The solution is not trivial because only tables of size $O(N)$, and not of size $O(N^d)$, are authorized.

\medskip
In this paper which complements paper~\cite{GrandjeanJachiet22} but can be read independently, we study the class of operations $\mathtt{op}: \mathbb{N}^k\to\mathbb{N}$, of any arity $k\ge 1$, satisfying the following property:
for each fixed integer $d\ge 1$, there exists an algorithm $\mathcal{A}$
such that, for all input integer $N\ge 2$, 
\begin{itemize}
\item \emph{Preprocessing (off-line phase):} algorithm $\mathcal{A}$ computes some tables in $O(N)$ time;
\item \emph{Operation (on-line phase):} algorithm $\mathcal{A}$ reads $k$ integers $x_1,\ldots,x_k<N^d$, then computes $\mathtt{op}(x_1,\ldots,x_k)$ in \emph{constant time}.  
\end{itemize}
For a set (or list) of operations $\mathtt{OP}$, let us denote $\cpp(\OP)$ 
(i.e. \emph{constant time with preprocessing} and operations in $\OP$, called $\mathcal{M}[\OP]$ 
in~\cite{GrandjeanJachiet22}) the class of operations computed by such an algorithm 
$\mathcal{A}$ if the primitive operations used by $\mathcal{A}$, in its off-line and on-line phases, belong to~$\mathtt{OP}$.

\medskip \noindent
In~\cite{GrandjeanJachiet22}, we proved that many operations belong to the class $\cpp(+)$:
\begin{itemize}
\item the four usual operations, addition, subtraction, multiplication and Euclidean division, but also
the logarithm function $(x,y)\mapsto \lf \log_x y\rf$ and the $c$th root $x \mapsto \lf x^{1/c} \rf$, for each fixed integer $c\ge 2$;
\item concatenation, shift and bitwise Boolean operations 
$\xorop,\orop,\andop,\notop$, etc., and more generally, all the operations that can be computed in linear time on Turing machines, or even more generally, on one-dimensional cellular automata~\cite{GrandjeanRT12,Terrier12}.
\end{itemize}

\noindent
The results of this paper are twofold: 
\begin{enumerate}

\item Most of the paper is devoted to establishing that the complexity class $\cpp(+)$, abbreviated as $\cpp$, is robust and extensive: 
\begin{itemize}

\item It is closed in the following sense: as demonstrated in~\cite{GrandjeanJachiet22} and in Sections~\ref{sec:closure} and~\ref{sec:timesDiv}, it is not modified if not only addition, but also all operations still proven to belong to this class (subtraction, multiplication, division, logarithm, etc.) are allowed as primitive operations in the algorithms; 
formally, if~$+\in \mathtt{OP}$ and $\mathtt{op}\in \cpp(\OP)$ then 
$\cpp(\mathtt{OP}\cup\{\mathtt{op}\})= \cpp(\mathtt{OP})$; 
in particular, this implies $\cpp =$ \linebreak
$\cpp(+,-,\times,\mathtt{div},\mathtt{mod})$;
\item It is not modified if the preprocessing time is reduced to $O(N^{1/c})$, 
for any fixed $c>1$, if the primitive operations are not only $+$ but also 
$\mathtt{div}$ and $\mathtt{mod}$, or increased to $O(N^c)$, for any fixed $c>1$, see Section~\ref{sec: N->N^1/c}
(moreover, this result is maximal because allowing a non-polynomial time preprocessing strictly expands the $\cpp$ class, see Appendix~\ref{app:NonPol-CstPP});
\item The constant-time procedure can be reduced to an expression (a return instruction) involving only additions and direct memory accesses (Section~\ref{sect:expression});
\item The $\cpp$ class is closed under composition (if the intermediate results are polynomial) and it is also closed under inverse for fast growing functions in the following sense: 
if $f:\mathbb{N}\to\mathbb{N}$ is an exponentially growing function with its inverse function $f^{-1}$ defined as $f^{-1}(x)\coloneqq \min\{y\mid f(y)\ge x\}$, then we have the equivalence $f\in\cpp \iff f^{-1}\in\cpp$ 
(Section~\ref{sec:logf}). 

\end{itemize}

\item Furthermore, Section~\ref{sec:minimality} proves that the class $\cpp$ is minimal in the following two senses: 
\begin{itemize}
\item If the preprocessing time is reduced to $N^{o(1)}$, instead of
  $N^{\varepsilon}$ for any $\varepsilon\in]0,1]$, then the class no
    longer even contains the multiplication operation (Subsection~\ref{subsec:N^o(1)});
\item If addition is replaced, as the only primitive operation, by a set of unary operations (e.g., the successor and predecessor functions), then the class no longer contains the addition operation, even though we allow preprocessing to use any space and time (Subsection~\ref{subsec:unary}).
\end{itemize}

\end{enumerate}

A crucial tool to establish the robustness of the $\cpp$ class is the ``Fundamental Lemma'' of~\cite{GrandjeanJachiet22} which we recall in this paper with its proof (Lemma~\ref{lemma:fund}). 
This technical lemma combines the notion of faithful (or lock-step) simulation, a very precise concept designed by Yuri Gurevich~\cite{DexterDG97,Gurevich93}, with that of linear preprocessing.
Furthermore, this lemma also allows us to establish the robustness of other complexity classes such as 
$\textsc{LinTime}$~\cite{Grandjean96,GrandjeanSchwentick02,Schwentick97} and 
$\cdlin$~\cite{Durand20,Segoufin14} (Theorems~\ref{th:transitivity2} and~\ref{th:+=+...mod}) and even nonlinear time complexity classes (Corollary~\ref{cor:invTimeN1/c} and Appendix~\ref{app:nonlinear}).

\medskip \noindent
\emph{Comparison with literature:} 
Since we study classical arithmetic operations, addition, multiplication, division, square root, etc., the reader may wonder how our algorithms compare to those of other communities, particularly in computer algebra, see for example~\cite{vonzurGathenG2013}. 
The \emph{setting} and the \emph{goal} of our results are \emph{completely different}.

While the algorithms in the literature deal with ``long integers'' whose reference is the binary length of the operands,  
our algorithms deal with ``short integers'', i.e. integers \emph{polynomial in the reference integer} $N$, which is often the size of the input. 
Equivalently, our operands have binary length $O(\log N)$ and are represented in base~$N$ with O(1) digits.
This is a necessary condition for obtaining \emph{constant-time} algorithms.

The main difference is the following: 
while the goal of arithmetic algorithms in the classical setting~\cite{AhoHU74,CormenLRS09,Knuth69,SandersMDD19} is to \emph{minimize the computation time} with respect to the length (or value) of the operands, our sole goal is to obtain \emph{constant-time algorithms modulo some preprocessing}\footnote{
However, the constant-time algorithm (with preprocessing) that we present in this paper for division is exactly a copy of  the tricky algorithm presented in Knuth's book~\cite{Knuth69}, Section~4.3.1, and attributed to Pope and Stein~\cite{PopeS60}. 
(Additionally, to achieve constant time, we first convert the operands to a suitable base and complete the algorithm by pre-computing some tables.)}.

\section{Preliminaries}\label{sec:preliminaries}

Our central object in this paper is the class of operations computable in constant time with linear-time preprocessing.

 \medskip \noindent
{\bf Polynomial and non-polynomial operations.} It will be useful to distinguish arithmetic operations according to their growth.

\begin{definition}[polynomial operation]\label{def:opN,d_pol}
An operation $\op:\mathbb{N}^k\to \mathbb{N}$ of arity $k$ is called 
\emph{polynomial} if there exists an integer $q$ such that for all integers $N\ge 2$ and 
$x_1,\dots,x_k<N$, we have $\op(x_1,\dots,x_k)<N^q$.
For all integers $d\ge 1$ and $N\ge 2$, we note $\op_{N,d}$ the restriction of $\op$ to operands $x_1,\dots,x_k$ less than $N^d$. 
Of course, we have $\op_{N,d}(x_1,\dots,x_k)=\op(x_1,\dots,x_k)<N^{qd}$.
\end{definition}

\noindent
Many usual arithmetic operations $+,-,\times,\mathtt{div}$, etc. are polynomial, but not the exponential function 
$(x,y)\mapsto x^y$.
For a non-polynomial operation, we consider its \emph{polynomial restrictions}.
  
\begin{definition}[polynomial restriction of a non-polynomial operation]\label{def:opN,d_nonpol}
Let $\op$ be a non-polynomial operation of arity $k$.
For all integers $d\ge 1$ and $N\ge 2$, we note $\op_{N,d}$ the function 
from $[0,N^d[^k$ to $[0,N^{d}[\cup\{\mathtt{overflow}\}$ defined by
\begin{equation*}
\op_{N,d}(x_1,\dots,x_k) \coloneqq
\begin{cases}
\op(x_1,\dots,x_k) & \mathtt{if}\; \op(x_1,\dots,x_k)<N^{d}\\
\mathtt{overflow} & \mathtt{otherwise}.
\end{cases}
\end{equation*}
\end{definition}

\subsection{The RAM model}\label{subsec:RAMmodel}
The reference computer model for sequential algorithms is the random access machine (RAM), apparently invented in 1963 by Sherpherson and Sturgis~\cite{ShepherdsonS63} and studied by Cook and Reckhow~\cite{CookR73}. 
Here, we mainly adopt the terminology of the well-known books on algorithms and 
complexity~\cite{AhoHU74,Papadimitriou94}, see also~\cite{CormenLRS09,SandersMDD19}.

\medskip
A RAM with $\OP$ operations, also called RAM($\OP$), consists of a program operating on an infinite sequence of registers (or cells) $R[0],R[1],\dots$ called the \emph{memory} of the RAM, with an additional register~$N$, called the \emph{reference register}, which contains an integer~$N\ge 2$, called the \emph{reference integer}, never modified during a computation.
Each memory register contains a non-negative integer\footnote{From now on, we will simply write ``integer'' to mean ``non-negative integer''.} 
which must be less than~$cN$, for a fixed integer~$c\ge 2$, and is 0 at the initial instant.

\medskip
The \emph{program} is a sequence $I_0, \dots, I_{r}$ of numbered/labeled instructions, also called $R$-instructions, of the following forms with given intuitive meaning, where $i,j$ are explicit integers, $\ell_0, \ell_1$ are valid indexes of instructions, $0\le\ell_0,\ell_1\le r$, and 
$\op$ is an operation from $\OP$ of arity $k$:

\begin{table}[H]
  \centering
\begin{tabular}{|clr|}
  \hline
   \bf{Instruction} &  \bf{Meaning} & \\ 
  \hline
  $\mathtt{Cst} \; i \; j$ & $R[i] \gets j$ & \\
  $\mathtt{Move} \; i \; j$ & $R[i] \gets R[j]$  & \\ 
  $\mathtt{Store}\; i \; j$ & $R[R[i]] \gets R[j]$  & \\ 
  $\mathtt{Load}\; i \; j$  & $R[i] \gets R[R[j]]$  & \\ 
  $\mathtt{Jzero} \; i \; {\ell_0} \; {\ell_1}$  & $\mathtt{if}\; R[i]=0\;
  \mathtt{then}\;\mathtt{goto}\;\ell_0\;\mathtt{else}\;\mathtt{goto}\;\ell_1$
  & \\ 
  $\op$ & $R[0]\gets \op (R[0],\ldots, R[k-1])$  & \\
  $\mathtt{GetN}\; i$ & $R[i]\gets N$  & \\ 
  $\mathtt{Input}\; i $ & Read $R[i]$  & \\ 
  $\mathtt{Output} \; i$ & Output $R[i]$ & \\ 
  $\mathtt{Halt}$ & & \\ 
  \hline
\end{tabular}
\caption{Set of $R$-instructions for a RAM($\OP$)}
\label{table:instRAM2}
\end{table}

\noindent
The $I_r$ instruction is $\mathtt{Halt}$.
At the initial instant of a computation of the RAM, the memory registers are at 0 and the reference register contains $N$, the reference integer, never modified, of the computation.
Then, the computation proceeds as one would expect.
The RAM first executes the $I_0$ instruction. 
After the RAM executes an $I_{i}$ instruction, it executes $I_{i+1}$, except if $I_{i}$ is a (branch) 
$\mathtt{Jzero}$ instruction or the $\mathtt{Halt}$ instruction.

 \medskip \noindent
{\bf Remarks about our RAM model.}
Our input and output instructions can be executed \emph{at any time}, a non-standard convention.
For example, the instruction $\mathtt{Input}\; i $ takes the current input value and stores it in the register $R[i]$.
This flexibility allows the implementation of many interesting algorithms, including those using preprocessing, such as enumeration algorithms.

For correct execution, each instruction $\mathtt{Input}\; i $ must read an integer less than $cN$ and the result of each 
$\op$ instruction must also be less than~$cN$.
This does not prevent a computation from handling much larger integers, for example integers less than $N^d$ (resp. less than $N^N$ in the general case), because  they can be represented by their $d$ (resp.~$N$) digits in base $N$ contained in $d$ (resp.~$N$) consecutive registers.

In the literature, there are other conventions regarding the contents of registers.
Some authors, e.g. Andersson and al.~\cite{AnderssonMRT96, AnderssonMT99,Thorup00}, require that each register contains an integer/binary word of length~$w$, that is an integer less than $2^w$: 
the role of our reference integer $N$ is comparable\footnote{More precisely, we allow integers less than $cN$ as register contents, for a fixed integer $c\ge 2$. 
This means \\ 
$2^{w-1}< cN \le 2^w$ or equivalently $w=\lc\log_2(cN)\rc$.} 
to that of $2^w$. 

Many other papers, e.g.~\cite{BerkholzKS17, CarmeliK18}, allow each register to hold an integer of length 
$O(\log N)$ for an input of size $N$, as Cormen and al. recommend in Subsection 2.2 of their book~\cite{CormenLRS09}, that is to say a polynomial integer $N^{O(1)}$: 
this is equivalent to our RAM model in the case where we also require the addresses used to be $O(N)$ as we prove in Appendix~\ref{app:O(log N)} of this paper; 
otherwise, this means that the RAM can use polynomial memory like in~\cite{BerkholzKS17, CarmeliK18}, see Section 2.5 in~\cite{GrandjeanJachiet22} for a comparison of the two models.

 \medskip \noindent
{\bf Why this RAM model?} Our RAM model is simple and flexible. 
Its main constraint is to use only integers in $O(N)$ as contents and therefore as register addresses, which forces us to use only linear memory.
This is a realistic condition and does not prevent the manipulation of much greater integers.
This also allows us to read and represent any input (graph, hypergraph, formula, etc.). 
As an example, an input graph with $N$ vertices $1,\dots,N$ and $m$ edges can be read by $2m$ $\mathtt{Input}$ instructions.

Our RAM model allows us to represent many types of sequential algorithms: 
the classical $\mathtt{Input}\mapsto\mathtt{Output}$ algorithms, but also algorithms decomposed into phases punctuated by input/output instructions, such as compilation algorithms~\cite{DarwicheM02, Marquis15}, enumeration algorithms~\cite{DurandG07, Durand20, Segoufin14, Uno15}, etc.

\medskip
To program our RAMs, it will be convenient to also define an extended version of the RAM using variables and arrays, which we call a \emph{multi-memory} RAM.

 \medskip \noindent
{\bf Multi-memory RAM.}
 A \emph{multi-memory} RAM with $\OP$ operations uses a fixed number of one-dimensional \emph{arrays of integers} $A_j[\;]$, 
a fixed number of \emph{integer variables} $v_i$,
and the \emph{reference integer}~$N$.  
Its \emph{program} uses \emph{$\OP$-expressions}
defined recursively as follows.

\begin{definition}[$\OP$-\emph{expressions}]
\begin{itemize}
\item Any fixed integer $i$, the reference integer $N$, and any integer
  variable $v_j$ are $\OP$-expressions.
\item If $\alpha$ is an $\OP$-expression, then an array element $A_j[\alpha]$ is an $\OP$-expression.
\item If $\alpha_1, \dots, \alpha_k$ are $\OP$-expressions and $\op$ is a $k$-ary operation in $\OP$, then $\op(\alpha_1, \dots, \alpha_k)$ is an $\OP$-expression.
\end{itemize}
\end{definition}
\noindent
\emph{Example:} $A_1[A_3[2]+N]\times A_2[v_1]$ is a $\{+,\times\}$-expression.

\medskip
The \emph{program} of a multi-memory RAM with $\OP$ operations is
a sequence $I_0, \dots, I_{r} $ of numbered instructions, called
\emph{array-instructions}, of the following forms with intuitive meaning, where 
$\alpha$,~$\beta$ are $\OP$-expressions and $\ell_0, \ell_1$ are valid indexes of instructions.

\begin{table}[H]
  \centering
\begin{tabular}{|lr|}
  \hline
   \bf{Instruction} &  \bf{Meaning} \\
  \hline
  $A_j[\alpha] \gets \beta$  & \\
  $v_j \gets \alpha$ & \\
  $\mathtt{Read}\; A_j[\alpha]$  & \\
  $\mathtt{Read}\; v_j $ & \\
  $\mathtt{Output} \; \alpha$ & \\
  $\mathtt{Jzero} \; \alpha \; \ell_0 \; \ell_1$  & $\mathtt{if}\; \alpha=0\;\mathtt{then}\;\mathtt{goto}\;\ell_0\;\mathtt{else}\;\mathtt{goto}\;\ell_1$ \\
  $\mathtt{Halt}$ & \\
  \hline
\end{tabular}
\caption{Set of array-instructions for a RAM($\OP$)}
\label{table:instRAM3}
\end{table}

We will use a strong notion of simulation, called \emph{lock-step simulation}, introduced by Gurevich in his paper~\cite{Gurevich93} initiating \emph{Abstract State Machines} (ASM, initially called ``Evolving Algebras'') and defined in detail in~\cite{DexterDG97}, page~280. 
What follows is essentially the same definition. 
For the sake of notational simplicity, we assimilate a (finite) computation $\mathcal{C}$ to the list of instructions $(\mathcal{I}_0,\mathcal{I}_1,\dots,\mathcal{I}_t = \mathtt{Halt})$ that it executes, but we consider that a computation $\mathcal{C}$ also includes its input list and its output list even if our notation does not mention them.

\begin{definition}
A machine/program $P'$ \emph{simulates a machine/program $P$ in lock-step with lag 
factor}~$\lambda$ if, for every computation $\mathcal{C}=(\mathcal{I}_0,\mathcal{I}_1,\dots,\mathcal{I}_t = \mathtt{Halt})$ of $P$, there exists a function \linebreak
$f:[0,t]\to\mathbb{N}$ such that $f(0)=0$ and $1\le f(i)-f(i-1)\le \lambda$, for each $i\in[1,t]$, and a (unique) computation 
$\mathcal{C}'=(\mathcal{I'}_0,\mathcal{I'}_1,\dots,\mathcal{I}'_{f(t)}=\mathtt{Halt})$ of~$P'$,
called the computation simulating $\mathcal{C}$, which satisfies conditions 1-3: 
\begin{enumerate}
\item for every $i\in[0,t]$, if $\mathcal{I}_{i}$ is an input (resp. output) instruction of 
$\mathcal{C}$, then $\mathcal{I}'_{f(i)}$ is an input \linebreak
(resp. output) instruction of $\mathcal{C}'$;
\item the only input/output instructions of $\mathcal{C}'$ are given by item~1;
\item the sequence of inputs (resp. outputs) is the same for $\mathcal{C}'$ as for $\mathcal{C}$.
\end{enumerate}
Then we also say that $P'$ \emph{simulates $P$ faithfully}.
\end{definition}

This notion is used in the following lemma:

\begin{lemma}\label{fact:array->R}
Let $\OP$ be a set of operations with $+\in\OP$. 
Then,
\begin{enumerate}
\item each array-instruction $\mathcal{I}$ of a \emph{RAM($\OP$)} can be simulated by a fixed finite sequence $\mathcal{S}(\mathcal{I})$ of \linebreak
$R$-instructions of a \emph{RAM($\OP$)};
\item if $\mathcal{I}$ is an input (resp. output) instruction, then $\mathcal{S}(\mathcal{I})$ contains exactly one input (resp. output) instruction; if $\mathcal{I}$ is not an input/output instruction, then $\mathcal{S}(\mathcal{I})$ contains no input/output instruction;
\item therefore, a multi-memory \emph{RAM($\OP$)} can be simulated faithfully by a \emph{RAM($\OP$)} with \linebreak
$R$-instructions.
\end{enumerate}
\end{lemma}

\begin{proof}[Sketch of proof] We only justify item~1 (item~2 is obvious and item~3 is a straightforward consequence of items~1 and~2).
First, we add new variables for storing the values of the subexpressions of the 
$\OP$-expressions. This allows us to simulate each array-instruction by a fixed number (equal to the number of its subexpressions) of ``atomic'' array-instructions
of the following forms: 
\begin{center}
$v_i\gets j$ ; $v_i \gets N$ ; $v_i\gets \op(v_{j_1},\dots,v_{j_k})$ ;
$v_i\gets A_j[v_h]$ ; $A_i[v_j] \gets v_h$ ; \\
$\mathtt{Jzero}\; v_i\; \ell_0\; \ell_1$ ;  
$\mathtt{Read}\; v_i$ ; $\mathtt{Output}\; v_i$ ; $\mathtt{Halt}$.
\end{center}
Then, we simulate each variable~$v_i$, where $i\ge 1$, by the register $R[(s+1)i]$, 
where $s$ is the number of arrays $A_1,\dots,A_s$ and $(s+1)i$ is written for $i+\cdots+i$ ($s+1$ times), which is greater than~$s$, 
and we simulate each array cell $A_j[i]$ by the register $R[(s+1)(i+1)+j]$, whose index is not less than $s+2$ and is congruent to $j$ modulo $s+1$.
Note that the registers $R[0],\dots,R[s]$ are available and can therefore be used as buffers. 
The reader can now easily verify that using buffers, each atomic array-instruction can be simulated by a fixed finite sequence of $R$-instructions, including addition instructions, meaning $R[0]\gets R[0]+R[1]$ (see Table~\ref{table:instRAM2}), to implement the above-mentioned sums $i+\cdots+i$, etc.
\end{proof}

\smallskip \noindent
{\bf Our algorithmic conventions.} 
From now on, we will place ourselves in an even broader framework corresponding to the usual presentation of algorithms, which we also call RAMs or programs interchangeably.
We freely use ``if...then...else...'' statements, ``while" and ``for'' loops, and procedures with ``return'' instructions, also called functions.
Additionally, we allow elaborate test conditions in conditional and while instructions. 
For example, using equalities $\alpha=\beta$ as test conditions, for $\alpha,\beta\in O(N)$, instead of the usual zero test, is justified by the fact that the (extended) array-instruction
$\mathtt{Jequal}\;\alpha\;\beta\;\ell_0\;\ell_1$, 
meaning $\mathtt{if}\; \alpha=\beta\;\mathtt{then}\;\mathtt{goto}\;\ell_0\;\mathtt{else}\;\mathtt{goto}\;\ell_1$,
is equivalent to the following sequence of array-instructions (see Table~\ref{table:instRAM3}) using an $\mathtt{Equal}$ array:
\begin{center}
$\mathtt{Equal}[\alpha]\gets 1$ ;
$\mathtt{Equal}[\beta]\gets 0$ ;
$\mathtt{Jzero}\;\mathtt{Equal}[\alpha]\;\ell_0\;\ell_1$.
\end{center}
Also, we allow a parameter or variable $x$ to represent any integer less than $N^d$, for a fixed integer~$d$, by its $d$ digits $x_i$ in base~$N$ contained in $d$ consecutive RAM registers. 
It will be convenient to confuse such a variable $x$ with the list of variables $(x_{d-1},\dots,x_0)$ for its $d$ digits.
We can easily verify that Lemma~\ref{fact:array->R} is still valid for the programming language thus extended.

\subsection{Complexity measures and complexity classes}\label{subsec:ComplexClasses}
We adopt the usual complexity measures: the \emph{time} of a program or procedure is the number of instructions it executes; its \emph{space} is the number of distinct registers, or array cells (and variables) that it modifies. 
A computation or part of a computation -- for example, preprocessing or the part between two consecutive input/output instructions -- is said to use \emph{linear time} (resp. \emph{constant time}) if it executes $O(N)$ (resp. $O(1)$) instructions.

We note that if a program $P'$ simulates another program $P$ in
lock-step with lag factor~$\lambda$, then by
Lemma~\ref{fact:array->R}, for every computation $\mathcal{C}$ of $P$,
if the time between any two consecutive input/output instructions
$\mathcal{I}_i,\mathcal{I}_{i'}$, $i<i'$, of $\mathcal{C}$ is $\tau$,
then the time between the corresponding input/output instructions
$\mathcal{I}'_{f(i)},\mathcal{I}'_{f(i')}$ in the computation
$\mathcal{C}'$ (of $P'$) simulating $\mathcal{C}$ is at most
$\lambda\tau$.  The same is true for the time before the first
input/output instruction and for the time after the last input/output
instruction.

The notion of faithful simulation, which preserves time up to a constant factor, will make it possible to analyze the tight complexity of problems involving successive computation phases separated by input/output instructions, mainly enumeration problems~\cite{AmarilliBJM17, BaganDG07, Durand20, DurandG07, Segoufin14, Strozecki19, Uno15}, but also many others studied recently, e.g.~\cite{BaganDGO08, BerkholzKS17, BerkholzKS18, BringmannCM22, BringmannCM22arxiv, CarmeliTGKR23, EldarCK24}. 

\medskip
Unless explicitly stated otherwise, a RAM is a RAM($+$), that is, it uses addition as its only primitive operation.
Let us define the main complexity classes that will interest us.

\begin{itemize}
\item {\bf Classes of $\mathtt{Input}\mapsto\mathtt{Output}$ problems:}
Among these classes of problems, we are mainly interested in this paper by the linear-time class defined in~\cite{Grandjean96,GrandjeanSchwentick02,Schwentick97}.
We denote by $\textsc{Time}(N)$ or $\textsc{LinTime}$ the class of problems\footnote{In Section~\ref{subsec:reducN^epsilon} (resp. Appendix~\ref{app:nonlinear}), we define more generally, for any function $T:\mathbb{N}\to \mathbb{N}$, $T(N)=N^{1/c}$, for a fixed integer $c\ge 1$ (resp. $T(N)\ge N$), the class $\textsc{Time}(T(N))$, and Corollary~\ref{cor:invTimeN1/c} (resp. Theorem~\ref{th:robustT}) establishes its robustness.}
  $\Pi: X \mapsto \Pi(X)$, with input $X=(x_1,\dots,x_N)$ of size $N$,
  and $0\le x_i<cN$ (for a fixed integer $c\ge 2$), which are computed by a RAM in time $O(N)$. 
  
\item {\bf Classes of enumeration problems:} Among many 
variants~\cite{AmarilliBJM17, Durand20, DurandG07, JohnsonPY88, Segoufin14, Strozecki19, Uno15}, let us recall here what we consider to be the class of minimal complexity: the class $\cdlin$. 
A problem $\Pi: X \mapsto \Pi(X)$, with input $X=(x_1,\dots,x_N)$ of size $N$, for 
$0\le x_i<cN$, and as output a finite set $\Pi(X)$ of elements (called solutions), each of \emph{constant size}, belongs to $\cdlin$ if there is a RAM which, after reading an input $X$ of size $N$, 
\begin{enumerate}
\item \emph{Preprocessing:} computes some tables in time $O(N)$, then,
\item \emph{Enumeration:} lists (= produces) without repetition all the solutions of $\Pi(X)$ with a \emph{constant delay} (= time between two successive solutions) and a \emph{constant space}, and stops immediately after the last solution.
\end{enumerate}
\item {\bf Classes of constant time operations:} Again, we focus on
  the complexity classes that we believe are minimal and robust, as
  established in the following sections.  For $\OP$ a set of arbitrary
  arity operations on $\mathbb{N}$, we define $\cpp(\OP)$ as the class
  of operations $\op: \mathbb{N}^k \to \mathbb{N}$ of arbitrary arity
  $k$ such that, for any fixed integer $d\ge 1$, there exists a
  RAM($\OP$) such that, for each reference integer $N\ge 2$,
\begin{enumerate}
\item \emph{Preprocessing:} 
the RAM computes some tables in time $O(N)$, then,
\item \emph{Operation:} 
the RAM reads $k$ integers $x_1,\dots,x_k<N^d$ and computes the value 
$\op_{N,d}(x_1,\dots,x_k)$ in \emph{constant time}.
\end{enumerate}
We say that $\cpp(\OP)$ is the class of operations computable by a RAM($\OP$) \emph{in constant time with linear initialization (= linear preprocessing)}.
\end{itemize}

\begin{remark}
If $\op$ is a \emph{polynomial} operation, then by Definition~\ref{def:opN,d_pol}, $\op_{N,d}(x_1,\dots,x_k)$ is none other than $\op(x_1,\dots,x_k)$.
By Definition~\ref{def:opN,d_nonpol}, it is also the case if $\op$ is non-polynomial and $\op(x_1,\dots,x_k)<N^d$.
Note that the condition that the value of $\op_{N,d}(x_1,\dots,x_k)$ is polynomial in $N$, i.e. has a constant size, is a necessary condition for it to be computable in constant time.
\end{remark}

One can easily verify that the condition for an operation $\op$ to be in $\cpp(\OP)$ is the same whether it is polynomial or non-polynomial, 
even if the definition of its restriction $\op_{N,d}$ is slightly different depending on the case (Definitions~\ref{def:opN,d_pol} and~\ref{def:opN,d_nonpol}).
Here is a unifying definition: 

\begin{definition}
An operation $\op$ of arity $k$ is in $\cpp(\OP)$ if, for all integers $d,d'\ge 1$ and 
$N\ge 2$, the function $\op_{N,d,d'}: [0,N^d[^k \;\to [0,N^{d'}[\cup\{\mathtt{overflow}\}$ defined 
for $\overline{x}\coloneqq (x_1,\dots,x_k)\in [0,N^d[^k$, by 
$\op_{N,d,d'}(\overline{x}) \coloneqq \op(\overline{x})$ if $\op(\overline{x})<N^{d'}$
and $\op_{N,d,d'}(\overline{x}) \coloneqq \mathtt{overflow}$ otherwise, 
is computable in constant time by a \emph{RAM($\OP$)} with linear preprocessing. 
\end{definition}
Those definitions are equivalent because $d$ and $d'$ are universally quantified in them.

\begin{convention}
Each operand $x_i<N^d$ of $\op_{N,d,d'}$ is represented by the $d$ digits of its $N$-notation in $d$ consecutive registers. 
The same is true for the result $\op_{N,d,d'}(\overline{x})<N^{d'}$ in $d'$ registers.
\end{convention}

\begin{example}\label{ex:pred}
The predecessor operation $\mathtt{pred}: \mathbb{N}\to \mathbb{N}$, defined as
$\mathtt{pred}(x)\coloneqq \max(0,x-1)$, \linebreak
belongs to $\cpp(+)$.
This is justified by the following linear-time preprocessing procedure $\textsc{LinPP.pred}()$, 
which computes the array $\textsc{pred}[1..N]$, of obvious meaning,
and the constant-time procedure $\textsc{CstP.pred}_d(x_{d-1},\dots,x_0)$, which computes
the function $x\mapsto \mathtt{pred}(x)$ for $x<N^d$. 
(Recall that $x$ is represented by the list of its $d$ digits $(x_{d-1},\dots,x_0)$ in base $N$.)
This procedure is given below for $d=2$ but it can be easily adapted to any fixed~$d\ge 1$.
\end{example}

 \begin{minipage}{0.05\textwidth}
        ~
      \end{minipage}
      \begin{minipage}{0.85\textwidth}
        \begin{algorithm}[H]\label{algo:pred}
          \caption{Pre-computation of \fun{pred} and computation of \fun{CstP.pred}$_2$}
      \begin{minipage}[t]{0.50\textwidth}
      \vspace{0.7em}
      \begin{algorithmic}[1]
        \Procedure{LinPP.pred}{\mbox{}}
  \State $\var{x} \gets 0$
  \While{$\var{x} \neq \var{N}$}
    \State $\arr{PRED}[\var{x}+1] \gets \var{x}$
    \State $\var{x} \gets \var{x}+1$
  \EndWhile
\EndProcedure
      \end{algorithmic}
      \end{minipage}
      \begin{minipage}[t]{0.55\textwidth}
      \begin{algorithmic}[1]
\Procedure{CstP.pred}{$\var{x}_{1},\var{x}_{0}$}
\If{$\var{x}_{0}=0$}
  \If{$\var{x}_{1}=0$}
  \State \Return $0$
  \Else
     \State \Return $(\arr{PRED}[\var{x}_1],\arr{PRED}[\var{N}])$
  \EndIf
  \Else
     \State \Return $(\var{x}_1,\arr{PRED}[\var{x}_{0}])$
\EndIf
\EndProcedure
      \end{algorithmic}
      \end{minipage}
    \end{algorithm}
      \end{minipage}

\medskip 
\begin{notation}
From now on, we will write $x-1$ instead of $\textsc{pred}[x]$ or $\textsc{CstP.pred}_d(x)$.
\end{notation}

\noindent 
The operation $x\mapsto \mathtt{pred}(x)$, for $x<N^d$, can now be considered as a primitive operation.

\section{Robustness of complexity classes through linear preprocessing}\label{sec:closure}

For the sake of completeness, we recall with its proof the following general lemma from~\cite{GrandjeanJachiet22} which establishes a ``transitivity'' property.

\begin{lemma}[Fundamental Lemma\label{lemma:fund}~\cite{GrandjeanJachiet22}]
Let $\OP$ be a set of operations which includes $+$.   
Let~$P$ be a program of \emph{RAM($\OP\cup\{\op\}$)}, where $\op$ is an operation in $\cpp(\OP)$.

Then $P$ can be faithfully simulated by a program $P'$ of \emph{RAM($\OP$)} after a linear-time initialization, also using only the operations of $\OP$. 

\end{lemma}

\begin{proof}
Let $P$ be a program of RAM($\OP\cup\{\op\}$).  
Let $\linppop()$ and $\cstpop(\mathbf{x})$ be respectively the initialization procedure and the computation procedure of $\op$.

\medskip \noindent
\emph{Intuition:}
The idea of the proof is to create a program $P'$ that starts by calling 
the initialization procedure $\linppop()$, 
and then calls the program~$P$ where each occurrence of a term 
$\op(\mathbf{x})$ in~$P$ is replaced by a call to $\cstpop(\mathbf{x})$. 

\medskip \noindent
This is the essence of our proof but we also need to take care of the following  problem: 
the definition of the $\cpp(\OP)$ class \emph{does not guarantee} that it is possible to make two (or more) calls to $\cstpop$. 
For instance, we might imagine that $\cstpop$ is destructive and prevents us from executing a second call to 
$\cstpop$ or that the answer might be wrong\footnote{This problem could be avoided by imposing, in the definition of the class $\cpp(\OP)$ given in Subsection 2.2, that performing an operation of this class does not modify the contents of memory, i.e., restores its initial contents.
However, adding this technical condition to the definition may seem artificial (even though it is satisfied by all our examples). 
We prefer our definition because it is simpler and seems more general. 
Fortunately, a consequence of the proof we now present is that the two definitions are equivalent.}.

\medskip
To overcome this problem, we can modify the procedure $\cstpop$ to ensure
that several successive calls are possible.  For that, we use a
modified procedure $\cstpopprime$ where we store the value of each
array cell or integer variable that is overwritten and at the end of
the call to $\cstpopprime(\mathbf{x})$ we restore the memory back to
its state before the call.  This is possible by maintaining a complete
log (a history) of the successive modifications (assignments/writes)
of the memory.  At the end of the call to $\cstpopprime(\mathbf{x})$,
before its return instruction, the initial state of the memory is
restored, step by step, by replaying the log in reverse chronological
order.

\medskip
For that, we introduce three new arrays called
 $\arr{Old}[\;]$, $\arr{Index}[\;]$ and $\arr{Dest}[\;]$, as well as two new variables called 
  $\var{NbWrite}$ and $\var{ReturnValue}$. 
  The modified procedure $\cstpopprime$ starts by setting $\mathtt{NbWrite}=0$,
  and then, for each assignment 
  $\var{v}_i\gets u$ or $\arr{A}_i[t]\gets u$ of $\cstpop$, 
  we write into  $\arr{Old}[\var{NbWrite}]$
  the current value $v$ of $\var{v}_i$ or $\arr{A}_i[t]$ that will be overwritten;
  afterwards, we write into $\mathtt{\arr{Dest}}[\var{NbWrite}]$ the integer $i$
  that indicates\footnote{It is convenient to assume that the set of integers $i\leq q$ which identify the individual variables $\var{v}_i$ is disjoint from the set of integers $i>q$ that number the arrays $\arr{A}_i$. Note that the test $i\leq q$ does not actually use the order symbol $\leq$ since it is equivalent to the disjunction of $q$ equalities $\bigvee_{j=1}^q i = j$.}
  which variable $\var{v}_i$ or array $\arr{A}_i$ is being written into, and, if the write was a write into an array $\arr{A}_i$ (i.e.\ if the integer $i$ is greater than~$q$), we also note the value of the index $t$ of the cell $\arr{A}_i[t]$ overwritten into $\arr{Index}[\var{NbWrite}]$; 
 finally, we write the new value $u$ into $\var{v}_i$ or $\arr{A}_i[t ]$ and we increment 
 $\var{NbWrite}$ by one.
 At the end of the procedure, instead of returning directly the expected value, we store it in the
 $\var{ReturnValue}$ variable. 
 
 \medskip
 Then, we restore the memory by replaying the log (write history) in reverse. 
 To do this, we decrease the counter variable $\var{NbWrite}$ by one as long as it is strictly positive and, each time, we restore the value $v=\arr{Old}[\var{NbWrite}]$ into the array cell (resp. variable) indicated by $i=\arr{Dest}[\var{NbWrite}]$, at index 
  $\var{ind}=\arr{Index}[\var{NbWrite}]$ for an array cell. 
  This means that we execute the assignment $\arr{A}_i[\var{ind}]\gets v$ if $i>q$ 
  (resp.\ $\var{v}_i\gets v$ if $i\leq q$).
  Thus, all the memory (the variables $\var{v}_1,\dots,\var{v}_q$ and 
  arrays $\arr{A}_{q+1},\dots,\arr{A}_{r}$)
  returns to the state in which it was before the call to $\cstpop(\mathbf{x})$, now transformed into 
  $\cstpopprime(\mathbf{x})$.
  Once done, we return the value stored in the $\var{ReturnValue}$ variable.

\medskip
Formally, the code of the $\cstpopprime$ 
procedure is the concatenation of the following calls using the procedures of Algorithm~2 below:
\begin{itemize}
\item $\fun{InitCounter}(\mbox{})$, which initializes $\var{NbWrite}$;
\item $\fun{CstPwithLOG}(\mathbf{x})$, which is $\cstpop(\mathbf{x})$ where each instruction of the form $\var{v}_{\var{i}} \gets u$ (resp.\ $\arr{A}_{\var{i}}[t] \gets u$, $\mathtt{Return}\;u$)
is replaced by the call $\fun{AssignVar}(\var{i},u)$ (resp.\ $\fun{AssignCell}(\var{i},t, u)$, 
$\fun{StoreReturn}(u)$);
\item $\fun{Restore}(\mbox{})$, which restores memory in the reverse chronological order, using the arrays $\arr{Dest}[\;]$, $\arr{Index}[\;]$ and $\arr{Old}[\;]$, 
and finally returns $\var{ReturnValue}$.
\end{itemize}
 \begin{minipage}{\almosttextwidth} 
  \begin{algorithm}[H] 
  \caption{Procedures for constructing $\cstpopprime$ from the code of $\cstpop$ whose list of   variables
is $\var{v}_1,\dots,\var{v}_q$ and the list of arrays is $\arr{A}_{q+1},\dots,\arr{A}_{r}$
  }
      \begin{minipage}[t]{0.53\textwidth} 
      
      \vspace{0.5em}
      
       \begin{algorithmic}[1]
       \Procedure{InitCounter}{$\mbox{}$}
       \State $\var{NbWrite}\gets 0$
       \EndProcedure
       \end{algorithmic}
       
      \vspace{1em}
      
      \begin{algorithmic}[1] 
     \Procedure{AssignVar}{\var{i},\var{u}} 
     \\ \Comment{called to replace 
     $\var{v}_{\var{i}}\gets \var{u}$} 
      \State $\arr{Old}[\var{NbWrite}]\gets \var{v}_{\var{i}}$
      \State $\arr{Dest}[\var{NbWrite}]\gets \var{i}$
      \State $\var{v}_{\var{i}} \gets \var{u}$
      \State $\var{NbWrite}\gets \var{NbWrite}+1$
    \EndProcedure
      \end{algorithmic} 
  
      \vspace{1em}
      
      \begin{algorithmic}[1] 
     \Procedure{AssignCell}{\var{i},\var{t},\var{u}} 
     \\ \Comment{called to replace $\arr{A}_{\var{i}}[\var{t}]\gets  \var{u}$} 
      \State $\arr{Old}[\var{NbWrite}]\gets \arr{A}_{\var{i}}[\var{t}]$
      \State $\arr{Dest}[\var{NbWrite}]\gets \var{i}$
      \State $\arr{Index}[\var{NbWrite}]\gets \var{t}$
      \State $\arr{A}_{\var{i}}[\var{t}] \gets \var{u}$
      \State $\var{NbWrite}\gets \var{NbWrite}+1$
      \EndProcedure
      \end{algorithmic} 
      
      \vspace{1em}
      
      \end{minipage} 
\begin{minipage}[t]{0.45\textwidth} 
 
 \vspace{1.5em}
 
 \begin{algorithmic}[1] 
\Procedure{StoreReturn}{$\var{u}$} 
\\ \Comment{called to replace $\mathtt{Return}\;\var{u}$}
\State $\var{ReturnValue}\gets \var{u}$
\EndProcedure
\end{algorithmic} 

      \vspace{1em}
      
\begin{algorithmic}[1] 
\Procedure{Restore}{$\mbox{}$}
\While {$\var{NbWrite}>0$}
  \State $\var{NbWrite}\gets \var{NbWrite}-1$
  \State $\var{i}\gets \arr{Dest}[\var{NbWrite}]$
  \If{$\var{i} > q$}
     \State $\var{ind}\gets \arr{Index}[\var{NbWrite}]$
     \State $\arr{A}_{\var{i}}[\var{ind}] \gets \arr{Old}[\var{NbWrite}]$
  \Else
      \State $\var{v}_{\var{i}} \gets \arr{Old}[\var{NbWrite}]$
  \EndIf
\EndWhile 
\\ \Comment{the initial state of the memory}
\\ \Comment{is now restored}
    \State \Return $\var{ReturnValue}$
\EndProcedure
\end{algorithmic} 

\end{minipage} 
      
\end{algorithm} 
 \end{minipage} 

\bigskip \noindent 
Note that the time overhead
(resp. the size of arrays $\arr{Old}[\;]$, $\arr{Index}[\;]$ and $\arr{Dest}[\;]$)
of the $\cstpopprime$ procedure
is proportional to the number of assignments performed by the original
constant-time procedure $\cstpop$.
Therefore, the $\cstpopprime$ procedure is also constant-time.
\end{proof}

The following result from~\cite{GrandjeanJachiet22} is a direct consequence of Lemma~\ref{lemma:fund}.

\begin{theorem}\label{th:transitivity_base}
Let $\OP$ be a set of operations including $+$ and 
let $\op$ be an operation in $\cpp(\OP)$. 

Then, each operation in $\cpp(\OP\cup\{\op\})$ is also in $\cpp(\OP)$.

This implies $\cpp(\OP\cup\{\op\}) = \cpp(\OP)$.
\end{theorem}

More generally, we deduce from Lemma~\ref{lemma:fund} similar results for all complexity classes defined by RAMs and which are closed by the addition of linear-time preprocessings.
This is the case for the classes $\textsc{LinTime}(\OP)$ and $\cdlin(\OP)$ which are 
the $\textsc{LinTime}$ and $\cdlin$ classes defined for RAM($\OP$).

\begin{theorem}\label{th:transitivity2}
Let $\OP$ be a set of operations including $+$ and 
let $\op$ be an operation in $\cpp(\OP)$. 
Then, the classes $\textsc{LinTime}$ and $\cdlin$ are invariant whether the primitive RAM operations belong to~$\OP$ or to $\OP\cup\{\op\}$. 
Formally, $\textsc{LinTime}(\OP\cup\{\op\}) = \textsc{LinTime}(\OP)$ and $\cdlin(\OP\cup\{\op\}) = \cdlin(\OP)$.
\end{theorem}

It follows from Theorem~\ref{th:transitivity_base} that the $\cpp(\OP)$ class is closed under composition.
\begin{theorem}\label{th:closedCompose}
Let $\OP$ be a set of operations including $+$. 
Let $g$ be an $m$-ary operation in $\cpp(\OP)$ and let $h_1,\dots,h_m$ 
be $m$ polynomial~\footnote{To define a composition of operations in such a way that it always makes sense in our framework, it seems necessary to assume that the intermediate results are polynomial in $N$.} 
$k$-ary operations in $\cpp(\OP)$. 
Then the composed $k$-ary operation $f$  
defined as $f(x_1,\dots,x_k) \coloneqq g(h_1(x_1,\dots,x_k),\dots,h_m(x_1,\dots,x_k))$
also belongs to $\cpp(\OP)$.
\end{theorem}

\begin{proof} An immediate consequence of Theorem~\ref{th:transitivity_base} is the following generalization: 
if $\op_1,\dots,\op_q$ are operations in $\cpp(\OP)$ with $+\in\OP$, then 
$\cpp(\OP\cup\{\op_1,\dots,\op_q\}) = \cpp(\OP)$. 
In particular, since $g,h_1,\dots,h_m \in \cpp(\OP)$, we obtain
$\cpp(\OP\cup\{g,h_1,\dots,h_m\}) = \cpp(\OP)$.
The theorem follows from the fact that $f$ obviously belongs to $\cpp(\OP\cup\{g,h_1,\dots,h_m\})$.
\end{proof}

\noindent
\emph{How to prove that an arithmetic operation is in $\cpp$?}
The basic principle is simple: 
use a ``mini-table'' of the operation, e.g. the multiplication table $A_{\times}$ in the next Section~\ref{sec:timesDiv}. 

If the operation is definable by simple arithmetic induction, which is the case for the usual arithmetic operations, addition, subtraction, multiplication, division, etc., and if the base $B$ is chosen judiciously to represent the operands, often $B\coloneqq \lc N^{1/2}\rc$ (sometimes $B\coloneqq \lc N^{1/c}\rc$ for a fixed integer $c>2$), we can obtain a mini-table of size $O(N)$ and compute it recursively in linear time; 
for example, the multiplication table $A_{\times}[0..B-1][0..B-1]$ of size $B^2=O(N)$, which is defined by 
$A_{\times}[x][y]\coloneqq x\times y$ for $x,y<B$.

There is a technical subtlety. 
In reality, our mini-table will be the one-dimensional table encoding the two-dimensional array above according to the usual ``sequential representation of multidimensional arrays'', see for example Section 2.4 of \cite{HorowitzS84}.
In particular, the multiplication table that we will actually compute and use will be the one-dimensional array 
$A_{\times}[0..B^2-1]$ defined by\footnote{Note that we are using the product $B \times x$ to define multiplication here! 
This seems like a vicious circle but it can be avoided: 
replace $B \times x$ by the element $\textsc{Mult}B[x]$ of an array $\textsc{Mult}B[0..B-1]$ defined by 
$\textsc{Mult}B[x]\coloneqq B\times x$. 
This array can be pre-computed inductively \emph{using only addition}.}
$A_{\times}[B \times x+y]\coloneqq x\times y$ for $x,y<B$.

\section{Multiplication and division are constant-time computable}\label{sec:timesDiv}

In this section, we sketch proofs that the multiplication and division
operations, $(x,y)\mapsto x \times y$ and $(x,y)\mapsto \lf x/y \rf$,
are in $\cpp$.  Recall that we assume that each input operand is
polynomial, \emph{i.e.} less than $N^d$, for a fixed integer $d$, and
is given (read) as the list of its $d$ digits in base $N$.  
To implement in constant time an arithmetic binary operation
$\mathtt{op}$ such as multiplication on operands less than $N$, or
less than $N^d$, it is convenient to convert these operands from base $N$ to 
base $B\coloneqq \lc N^{1/2} \rc$, which can be pre-computed by a RAM(+) in time $O(N)$ as established in  Appendix~\ref{app:B}.
This has two advantages:
\begin{itemize}
\item it will allow us to pre-compute a ``mini-table'' of the values $\op(x,y)$
  for all integers $x,y<B$, in time $O(B^2)=O(N)$;
\item any integer $x<N$ (and therefore less than $B^2$) is represented
  in base $B$ by the two-digit number $(x_1,x_0)_B$, with $x_1=x \divop B$
  and $x_0=x\modop B$; this principle can be extended to integers
  $x<cN$, for any fixed integer $c\ge 2$.
\end{itemize}
The conversion of integers from base $N$ to base $B$ will be done using the ``modulo $B$'' and ``div $B$'' tables that we now present.
The arrays $\textsc{Mod}B[0..cN]$ and $\textsc{Div}B[0..cN]$ defined by
$\textsc{Mod}B[x]\coloneqq x\modop B$ and
$\textsc{Div}B[x]\coloneqq x\divop B$ are (simultaneously) pre-computed in time $O(N)$ by the inductive formula
\begin{equation*}
(\textsc{Div}B[x], \textsc{Mod}B[x]) \coloneqq
\begin{cases}
(0,0) & \mathtt{if}\;x=0\\
(\textsc{Div}B[x-1]+1,0) & \mathtt{if}\;x>0 \;\land\; 
\textsc{Mod}B[x-1]=B-1\\
(\textsc{Div}B[x-1], \textsc{Mod}B[x-1]+1)& \mathtt{otherwise}.
\end{cases}
\end{equation*}

\noindent
The pre-computed mini-table of a binary operation $\op$, such
as subtraction or multiplication, on integers less than $B$, will be a one-dimensional array $A_{\op}[0..B^2-1]$ defined for all $x,y<B$ by 
\[ A_{\op}[Bx+y] \coloneqq \op(x,y).
\]
To manipulate such a table, we will use the additional array $\textsc{Mult}B[0..B-1]$ defined by
$\textsc{Mult}B[x]\coloneqq Bx$ and pre-computed in time $O(B)$ by the following induction: 

$\textsc{Mult}B[0]=0$ and $\textsc{Mult}B[x]=\textsc{Mult}B[x-1]+B$ when $x\neq 0$.

\smallskip \noindent
From now on, for the sake of simplicity, we write $Bx$ instead of $\textsc{Mult}B[x]$.

\smallskip
In particular, we have for multiplication $A_{\times}[Bx+y] \coloneqq x \times y$, for $x,y<B$.
Therefore, the induction $x\times 0=0$ and $x\times y=x\times (y-1) +x$, for $x,y<B$, is translated into the following inductive pre-computation of the array $A_{\times}[0..B^2-1]$, which can be implemented in time $O(B^2)=O(N)$ using two nested loops on $x$ and $y$:
\begin{equation}\label{eq:tablemult}
A_{\times}[Bx+y] \coloneqq
\begin{cases}
\;0 & \mathtt{if}\;\;y=0 \\ 
A_{\times}[Bx+(y-1)] + x & \mathtt{otherwise}.
\end{cases}
\end{equation}

\noindent
Thus, the array $A_{\times}[0..B^2-1]$ is a multiplication table  
for integers less than $B$.
 
\begin{notation}[``two-dimensional'' array]
For an array $A[0..B^2-1]$ and integers $x,y<B$, it is intuitive to write $A[x][y]$ to denote $A[Bx+y]$.
With this notation, the previous result is rewritten $A_{\times}[x][y]=x\times y$, for $x,y<B$.
\end{notation}

\begin{notation}[radix-$b$ notation]
 We write $x=(x_{k-1},\dots,x_1,x_0)_b$ to denote the representation of an integer $x$ in a base $b\ge 2$ --  \emph{radix-$b$ notation} or \emph{$b$-notation} -- with $k$ digits: $x=\Sigma_{i=0}^{k-1} x_i\times b^i$.
 \end{notation}

To compute in constant time the product of two integers $x,y<B^d$ given by their $B$-notations 
$(x_{d-1},\dots,x_1,x_0)_B$ and $(y_{d-1},\dots,y_1,y_0)_B$, we proceed in two stages:

\medskip
1) We compute the ``pseudo-representation'' in base $B$ of the product $z=x\times y=$ \linebreak
$z'_{2d-1}B^{2d-1}+\cdots+z'_1B+z'_0$ by the ``for'' instruction 
 
for each $i=0,1,\dots,2d-1$ set
$z'_j \gets \sum_{i=0}^j A_{\times}[x_i] [y_{j-i}]$

\noindent
where $x_i,y_i$ are zero for $i\ge d$. 
Note that each $z'_j$ is less than $dB^2$, which is $O(N)$, and can therefore be entirely contained in a single register.
 
\medskip
2) We normalize the representation $(z'_{2d-1},\dots,z'_1,z'_0)$ of $z$ to its $B$-notation 
$(z_{2d-1},\dots,z_1,z_0)_B$, 
 by executing the following 
 loop which uses the new variables $r_0=0,r_1,\dots,r_{2d-1}$:
 \begin{center}
 for $i$ from 0 to $2d-1$ set
 $(r_{i+1}, z_i) \gets ( \textsc{Div}B[z'_i + r_i], \textsc{Mod}B[z'_i + r_i])$.
 \end{center}
 This is possible because $z'_i + r_i=O(N)$, for all $i$.
 
  \smallskip
 Note that it is even easier to add two integers $x,y<B^d$ represented in base $B$ using a similar normalization loop\footnote{Also note that to compute a difference 
 $x-y$ for $0\le y \le x <B^d$, taking into account the carries, we can pre-compute an array $A_{\mathtt{diff}}[0..4B^2-1]$ such that $A_{\mathtt{diff}}[2Bx+y]\coloneqq \max(0,x-y)$ 
for $0 \le x,y < 2B$.
This also allows you to use tests of the form ``if $x<y$ then...'', etc. computable in constant time.
See~\cite{GrandjeanJachiet22} for more details.}. 

\medskip \noindent
{\bf Converting integers from base $N$ to base $B$.}
We are now ready to convert an input integer $X=(x_{d-1},\dots,x_1,x_0)_N$ from base $N$ to base $B$.
For this, we use Horner's method. 
Let $X_i$ be the integer whose radix-$N$ notation is $(x_{d-1},\dots,x_i)_N$, for $i=0,\dots,d-1$, 
and let $X_d\coloneqq 0$. 
We convert to base $B$ the integers $X_{d-1},\dots,X_1$ and $X_0=X$, in this order, by a decreasing recurrence.
Suppose that, for a certain $i<d$, the integers $X_{i+1}$, $N$ and $x_i$ are given by their $B$-notations (for example, it is $(N_1,N_0)_B$ for $N$, where 
$N_1=\textsc{Div}B[N]$ and $N_0=\textsc{Mod}B[N]$; for $X_{i+1}$, it is given by the recurrence hypothesis).
From Horner's formula $X_{i}=X_{i+1}\times N +x_i$, we obtain the $B$-notation of $X_i$ from that of $X_{i+1}$, by one multiplication and one addition, using the multiplication and addition procedures in base $B$ described above.

Note that the converse conversion from base $B$ to base $N$, used e.g. to present the outputs in the same base $N$ as inputs, can be carried out in a similar way by Horner's method. 
That is left to the reader.

\medskip
Of course, the most difficult of the four usual operations is division. 
However, when the divisor has only one digit, the quotient can be easily calculated in constant time using the classical method which takes into account the two most significant digits of the dividend to compute the leading digit of the quotient.

\medskip \noindent
{\bf Dividing by a one-digit integer.} 
Let $x=(x_{d-1},\dots,x_0)_B$ be an integer and $y<B$ be a positive integer.
The classical method computes the $B$-notation $(z_{d-1},\dots,z_0)_B$ of the quotient $z\coloneqq \lf x/y \rf$ by executing the following loop which uses the variables $r_d \coloneqq 0,r_{d-1},\dots,r_1,r_0$:

\smallskip
for $i$ from $d-1$ downto 0 set $z_i \gets (Br_{i+1}+x_i) \divop y\;$ and
$\; r_i \gets (Br_{i+1}+x_i) \modop y$.

\smallskip \noindent
Note that we have to compute in constant time expressions of the form $(Br+x) \divop y$ and
$(Br+x) \modop y$ for $r,x<B$. 
This can be done using the tables $A_{\mathtt{div}}[B^2-1]$, 
$A_{\mathtt{mod}}[B^2-1]$, $A_{B\times\mathtt{div}}[B^2-1]$
and $A_{B\times\mathtt{mod}}[B^2-1]$ defined by
$A_{\mathtt{div}}[x][y]\coloneqq x\divop y$, $A_{\mathtt{mod}}[x][y]\coloneqq x\modop y$,
$A_{B\times\mathtt{div}}[x][y] \coloneqq (Bx) \divop y$ and 
$A_{B\times\mathtt{mod}}[x][y] \coloneqq (Bx) \modop y$, for $x,y<B$, which can be pre-computed in time $O(B^2)=O(N)$. 
We leave it to the reader to fill in the details.

\medskip \noindent
{\bf Dividing by a large integer.} 
We present the tricky solution presented in Knuth's book~\cite{Knuth69}, Section~4.3.1, and attributed to Pope and Stein~\cite{PopeS60}. 
Let $x,y$ be positive integers, given by their $B$-notations, with $y \le x <B^d$. 
We want to compute the $B$-notation of the quotient $\lf x/y \rf$. 

\smallskip
-- First, we will apply the classical strategy which reduces the division problem to the problem of computing, for the integer $u$ whose $B$-notation is the prefix of the $B$-notation of $x$ such that $y\le u < B y$ and has exactly one more digit as that of $y$, the one-digit quotient $q\coloneqq \lf u/y \rf$. 

\smallskip
-- Second, by Pope and Stein's method~\cite{PopeS60} described in~\cite{Knuth69}, we will justify an approximation $\tilde{q}-2 \le q \le \tilde{q}$ of the quotient $q$, where the approximated quotient $\tilde{q}$ can be computed in constant time from the two leading digits of $u$ and the leading digit of $y$. 

This will allow us to compute successively, in constant time, $\tilde{q}$ and then $q$, and finally $\lf x/y \rf$.

\medskip
1) \emph{Classic strategy\footnote{E.g., if in decimal notation
we have $x=23768$ and $y=65$ then we obtain $x=237\times 10^2+68$, $y=(6,5)_{10}$ and $u=(2,3,7)_{10}$;
now if we take $x=36524$ and $y=14$ then we obtain $x=36\times 10^3+524$, $y=(1,4)_{10}$ and $u=(0,3,6)_{10}$, note the zero completion; 
in both cases, the number of digits used to represent $y$ (resp. $u$) is $m=2$ (resp. $m+1=3$).
}:}
We find (in constant time) the unique integers~$u,k,x'$ and $m$ such that 
$x=u\times B^k+x'$, $x'<B^k$ and $y\le u < B\times y$, 
with $y=(y_{m-1},\dots,y_0)_B$, $\;y_{m-1}\ge 1$, and $u=(u_m,u_{m-1},\dots,u_0)_B$.

\smallskip
If $q$ and $r$ denote the quotient and the remainder of the division of $u$ by $y$, we obtain 
\begin{center}
$x=(qy+r)\times B^k +x'=qy B^k + rB^k+x'$ with $1\le q <B$ and $0\le r <y$.
\end{center}
Note that we have $rB^k+x'<(y-1)B^k+B^k=y\times B^k.$
If we set $x''\coloneqq rB^k+x'$, we obtain $x/y=q\times B^k + x''/y.$
This gives the ``recursive'' formula
\begin{center}
$\lf x/y\rf=q B^k + \lf x''/y\rf$ with $x''<B^k.$
\end{center}
In particular, the positive quotient $q=\lf u/y\rf<B$ is the \emph{leading} digit of the quotient $\lf x/y \rf$.
This \emph{justifies} the classic division strategy that we have just recalled.
Thus, to prove that the quotient $\lf x/y\rf$ is computable in constant time, it \emph{suffices} to prove that the quotient $q=\lf u/y\rf$ can be computed in constant time (because the integers $u$ and $x''= rB^k+x'$ themselves are computed in constant time). 

\medskip
2) \emph{Computing the leading digit $q$ of the quotient:} The following surprising (non intuitive) lemma shows that there are at most three possible values for $q$ if the most significant digit of the divisor $y$ is at least 
$\lf B/2 \rf$.

\smallskip
Basically, this lemma expresses that the computation of the leading digit $q$ of the quotient $\lf u/y \rf$ reduces  to that of the quotient $\lf u_{(2)}/v\rf$  
of the integer $u_{(2)}\coloneqq (u_{m},u_{m-1})_B$ which is the ``prefix'' of length 2 (the two leading digits) of 
$u=(u_{m},\dots,u_0)_B$ by the leading digit $v$ of $y$.

\begin{lemma}[Approximation Lemma: Theorems A and B in~\cite{Knuth69}]\label{lemma:tildeq}
In the notation above, if we have $v\ge \lf B/2 \rf$ and if we set 
$\tilde{q} \coloneqq \min \left(\lf u_{(2)} / v \rf, \;B-1 \right)$,
then the leading digit $q$ of $\lf x/y \rf$ satisfies $\tilde{q}-2 \le q \le \tilde{q}$.
\end{lemma}

\begin{proof}
See the proof in Appendix~\ref{subsec:ApproxLemma} (or that of Theorems A and B in Section~4.3.1 of~\cite{Knuth69}).
\end{proof}

As Knuth~\cite{Knuth69} writes, ``the most important part of this result is that the conclusion is independent of $B$: no matter how large $B$ is, the trial quotient $\tilde{q}$ will never be more than 2 in~error''.

\medskip
Note that the term $\lf u_{(2)} / v \rf$ mentioned in the statement of Lemma~\ref{lemma:tildeq} can be computed in constant time by the procedure given in the paragraph ``Dividing by a one-digit integer'' because $v<B$. 
Thus, $\tilde{q}$ can be computed in constant time and so can $q$ itself as follows: 

\smallskip
if $x\ge \tilde{q}\times y$ then $q\gets \tilde{q}$ 
else if $x\ge (\tilde{q}-1)\times y$ then $q\gets \tilde{q}-1$
else $q\gets \tilde{q}-2$.

\medskip
It remains to explain how, without loss of generality, the condition $v\ge \lf B/2 \rf$ of 
Lemma~\ref{lemma:tildeq}, called by Knuth~\cite{Knuth69} a ``normalization condition'', can be assumed. 
Indeed, if $v<\lf B/2 \rf$ then Fact~\ref{fact:mu} below states that the normalization condition can be obtained if we multiply both $u$ and $y$ by $\mu \coloneqq \lf B/(v+1) \rf<B$,
which does not change the value of~$u/y$.

\smallskip
From the $B$-notation $(y_{m-1},\dots,y_0)_B$ of $y$, the following classical procedure computes in constant time the $B$-notation $(y'_m,y'_{m-1},\dots,y'_0)_B$ of the product 
$y'\coloneqq y\times\mu$ using the variables $c_0\coloneqq 0,\;c_1,\dots,c_m$ for successive carries and $z_0,\dots,z_{m-1}$ for intermediate results:

\smallskip \noindent
for $i$ from $0$ to $m-1$ set $z_i\gets y_i\times \mu+c_i$ and set 
$(c_{i+1},y'_i)\gets (z_i \divop B,\;z_i \modop B)$ ; $y'_m\gets c_m$.

\begin{fact}\label{fact:mu} 
\begin{itemize}
\item In this product procedure,
the carry $c_i$ is at most $(B-1)/(v+1)$, for each $i<m$;
\item we have $y'_m=c_m=0$ so that $y\times\mu$ is $y'=(y'_{m-1},\dots,y'_0)_B$, i.e. $y'$ has $m$ digits like $y$;
\item if $v=y_{m-1}<\lf B/2 \rf$, 
then $y'_{m-1} \ge \lf B/2 \rf$ ($y'$ satisfies the normalization condition).
\end{itemize}
\end{fact}

\begin{proof}
The first point of Fact~\ref{fact:mu} is proven by induction on the index $i=0,1,\dots,m-1$. 
We have $c_0\coloneqq 0$. 
Suppose that $c_i\le (B-1)/(v+1)$ for an index $i<m-1$.
We obtain at index $i$ the integer $z_i = y_i\times \mu + c_i$,
which gives $y'_i = z_i \modop B$ and $c_{i+1} = z_i \divop B$. 
We therefore get 
\begin{center}
$z_i \le B(B-1)/(v+1) +(B-1)/(v+1)$,
\end{center}
from which we deduce
$c_{i+1}= \lf z_i/B \rf \le (B-1)/(v+1)$ because $(B-1)/(v+1)<B$. 
This concludes the proof by induction of the first point of Fact~\ref{fact:mu}.
In particular, for $i=m-1$, we have 
$c_{m-1} \le (B-1)/(v+1)$, hence
$z_{m-1}=v\times \mu + c_{m-1}\le \frac{Bv}{v+1}+\frac{B-1}{v+1}=\frac{B(v+1)-1}{v+1}<B$.

\smallskip \noindent
We therefore obtain $y'_{m-1}=z_{m-1}<B$ and $y'_m=c_m=0$
so that the product $y\times \mu$ can be written with $m$ digits like $y$.

\medskip
It remains to prove $y'_{m-1}\ge \lf B/2 \rf$. 
We will use the following inequality.

\begin{claim}\label{claim:ineq}
$v\left(\frac{B}{v+1}-1\right)\ge B/2 -1$.
\end{claim}

\begin{proof}[Proof of the claim]
$v\left(\frac{B}{v+1}-1\right)-(B/2 -1)= (vB-v^2-v-vB/2-B/2+v+1)/(v+1)=$

\smallskip \noindent
$(vB/2-v^2-B/2+1)/(v+1)= ((v-1)B/2-(v-1)(v+1))/(v+1)=(v-1)(B/2-v-1)/(v+1)$.

\smallskip \noindent
This last expression is $\ge 0$  
because we have $v\ge 1$ and $v<\lf B/2 \rf \le B/2$ so that
$v+1\le \lf B/2 \rf \le B/2$.
Claim~\ref{claim:ineq} is proven.
\end{proof}

\noindent
\emph{End of the proof of Fact~\ref{fact:mu}.}
We have  $y'_{m-1}= z_{m-1} 
\ge y_{m-1}\times \mu = v \times \lf \frac{B}{v+1}\rf>v \left(\frac{B}{v+1}-1\right)$.
From Claim~\ref{claim:ineq}, we deduce
$y'_{m-1}> B/2 -1\ge \lf B/2 \rf -1$.
Hence, $y'_{m-1}\ge  \lf B/2 \rf$.
\end{proof}

Thus, we have established that, for each fixed $d$, the division operation 
$(x,y)\mapsto \lf x/y \rf$, for $0<y\le x <N^d$, is, like $+,-,\times$, computable in constant time with linear preprocessing.

\begin{corollary}\label{cor:divmod}
The $\mathtt{div}$ and $\;\mathtt{mod}$ operations belong to $\cpp$.
\end{corollary}

In our previous paper~\cite{GrandjeanJachiet22}, we gave two other constant-time algorithms computing division with linear preprocessing.

\medskip
From Corollary~\ref{cor:divmod} and Theorems~\ref{th:transitivity_base} and~\ref{th:transitivity2}, we deduce the following theorem which states that the complexity classes $\cpp$, $\textsc{LinTime}$ and $\cdlin$ are invariant whether the set of the allowed operations is $\{+\}$ or $\{+,-,\times,\mathtt{div},\mathtt{mod}\}$ or any set of $\cpp$ operations provided it includes $+$.

\begin{theorem}\label{th:+=+...mod}
For any set $\OP$ of $\cpp$ operations, we have the equalities: 

$\bullet$ $\cpp=\cpp(+)=\cpp(+,-,\times,\mathtt{div},\mathtt{mod})=\cpp(\OP\cup\{+\})$; 

$\bullet$ $\textsc{LinTime}(+)=\textsc{LinTime}(+,-,\times,\mathtt{div},\mathtt{mod})
=\textsc{LinTime}(\OP\cup\{+\})$; 

$\bullet$ $\cdlin(+)=\cdlin(+,-,\times,\mathtt{div},\mathtt{mod})=\cdlin(\OP\cup\{+\})$.
\end{theorem}

\medskip \noindent
{\bf More base conversions.}
The following technical lemma will be used in the proofs of Theorems~\ref{th: N->N^1/c} and~\ref{th:N^c->N} in the next Section~\ref{sec: N->N^1/c}.

\begin{lemma}\label{lemma: N to cRootN}
Let $c,d$ be any fixed positive integers. 
\begin{enumerate}
\item There exists a $\mathrm{RAM}(+,\mathtt{div},\mathtt{mod})$ 
which, for any reference integer $N$, with $N_1\coloneqq \lc N^{1/c} \rc$ known, 
and any input integer $X<N^d$ given in base $N$ in $d$ registers  $X=(x_{d-1},\dots,x_0)_N$, computes in constant time the radix-$N_1$ notation of~$X$  in $cd$ registers $(x'_{cd-1},\dots,x'_0)_{N_1}$. 
\item Conversely, there exists a $\mathrm{RAM}(+)$ 
which,
for any reference integer of the form  \linebreak 
$N_1\coloneqq \lc N^{1/c} \rc$, for an integer $N\ge 2^{cd}$, 
\begin{itemize}
\item \emph{Preprocessing:} computes in time $O(N_1)$ the tables necessary for the constant-time computation of operations $\mathtt{div}$ and $\mathtt{mod}$ on operands $<(N_1)^{cd}$, and
\item \emph{Conversion:} by using these tables and after reading an integer $Y<(N_1)^{cd}$, given in base $N_1$ in~$cd$ registers, $Y=(y_{cd-1},\dots,y_0)_{N_1}$, computes in constant time the $N$-notation\footnote{An integer $Y<(N_1)^{cd}$, for $N_1\coloneqq \lc N^{1/c} \rc$ and $N\ge 2^{cd}$, is less than $N^{d+1}$ and has therefore a notation in base $N$ with $d+1$ digits.
This is due to the following series of inequalities:
$Y< \lc N^{1/c} \rc^{cd}< (N^{1/c}+1)^{cd}\le (2N^{1/c})^{cd}=2^{cd}N^d\le N^{d+1}$.}
of~$Y$, $(y'_{d},\dots,y'_0)_{N}$, 
in the form of the list of integers 
$z_0,\dots,z_{cd+c-1}<N_1$ 
contained in $cd+c$ registers, so that
$(z_{ic+c-1},\dots,z_{ic})_{N_1}$ is the $N_1$-radix notation of $y'_i$, for each $i=0,\dots,d$,  
in other words so that $y'_i = \sum_{j=0}^{c-1}(N_1)^j \times z_{ic+j}$.
\end{itemize}
\end{enumerate}
\end{lemma}

\begin{proof}
Item~1 is justified as follows.
The $N_1$-digits $x'_0,\dots,x'_{cd-1}<N_1$ of $X<N^d\le (N_1)^{cd}$
are computed by the following program which uses the 
$cd+1$ variables $X_0,X_1,\dots,X_{cd}$:

\smallskip
$X_0 \gets X$;
for $i$ from $0$ to $cd-1\;$ set $\;x'_i \gets X_i \modop N_1\;$ and 
$\; X_{i+1}\gets X_i \divop N_1$.

\smallskip \noindent
The proof of Item~2 is similar. 
The $N$-digits $y'_0,\dots,y'_{d}$ are computed as follows:

\smallskip
$Y_0 \gets Y$;
for $i$ from $0$ to $d\;$ set \;$y'_i \gets Y_i \modop N\;$ and 
$\; Y_{i+1}\gets Y_i \divop N$.

\smallskip \noindent Note the technical point that since $N_1$ is the reference integer, each $N$-digit $y'_i<N \le (N_1)^c$ \linebreak
\emph{cannot} be contained in a \emph{single} register but must be represented by its $N_1$-digits \linebreak
$z_{ic},\dots,z_{ic+c-1}$ contained in $c$ registers.
\end{proof}

\section{Robustness of the $\cpp$ class for preprocessing time}\label{sec: N->N^1/c}

The reader may legitimately wonder why we required the preprocessing time in the definition of the $\cpp$ class to be $O(N)$.
Is requiring a linear-time preprocessing essential?
Could we instead take $N^{\varepsilon}$ with $0<\varepsilon<1$, or $N^c$ for any integer~$c>1$, or $N^{o(1)}$?
This section and Subsection~\ref{subsec:N^o(1)} answer these questions.

\medskip \noindent
\emph{Intuition of why preprocessing time $O(N)$ can be reduced to $N^{\varepsilon}$:} A $k$-ary operation $\op$ is in $\cpp$ if \emph{for any} fixed integer $d\ge 1$ \emph{and any reference integer} $N\ge 2$, the restriction $\op_{N,d}: [0,N^d[^k \to \mathbb{N}$ of $\op$ is computable in constant time after an $O(N)$ time preprocessing. 
In this definition of $\cpp$, observe the universal quantifications $\forall d$ and $\forall N$.
This allows us to replace the reference integer $N$ by its root $N_1\coloneqq \lc N^{1/c}\rc$, for any fixed integer $c$, as a new reference integer.
Thus, any operand or result $x<N^d$ of the $\op$ operation is less than $(N_1)^{cd}$: the degree $d$ is replaced by the degree $cd$.
Therefore, starting from any reference integer $N$ we obtain the following preprocessing broken down into two steps:

\smallskip
1. compute $N_1\coloneqq \lc N^{1/c}\rc$;

2. execute the preprocessing procedure of $\op$ on the reference integer $N_1$.

\smallskip \noindent
The point is that this new preprocessing runs in $O(N^{1/c})$ time because it is the time to compute the function $N\mapsto \lc N^{1/c}\rc$ (step~1) and the time of step 2 is also $O(N_1)$.
This is the principle of the proof that the preprocessing time can be reduced to $O(N^{1/c})$ for any $c>1$, therefore to $N^{\varepsilon}$ for any positive $\varepsilon>0$ (Theorem~\ref{th: N->N^1/c} and Corollary~\ref{cor:N->N1/c} below).

Using the same idea, we will show that we can also reduce a preprocessing time $N^c$ for a fixed integer $c>1$, to $O(N)$, see Theorem~\ref{th:N^c->N} below.

\medskip 
The following general definition which extends 
that of the complexity class $\cpp(\OP)$ by mentioning the preprocessing time will allow us to formalize precisely the results of this section.
\begin{definition}\label{def:PP T(N)}
For a set (or a list) $\OP$ of operations and a function $T:\mathbb{N}\to \mathbb{N}$, we denote by
$\cpp_{T(N)}(\OP)$ the class of operations computable in constant time by $\mathrm{RAM}(\OP)$ with preprocessing of time $O(T(N))$, 
using register contents $O(\max(T(N),N))$ and addresses $O(T(N))$. 
\end{definition}

\begin{remark}
Note that this definition involves an extension of our RAM model in the case where the time function $T(N)$ is not $O(N)$ but is larger (see Theorem~\ref{th:N^c->N} below): 
then register contents are $O(T(N))$ instead of $O(N)$. (See also Appendix~\ref{app:nonlinear}.)
\end{remark}

\noindent
In the notation introduced in Definition~\ref{def:PP T(N)}, Theorem~\ref{th:+=+...mod} is reformulated as follows:

$\cpp=\cpp_N(+)=\cpp_N(+,-,\times,\mathtt{div},\mathtt{mod})$.

\subsection{Reduced preprocessing time}\label{subsec:reducN^epsilon}
The following results state that the linear preprocessing time can be reduced to 
$N^{\varepsilon}$, for any positive $\varepsilon<1$.

\begin{theorem}\label{th: N->N^1/c} 
Let $\mathtt{op}$ be any $k$-ary operation. 
The following Assertions 1 and 2 are equivalent.
\begin{enumerate}
\item For each fixed integer $d\ge 1$, there is a $\mathrm{RAM}(+,\mathtt{div}, \mathtt{mod})$
$M_d$ such that, for any reference integer~$N$:
\begin{itemize}
\item \emph{Preprocessing:} $M_d$ computes some tables in time $O(N)$;
\item \emph{Operation:} by using these tables and 
after reading any $k$ integers $x_1,\ldots,x_k<N^d$, $M_d$ computes
$\mathtt{op}(x_1,\ldots,x_k)$ in \emph{constant time}. 
\end{itemize}
\item For all fixed integers $c,d\ge 1$, there is a $\mathrm{RAM}(+,\mathtt{div}, \mathtt{mod})$ $M_{c,d}$ such that, for any reference integer~$N$:
\begin{itemize}
\item \emph{Preprocessing:} $M_{c,d}$ computes some tables in time $O(N^{1/c})$;
\item \emph{Operation:} by using these tables and 
after reading any $k$ integers $x_1,\ldots,x_k<N^d$, $M_{c,d}$ computes 
$\mathtt{op}(x_1,\ldots,x_k)$ in \emph{constant time}.  
\end{itemize}
\end{enumerate}
This implies $\cpp_N(+,\mathtt{div},\mathtt{mod})=\cpp_{N^{1/c}}(+,\mathtt{div},\mathtt{mod})$.
\end{theorem}

\begin{corollary}\label{cor:N->N1/c}
For each integer $c\ge 1$, we have the equalities:
\begin{center}
$\cpp_N(+,-,\times,\mathtt{div},\mathtt{mod})=\cpp_N(+)=$\\
$\cpp_N(+,\mathtt{div},\mathtt{mod})=\cpp_{N^{1/c}}(+,\mathtt{div},\mathtt{mod}) = $
$\cpp_{N^{1/c}}(+,-,\times,\mathtt{div},\mathtt{mod})$.
\end{center}
\end{corollary}

\begin{proof}[Proof of Theorem~\ref{th: N->N^1/c}]
Assertion 1 is the special case $c=1$ of Assertion 2. 
Let us therefore assume Assertion 1. 
Without loss of generality, we assume $N\ge 2^{cd}$ (hypothesis of Lemma~\ref{lemma: N to cRootN}, Item~2).
The idea is to replace the reference integer $N$ by its root $N_1\coloneqq \lc N^{1/c} \rc$ and therefore to replace the degree~$d$ by~$cd$; this is justified by the fact that all operands $x_1,\ldots,x_k<N^d$ are less than $(N_1)^{cd}$. 
To satisfy Assertion~2, the nontrivial part of the following program of the RAM $M_{c,d}$ is to design a preprocessing whose time is $O(N^{1/c})$ so that the operation procedure, \emph{including} \emph{radix conversions}, remains constant-time, while using only addition 
and $\mathtt{div}$ and $\mathtt{mod}$ operations.

\begin{itemize}

\item \emph{Preprocessing:} From the reference integer $N$,\\ 
-- compute its $c$th root $N_1\coloneqq \lc N^{1/c} \rc$ (see below) and the $c-2$ variables $N_2,\dots,N_{c-1}$ defined by $N_i\coloneqq (N_1)^i$ (use the induction $N_i=N_{i-1}\times N_1$) and then, for each $i=2,\dots,c-1$, compute inductively the array 
$\textsc{Mult}N_i[0..N_1-1]$ defined by $\textsc{Mult}N_i[x]\coloneqq N_i\times x$, all that in time $O(N_1)$ and using only addition;\\ 
-- then, run the preprocessing of $M_{c\times d}$ on the ``new'' reference integer~$N_1$, which constructs some tables in time $O(N_1)=O(N^{1/c})$.

\item \emph{Operation:} By using all those tables and after reading $k$ operands 
$x_1,\ldots,x_k<N^d \le (N_1)^{cd}$ initially represented in base $N$:
\begin{enumerate}

\item convert each operand $x_i$ to base $N_1$ into $cd$ registers (according to Item~1 of Lemma~\ref{lemma: N to cRootN}, it takes a constant time);

\item then, run the operation procedure of $M_{c\times d}$ which computes the radix-$N_1$ representation of $y\coloneqq\mathtt{op}(x_1,\ldots,x_k)$ in constant time;

\item then, convert the result $y=(y_{cd-1},\dots,y_0)_{N_1}$ from base $N_1$ to base~$N$ in constant time:
according to Item~2 of Lemma~\ref{lemma: N to cRootN}, the radix-$N$ representation
of~$y$, $(y'_{d},\dots,y'_0)_{N}$, is obtained in the form of the list of integers 
$z_0,\dots,z_{cd+c-1}<N_1$ contained in $cd+c$ registers, so that
 $(z_{ic+c-1},\dots,z_{ic})_{N_1}$ is the $N_1$-radix representation of $y'_i$, for each $i=0,\dots,d$,  
in other words so that $y'_i = \sum_{j=0}^{c-1}(N_1)^j \times z_{ic+j}= 
z_{ic}+\sum_{j=1}^{c-1}N_j\times z_{ic+j}$;

\item therefore, set $y'_i \gets z_{ic }+ \textsc{Mult}N_1[z_{ic+1}]+\cdots + \textsc{Mult}N_{c-1}[z_{ic+c-1}]$, for each \linebreak 
$i=0,\dots,d$.

\end{enumerate}

\end{itemize}
Of course, the only point that remains for us to justify is that the
function $N\mapsto \lc N^{1/c} \rc$ can be computed by a
RAM$(+,\mathtt{div},\mathtt{mod})$ in time $O(N^{1/c})$.  For
simplicity, let us prove this for the case $c=3$, which is easy to
generalize to any $c\ge 2$.  Clearly, the following code based on the
inductive identity $(x+1)^3=x^3+3x^2+3x+1$, which only uses addition, $\mathtt{div}$ and $\mathtt{mod}$
\linebreak
($3x$ means $x+x+x$), computes the integer~$x$ such that $(x-1)^3<N
\le x^3$, i.e. $x=\lc N^{1/3} \rc$.

\begin{minipage}{0.1\textwidth}
  ~
\end{minipage}
\begin{minipage}{0.80\textwidth}
\begin{algorithm}[H]
  \caption{Computation of $x \coloneqq \lc N^{1/3} \rc$}
  \begin{algorithmic}[1]
    \State $\var{x} \gets 1$ ; $\var{x}_2 \gets 1$ ; $\var{x}_3 \gets 1$
    \While{$\var{x}_3 < N$}
    \Comment{see $\dagger$}
    
    \State \Comment{loop invariant: {\normalfont$\var{x}_2=\var{x}^2$; $\var{x}_3=\var{x}^3$ \emph{and}  $\var{x}_3 < N$}}
      \State $\var{x}_3 \gets \var{x}_3+ 3\var{x}_2 + 3\var{x} +1$ \Comment{$(x+1)^3$}
      \State $\var{x}_2 \gets \var{x}_2 + 2\var{x} +1$ \Comment{$(x+1)^2$}
      \State $\var{x}\gets \var{x}+1$
      \EndWhile
      \State
    \Comment{finally, $(x-1)^3< N \le x_3=x^3$}
  \end{algorithmic}
\end{algorithm}

$\dagger$ {The comparison operation can be obtained by pre-computing
  some tables~\cite{GrandjeanJachiet22} but this algorithm is meant to
  run before we have computed those tables. To avoid doing circular
  dependencies, we will use the division to translate this comparison
  $\var{x}_3 < N$ to the following expression 
  $(N \divop \var{x}_3 \neq 1  \text{ and } N \divop \var{x}_3 \neq 0)$ or
  $(N \divop \var{x}_3 = 1 \text{ and } N \modop \var{x}_3 \neq 0)$ which means
  $2 \var{x}_3 \le N$ or $\var{x}_3 < N < 2 \var{x}_3$.
  }
\end{minipage}

\bigskip
\noindent
Since $x$ is incremented exactly $\lc N^{1/3} \rc$ times, the algorithm runs in
$O(N^{1/3})$ time.
In the general case, use the binomial identity $(x+1)^c=\sum_{j=0}^c \binom{c}{j}x^j$ and introduce $c-1$ new variables $x_2,\dots, x_c$ with values $x_j=x^j$. 
Note that each program only uses $+, \mathtt{div}, \mathtt{mod}$ (no multiplication) 
because, since~$c$ and $j$ are fixed, each coefficient $\binom{c}{j}$ is also an explicit integer.

\medskip
This completes the proof of Theorem~\ref{th: N->N^1/c}.
\end{proof}

Theorem~\ref{th: N->N^1/c} implies that some complexity classes  $\textsc{Time}(T(N))$ with $T(N)=o(N)$ are robust despite the fact that a time $o(N)$ allows only a small part of the input to be consulted.
This is possible, as expressed by Corollary~\ref{cor:invTimeN1/c} which follows, if the input is given in the form of an array $I[1..N]$ and if the $R$-instruction $\mathtt{Input}\; i$ of Table~\ref{table:instRAM2} interpreted as Read $R[i]$ is replaced by the instruction $\mathtt{Input}\; i \; j$ interpreted as $R[i] \gets I[R[j]]$, for explicit integers $i,j$.

\begin{corollary}\label{cor:invTimeN1/c}
For each integer $c \ge 1$, the complexity class $\textsc{Time}(N^{1/c})$ is invariant according to the set 
$\OP$ of primitive operations provided $\{+,\mathtt{div},\mathtt{mod}\} \subseteq \OP \subset \cpp$.
\end{corollary}

\noindent
Can we get rid of the $\mathtt{div}$ and $\mathtt{mod}$ operations in Theorem~\ref{th: N->N^1/c} and Corollaries~\ref{cor:N->N1/c},~\ref{cor:invTimeN1/c}? More precisely:
\begin{openpb}
Have we $\cpp_N(+)=\cpp_{N^{1/c}}(+)$, for all $c>1$?\\
This would weaken the hypothesis of Corollary~\ref{cor:invTimeN1/c} which would become 
$+\in \OP\subset \cpp$.
\end{openpb}

\subsection{Increased preprocessing time}\label{subsec:polPP}

Theorem~\ref{th: N->N^1/c} states that the preprocessing time can be reduced from $O(N)$ to $N^{\varepsilon}$ for any positive $\varepsilon<1$.
Conversely, does increasing preprocessing time to any polynomial expand the class of operations that can be computed in constant time?
The following theorem answers this question negatively.

\begin{theorem}\label{th:N^c->N}
For each fixed integer $c\ge 1$ and any finite set of operations $\OP\subseteq\cpp$, 
we have $\cpp = \cpp_N(+) = \cpp_{N^c}(+) = \cpp_{N^c}(\OP\cup\{+\})$.
\end{theorem}

\begin{proof}
The equality $\cpp_{N^c}(+) = \cpp_{N^c}(\OP\cup\{+\})$ is proven in the same way as its special case $c=1$ stated in Theorem~\ref{th:+=+...mod}: 
this is true ``a fortiori'' because by definition the operations in $\cpp$ require a preprocessing time $O(N)$, which is $O(N^c)$.
The inclusion $\cpp_N(+) \subseteq \cpp_{N^c}(+)$ is just as obvious. 
To prove reciprocal inclusion, consider an operation $\op\in \cpp_{N^c}(+)$, say of arity 1 to simplify the notation.
By hypothesis, for any integer $d'\ge 1$, there exists a RAM(+) $M_{d'}$ which computes the function
$\op_{N,d'}: [0,N^{d'}[\to [0,N^{d'}[\cup \{\mathtt{overflow}\}$ in constant time after a preprocessing of time $O(N^c)$.

\medskip \noindent
\emph{Idea of the simulation:} It is very simple. Choose $d'\coloneqq cd$ and apply $M_{cd}$ to 
$\lc N^{1/c} \rc$ as the new reference integer. 
Therefore, the preprocessing time of $M_{cd}$ is $O\left(\lc N^{1/c} \rc^c\right)=O(N)$.

\smallskip
For a fixed integer $d\ge 1$, consider the RAM($+,\times,\mathtt{div},\mathtt{mod})$ $M'_d$ with the following program:

\smallskip \noindent
\emph{Preprocessing:} 
From the reference integer $N$ (assumed $\ge 2^{cd}$), compute the integer 
$N_1\coloneqq \lc N^{1/c} \rc$.
Then, execute the \emph{preprocessing phase} of $M_{cd}$ on the reference integer $N_1$.

\medskip \noindent
\emph{Operation:} 
Read an operand $x\le N^d$, i.e. read the $d$ digits $x_0,\dots,x_{d-1}$ of the notation of~$x$ in base $N$.
Then convert $x$ to base $N_1$ by the
algorithm of Lemma~\ref{lemma: N to cRootN} (item~1): we obtain $x=(x'_{cd-1},\dots,x'_{0})_{N_1}$ on which we execute the \emph{operation phase} of $M_{cd}$, which generates as output the integer 
$y\coloneqq \op(x)$ also in base $N_1$, if $\op(x)< (N_1)^{cd}$,
therefore in the form $y=(y_{cd-1},\dots,y_{0})_{N_1}$, and $\mathtt{overflow}$ otherwise.
Finally, compute the radix-$N$ notation of $y$, which is less than $N^{d+1}$ 
(because $(N_1)^{cd}<(N^{1/c}+1)^{cd}\le (2N^{1/c})^{cd} \le 2^{cd}N^d\le N^{d+1}$), 
hence of the form $(y'_d,\dots,y'_0)_N$, 
by the following algorithm (using the auxiliary variables $z_0,\dots,z_d,z_{d+1}$): 

$\;z_0 \gets y$ ; for $i$ from $0$ to $d\;$ set \;$y'_i \gets z_i \modop N\;$ and 
$\; z_{i+1}\gets z_i \divop N$.

\medskip \noindent
\emph{Complexity of $M'_d$:} The preprocessing time is $O(N)$ because from $N$ we can compute 
$N_1 \coloneqq \lc N^{1/c} \rc$ in time $O(N^{1/c})$ and perform the preprocessing of $M_{cd}$ applied to $N_1$ in time $O((N_1)^c)=O(N)$.
Moreover, the time of the operation phase of $M'_d$ is constant since it is composed of three constant-time procedures.

\medskip \noindent
\emph{Conclusion of the proof:}
From any input $x<N^d$ read in base $N$, the RAM $M'_d$ returns $\op(x)$, also in base~$N$, if 
$\op(x)< (N_1)^{cd}<N^{d+1}$, and returns $\mathtt{overflow}$ otherwise.
Noting that $N^d\le (N_1)^{cd}$, we can slightly modify the return statement of $M'_d$ as follows: 
if $\op(x)<N^d$ then return $\op(x)$ else return $\mathtt{overflow}$.
In other words, the RAM $M'_d$ thus modified computes the function $\op_{N,d}$.
We have proven $\op\in\cpp_N(+,\times,\mathtt{div},\mathtt{mod})$, hence $\op\in\cpp_N(+)$.
This concludes the proof of Theorem~\ref{th:N^c->N}.
\end{proof}

It is natural to ask whether Theorem~\ref{th:N^c->N} extends to non-polynomial time preprocessing.
The answer is no. 


\begin{theorem}\label{th:NonPol-CstPP}
Let $T(N)$ be a time function greater than any polynomial, i.e. such that for every~$c$, there exists an integer $N_0$ such that $T(N)>N^c$ for each $N\ge N_0$.
Then the inclusion $\cpp_N \subset \cpp_{T(N)}$ is strict.
\end{theorem}

The proof which essentially uses the time hierarchy theorem for RAMs of Cook and Reckhow~\cite{CookR73}
is given in Appendix~\ref{app:NonPol-CstPP}.

\section{Minimality of the $\cpp$ class}\label{sec:minimality}

In this section, we prove two minimality properties of the class $\cpp(+)$ of operations computable, using only addition, in constant time with linear-time preprocessing:
\begin{enumerate}
\item the class becomes \emph{strictly smaller}, and degenerates, if the preprocessing time is reduced to~$N^{o(1)}$ instead of $O(N)$;
\item as the only primitive RAM operation, addition \emph{cannot be replaced} by any set of \emph{unary operations}.
\end{enumerate}

\subsection{Minimality for preprocessing time}\label{subsec:N^o(1)}

Can preprocessing time be reduced to $N^{o(1)}$, while using only addition? 
The following result answers ``no'': 
multiplication can no longer be computed in constant time if the preprocessing time -- or even more weakly, the size of the pre-computed tables -- is reduced to~$N^{o(1)}$.

\begin{theorem}\label{th: no RAM N^o(1)} 
There is \emph{no} RAM such that, for any reference integer $N$:
\begin{enumerate}
\item \emph{Preprocessing:} The RAM computes some tables of size $S(N)=N^{o(1)}$; %
\item \emph{Product:} After reading two integers $x,y<N$, the RAM, using addition as its only primitive operation, computes $x \times y$ in \emph{constant time}.
\end{enumerate}
This implies $\times\not\in\cpp_{N^{o(1)}}(+)$ and the strict inclusion $\cpp_{N^{o(1)}}(+) \subset \cpp_N(+)$.
\end{theorem}

\noindent
\emph{Intuition of the proof:} 
The proof by contradiction is essentially a cardinality argument, which can apply to more general situations. 
Let $\OP$ be a finite set of operations. 
If a RAM($\OP$) can compute the function $(x,y)\mapsto x\times y$, for $x,y<N$, in constant time using some pre-computed tables of size $N^{o(1)}$, then we can deduce the following items. 

-- The constant-time computation that computes $x\times y$ can only use the integers $x,y,z_1,\dots,z_{\ell}$ where $z_1,\dots,z_{\ell}$ are the $\ell=N^{o(1)}$ distinct integers contained in the RAM registers at the end of the preprocessing or explicitly mentioned in the program.
 
-- Therefore, each $(x,y)\in[0,N[^2$ satisfies an equation of the form 
$x\times y=E(x,y,z_1,\dots,z_{\ell})$, where $E$ is an expression constructed from $x,y,z_1,\dots,z_{\ell}$ and the operations of $\OP$, and of size bounded by a constant integer $k$ (because the computation time is constant).

-- Such an expression $E$ can be regarded as a tree whose each internal node is labeled by an operation of $\OP$ and each leaf has a label among $x,y,z_1,\dots,z_{\ell}$.

-- Therefore, the number $\delta$ of the possible expressions $E(x,y,z_1,\dots,z_{\ell})$ is at most the product of the number $b(k)$ of such trees of size $\le k$, \emph{which is a constant} dependent on $k$ \emph{if the leaves are anonymous} (have no labels),
by the number of possible lists of (labels of) leaves of such a tree, which is at most the number of words $w=w_1,\dots,w_m$ of length $m\le k$ on the alphabet $\{x,y,z_1,\dots,z_{\ell}\}$.
This gives $\delta \le b(k) \times \sum_{m=1}^k (\ell+2)^m$ which is $N^{o(1)}$ because $k$ and $b(k)$ are constants and $\ell=N^{o(1)}$.

-- Thus, each of the $N^2$ pairs of operands $(x,y)\in[0,N[^2$ satisfies one of the $\delta=N^{o(1)}<N^{1/2}$ equations $x\times y=E(x,y,z_1,\dots,z_{\ell})$, shortly written $x\times y=E(x,y)$. 
This implies that at least one equation $x\times y=E(x,y)$ has more than $N^{3/2}$ solutions $(x,y)\in[0,N[^2$.

-- This \emph{multiplicity of solutions} for the \emph{same equation} will lead to a \emph{contradiction} in many cases, for example when the set $\OP$ of primitive operations is $\{+\}$ as we prove for 
Theorem~\ref{th: no RAM N^o(1)}.

\begin{proof}[Detailed proof of Theorem~\ref{th: no RAM N^o(1)} (by contradiction)]
Suppose that there exists a RAM $M$ satisfying conditions 1 and 2 of the theorem. 
Let $0=z_1<\ldots<z_{\ell}$ be the (ordered) list of the integers explicitly mentioned in the program of~$M$, including 0, or contained in its registers at the end of Step~1 (Preprocessing).
By definition, $\ell$ and the list $(z_1,\dots,z_{\ell})$ depend on $N$ but not on $x,y$, 
each $z_j$ is $O(N)$ and $\ell=O(S(n))=N^{o(1)}$. 

Let~$k$ be a constant upper bound of the time of Step~2 (Product). 
Let us only consider the integer operands $x,y < \lf N^{1/2} \rf$. 
Observe that we have $x\times y < N$ so this product is obtained in a single register. 
         
\begin{claim}\label{claim:exp}
At any time $i\le k$ of Step~2, each register of $M$ contains an integer of the form\footnote{The notation  $[v]$, for an integer $v$, denotes the interval of integers $\{1,\dots,v\}$.}
\begin{eqnarray}\label{eqn: for x,y}
ax+by+\sum_{j\in[2^i]} z_{f(j)}
\end{eqnarray}
where $0\le a,b\le 2^i$ and $f$ is a function from $[2^i]$ to $[\ell]$ ($a,b$ and $f$ depend on the register).
\end{claim}

\begin{proof}[Proof of the claim] 
It is proved by induction on the time $i\in[0,k]$. 
The claim is trivially true for $i=0$ since then expression~(\ref{eqn: for x,y}) reduces to $ax+by+z_{f(1)}$, for $a,b\le 1$: 
this expression represents~$x$ with $(a,b)=(1,0)$ (resp. $y$ for $(a,b)=(0,1)$) and $f(1)=1$ (remember that $z_1=0$) and it represents $z_j$, for $j\in [\ell]$, with $(a,b)=(0,0)$ and $f(1)=j$.

Suppose the claim is true at a time $i<k$. 
We only need to prove it at time $i+1$ when the instruction performs the addition of two registers.
By hypothesis, these two registers contain at time $i$ two integers $u,v$ of the form
$u=ax+by+\sum_{j\in[2^i]} z_{f(j)}$ and $v=cx+dy+\sum_{j\in[2^i]} z_{g(j)}$,
where $a,b,c,d\le 2^i$ and $f,g$ are functions (dependent on the two registers) from $[2^i]$ to $[\ell]$.
We obtain at time $i+1$ the sum $u+v=(a+c)x+(b+d)y+\sum_{j\in[2^i]}(z_{f(j)}+z_{g(j)})$.
We get $a+c\le 2^{i+1}$, $b+d\le 2^{i+1}$ and 
$\sum_{j\in[2^i]} (z_{f(j)}+z_{g(j)}) = \sum_{j\in[2^{i+1}]} z_{h(j)}$ 
for the function $h:[2^{i+1}]\to [\ell]$ defined 
by $h(j)\coloneqq f(j)$ if $j \le 2^i$ and $h(j)\coloneqq g(j-2^i)$ if $j > 2^i$.
This proves Claim~\ref{claim:exp}.
\end{proof}

In particular, Claim~\ref{claim:exp} implies that the output $x\times y$ is of the form~(\ref{eqn: for x,y}) for $i=k$. 
This means that the variables $x,y$ satisfy an equation of the form
\begin{eqnarray}\label{eqn: for x,y,xy}
ax+by+\sum_{j\in[2^k]} z_{f(j)}= x\times y
\end{eqnarray}
where $0\le a,b\le 2^k$ and $f$ is a function 
from $[2^k]$ to $[\ell]$.
Let us define the set of pairs of operands $P\coloneqq \{(x,y)\in \mathbb{N}^2 : x,y < \lf N^{1/2} \rf \}$.
Basically, the contradiction will come from the fact that the cardinality\footnote{As usual, we denote by $\vert A \vert$ the cardinality of a set $A$.} $\vert P \vert$ of $P$
is much greater than the number of possible equations~(\ref{eqn: for x,y,xy}). 
More precisely, for any fixed integer $N$, the number of possible values $c=\sum_{j\in[2^k]} z_{f(j)}$
is less than or equal to the number of functions $f:[2^k]\to\ell$, which is $\ell^{2^k}$.
Therefore, for a fixed $N$ the number of distinct equations~(\ref{eqn: for x,y,xy}), which is equal to the number of distinct tuples of coefficients $(a,b,c)$ satisfying the above conditions, is less than or equal to 
\begin{center}
$L\coloneqq (2^k +1)^{2} \times \ell^{2^k}=N^{o(1)}$
\end{center} 
since we have $\ell=N^{o(1)}$ and $2^k$ is a constant.
Let $E$ be the set of possible tuples of coefficients $(a,b,c)$ in equation~(\ref{eqn: for x,y,xy}). 
We have $\vert E \vert \le L=N^{o(1)}$.  

\begin{claim}\label{claim:3/4}
For each sufficiently large fixed $N$, there exists $(\alpha,\beta,\gamma)\in E$ such that the equation $\alpha x+\beta y+\gamma = x\times y$ has more than $N^{3/4}$ solutions in $P$.

Formally, $\vert \{(x,y)\in P : \alpha x+\beta y+\gamma = x\times y \} \vert > N^{3/4}$.
\end{claim}

\begin{proof}[Proof of the claim by contradiction]
Assume that for each $(a,b,c)\in E$ we have
\begin{center} 
$\vert \{(x,y)\in P : ax+by+c=x\times y\}\vert \le N^{3/4}.$
\end{center} 
Then the number of solutions $(x,y)\in P$ of an equation $ax+by+c=x\times y$ for $(a,b,c)\in E$ is
\begin{center} 
$\vert \bigcup_{(a,b,c)\in E} \{(x,y)\in P : ax+by+c=x\times y \}\vert \le \vert E \vert \times N^{3/4}=
N^{3/4+o(1)}.$
\end{center} 
Since $\vert P \vert=\Theta(N)$, this implies that there is a pair $(x,y)\in P$ 
which is never a solution of an equation $ax+by+c=x\times y$, for $(a,b,c)\in E$.
This contradicts Claim~\ref{claim:exp}. 
\end{proof}

Let us consider the interval $I\coloneqq \{x\in\mathbb{N} : x < \lf N^{1/2} \rf \}$ of cardinality 
$\vert I \vert \le N^{1/2}$. 
Define the set of solutions 
$\mathtt{Sol}(\alpha,\beta,\gamma)\coloneqq \{(x,y)\in P : \alpha x+\beta y+\gamma =x\times y \}$ 
for the tuple $(\alpha,\beta,\gamma)$ of Claim~\ref{claim:3/4}.

\begin{claim}\label{claim:at least 2}
The number $\nu$ of $x\in I$ such that $\vert\{y\in I : (x,y) \in \mathtt{Sol}(\alpha,\beta,\gamma)\}\vert \ge 2$
is at least $2$.
\end{claim}

\begin{proof}[Proof of the claim]
Suppose $\nu\le 1$. Assume $\nu = 1$ (the case $\nu = 0$ is simpler). 
Let $x_0$ be the unique integer in $I$ such that 
$\vert\{y\in I : (x_0,y)\in \mathtt{Sol}(\alpha,\beta,\gamma)\}\vert \ge 2$.
Decompose $\mathtt{Sol}(\alpha,\beta,\gamma)=S_0\cup S_1$ so that $S_0$ (resp. $S_1$) is the set of pairs
$(x,y)\in \mathtt{Sol}(\alpha,\beta,\gamma)$ such that $x=x_0$ (resp. $x\neq x_0$).
Of course, $\vert S_0 \vert \le \vert I \vert$. 
We also have $\vert S_1 \vert \le \vert I \vert$ 
since $\vert\{y\in I : (x,y) \in \mathtt{Sol}(\alpha,\beta,\gamma)\}\vert < 2$
for each fixed $x\in I\setminus\{x_0\}$. 
This implies $\vert \mathtt{Sol}(\alpha,\beta,\gamma) \vert \le 2\vert I \vert \le 2 N^{1/2}\le N^{3/4}$ 
for $N\ge 16$, which contradicts Claim~\ref{claim:3/4}. 
\end{proof}

\noindent \emph{End of the proof of Theorem~\ref{th: no RAM N^o(1)}}.
According to Claim~\ref{claim:at least 2}, there exist two distinct integers $x_0,x_1\in I$ 
for which there exist two distinct integers $y_1,y_2\in I$ 
(associated with $x_0$) and two distinct integers $y_3,y_4\in I$ 
(associated with $x_1$) such that the four pairs $(x_0,y_1)$, $(x_0,y_2)$, 
$(x_1,y_3)$ and $(x_1,y_4)$ belong to $\mathtt{Sol}(\alpha,\beta,\gamma)$.
This means that we have 
$\alpha x_0+\beta y_1+ \gamma =x_0y_1\;$ and $\;\alpha x_0+\beta y_2+ \gamma =x_0y_2$, 
which implies (by subtraction) $\beta (y_1-y_2)=x_0(y_1-y_2)$ and then $\beta =x_0$ 
(since $y_1\neq y_2$).
Considering $x_1,y_3,y_4$ instead of $x_0,y_1,y_2$, we deduce in exactly the same way $\beta =x_1$, which contradicts $x_0\neq x_1$.
This concludes the proof of Theorem~\ref{th: no RAM N^o(1)}.
\end{proof}

\begin{remark}
The main arguments of the proof of Theorem~\ref{th: no RAM N^o(1)} with Claims~\ref{claim:exp} and~\ref{claim:3/4} apply not only to multiplication but also to most binary operations. 
E.g., subtraction does not belong\footnote{To prove this, Claims~\ref{claim:exp} and~\ref{claim:3/4} remain unchanged (with $x-y$ instead of $xy$ in Claim~\ref{claim:3/4})  but Claim~\ref{claim:at least 2} is replaced by the statement (easily deduced from Claim~\ref{claim:3/4}): there exist integers $x_0,y_1,y_2$ such that $x_0>y_1>y_2$, and $(x_0,y_1)$ and $(x_0,y_2)$ are two solutions of the equation 
$\alpha x+\beta y+\gamma = x-y$. This implies by subtraction $\beta(y_1-y_2)=y_2-y_1$, which leads to 
$\beta=-1$, contradicting the condition $\beta\ge 0$.}~to $\cpp_{N^{o(1)}}(+)$.
This shows that the complexity class of operations computed in constant time with preprocessing time 
$N^{o(1)}$ (or even pre-computed tables of size~$N^{o(1)}$) is degenerate. 
\end{remark}

To what extent does the validity of Theorem~\ref{th: no RAM N^o(1)} depend on the condition that the only allowed operation is addition?
In other words, is Theorem~\ref{th: no RAM N^o(1)} still true if other arithmetic operations are allowed, 
e.g. $\mathtt{div}$ and $\mathtt{mod}$ as in Theorem~\ref{th: N->N^1/c}?
More precisely:
\begin{openpb}
Is the inclusion
$\cpp_{N^{o(1)}}(+,\mathtt{div},\mathtt{mod}) \subset \cpp_{N}(+,\mathtt{div},\mathtt{mod})$ strict?
\end{openpb}

\begin{remark}
We strongly conjecture but have no proof that $\times \not\in \cpp_{N^{o(1)}}(+,\mathtt{div},\mathtt{mod})$.
Nevertheless, we can extend Theorem~\ref{th: no RAM N^o(1)} with subtraction, i.e. prove 
$\times \not\in \cpp_{N^{o(1)}}(+,-)$. 
To do this, replace the expression~(\ref{eqn: for x,y}) in Claim~\ref{claim:exp} by 
$ax+by+\sum_{j\in[2^i]} z_{f(j)} - \sum_{j\in[2^i]} z_{g(j)}$
where $a,b\in[-2^i, 2^i]$ and $f,g$ are functions from $[2^i]$ to $[\ell]$.
The rest of the proof remains unchanged.
\end{remark}

\subsection{Minimality for primitive operations}\label{subsec:unary}

As we have proven that the $\cpp$ class does not change if we add to addition as native operations subtraction, multiplication, division, logarithm, etc.,
it is natural to ask whether limiting the set of native operations to successor and predecessor operations (instead of addition) \emph{restricts} the $\cpp$ class.
The following theorem gives a strongly positive answer to this question. 

\begin{theorem}\label{th: AddNotCTonSuccRAM} 
Let $\mathtt{UnaryOp}$ be a set of \emph{unary} operations.
There is \emph{no} RAM using only $\mathtt{UnaryOp}$ operations (no addition) such that, for any input integer $N$:
\begin{enumerate}
\item \emph{Preprocessing:} the RAM computes some tables using any space and time;
\item \emph{Addition:} the RAM reads any two integers $x,y<N$ and computes their sum $x+y$ in \emph{constant time}. 
\end{enumerate}
Furthermore, the result is still valid if the RAM is allowed to use any integers (without limiting register contents and addresses) and any (in)equality test, with $=,\neq,<,\leq$, in its branch instructions.

In particular, this implies $+ \in \cpp(+) \setminus \cpp(\mathtt{UnaryOp})$. 
\end{theorem}

\begin{remark}
By not limiting the preprocessing (item~1) and the integers used, 
Theorem~\ref{th: AddNotCTonSuccRAM} states a stronger result than what is required. 
\end{remark}

\begin{proof}[Proof of Theorem~\ref{th: AddNotCTonSuccRAM} by contradiction] 
Suppose that there exists a RAM($\mathtt{UnaryOp}$), denoted $M$,
which satisfies conditions 1,2 of Theorem~\ref{th: AddNotCTonSuccRAM}. 
Without loss of generality, the following conditions can be assumed.
\begin{itemize}
\item Step 1 of $M$ (the \emph{preprocessing}) ends with the internal memory registers $R[0],R[1],\ldots$ containing the successive integers 
$p(0),p(1),\ldots$, for a function $p:\mathbb{N}\to \mathbb{N}$.
\item Step 2 of $M$ (the \emph{addition}) does not modify the internal memory registers $R[0],R[1],\ldots$ (which are read-only throughout Step 2) but uses a fixed number of read/write registers, denoted~$A$ (the accumulator) and $B_1,\ldots,B_k$ (the buffers), all initialized to 0, and two read-only registers $X$ and~$Y$, which contain the operands $x$ and $y$, respectively. 
The program of Step 2 is a sequence $I_0,\ldots, I_{r}$ of numbered instructions of the following forms:
\begin{itemize}
\item $A\gets 0$ ; $A\gets X$ ; $A\gets Y$;
\item $A\gets \mathtt{op}(A)$,  
for some $\mathtt{op}\in \mathtt{UnaryOp}$;
\item $A\gets R[A]$; 
\item $B_i\gets A$ 
; $A\gets B_i$, for some $i\in\{1,\ldots,k\}$;
\item $\mathtt{if}$ $A\prec B_1$ $\mathtt{then}$ $\mathtt{goto}$ $\ell_0$ 
$\mathtt{else}$ $\mathtt{goto}$ $\ell_1$, for $\prec$  $\in \{=,\neq,<,\leq,>,\geq\}$ and some 
$\ell_0,\ell_1\in \{0,\ldots,r\}$ (branch instruction);
\item $\mathtt{return}$ $A$ (this is the last instruction $I_{r}$; recall that the instruction executed after any other instruction $I_{j}$, $j<r$, is $I_{j+1}$, except after a branch instruction).
\end{itemize}
\end{itemize}

\noindent
\emph{Justification:} The constraint that internal memory must be read-only throughout Step 2 is made possible by the assumption that the time of Step 2 is less than some constant. 
Thus, Step 2 can consult/modify only a constant number of registers $R[i]$. 
Instead of modifying the contents of a register $R[i]$, we copy it -- at the first instant when it must be modified -- into one of the buffer registers $B_1,\ldots,B_k$ so that its contents can be modified like that of the original register $R[i]$ that it simulates. 

More precisely, if $m$ is an upper bound of the time of
the simulated (original) constant-time program so that this program
modifies (at most) $m$ distinct registers $R[i]$, then $k\coloneqq
2m+2$ buffer registers $B_1,\ldots,B_k$ are sufficient for the
simulation: $B_1$ is the original buffer; the contents of $B_2$ is the
number of original registers $R[i]$ which have been currently modified
(in the simulated/original constant-time program); the sequence of $m$
buffers $B_3,\ldots,B_{m+2}$ (resp.\ $B_{m+3},\ldots,B_{2m+2}$) contain
the current sequence of addresses $i$ (resp.\ sequence of contents) of
these registers\footnote{When the original RAM tries to write
a value $v$ in a register $R[i]$, then the simulating RAM $M$ consults
the list of contents of the $s$ ``address" buffers
$B_3,\ldots,B_{s+2}$, where $s$ is the contents of the counter $B_2$.
If the address~$i$ is found, that means, for some $j$ 
($3 \leq j \leq s+2$), $B_j$ contains $i$, then the value $v$ is written in $B_{j+m}$
(the buffer that simulates $R[i]$).  Otherwise, the (not found)
address $i$ is copied in the buffer $B_{s+3}$ (the first buffer not
yet used after $B_{s+2}$), the value $v$ is written in the buffer
$B_{s+3+m}$ -- this buffer now simulates $R[i]$ -- and the counter
buffer $B_2$ is incremented by 1: it now contains $s+1$} $R[i]$. 
The reader can verify that the time of Step 2 of $M$ is always less than a
constant $\mu=O(m^2)$.

\bigskip
For $0\leq i \leq \mu$, let $(A(i),B_1(i),\ldots,B_k(i), L(i))$ be the sequence of contents of the read/write registers $A,B_1,\ldots,B_k$ and the value $L(i)$ of the ordinal instruction counter at instant~$i$ of Step~2. 
At instant 0, we have
$(A(0),B_1(0),\ldots,B_k(0), L(0))=(0,0,\ldots,0,0)$. For convenience, we adopt the following convention
when $L(i)=r$, for some $i<\mu$ (the instruction $I_{r}$ is $\mathtt{return}$~$A$):
\begin{eqnarray}\label{eqn:return}
(A(i+1),B_1(i+1),\ldots,B_k(i+1),L(i+1))\coloneqq (A(i),B_1(i),\ldots,B_k(i), r)
\end{eqnarray}

The following claim can be proved by an easy induction, see~\cite{GrandjeanJachiet22} for details.

\begin{claim}\label{claim:unary}
At each instant $i$ \emph{($0\leq i \leq \mu$)} of Step 2, an equation of the following form is true:
\begin{eqnarray}\label{eq: unarySucc}
(A(i),B_1(i),\ldots,B_k(i), L(i)) =
\left\{
\begin{array}{l}
\mathtt{if}\; t_1\;\mathtt{then}\;(a_1,b_{1,1},\ldots,b_{k,1},\lambda_1)\\
\ldots\ldots\ldots\ldots\ldots\ldots\ldots\ldots\ldots\ldots \\
\mathtt{if}\; t_q\;\mathtt{then}\;(a_q,b_{1,q},\ldots,b_{k,q},\lambda_q)\\
\end{array}
\right .
\end{eqnarray}

\noindent
Here, $\lambda_1,\dots,\lambda_q$ are explicit labels in $\{0,\dots,r\}$,
each test $t_j$ is a conjunction of (in)equalities 
$\bigwedge_{h} c_h \prec_h d_h$ with $\prec_h$  $\in \{=,\neq,<,\leq,>,\geq\}$
and each term $a_h,b_{h,j},c_h,d_h$ is of the (unary) form 
$(f_1\circ f_2 \circ \cdots \circ f_s)(u)$ where $u$ is either~$X$, or~$Y$, or the constant~0, and $s\geq 0$, and each $f_j$ is either a unary operation $\op\in\mathtt{UnaryOp}$ or the ``internal memory function" $p:\mathbb{N}\to \mathbb{N}$
which maps each address $\alpha$ to the contents $p(\alpha)$ of the corresponding read-only register $R[\alpha]$ (note that the function $p$ depends on $N$).

Moreover, the family of tests $(t_j)_{1\leq j \leq q}$ satisfies the \emph{partition property}: 
for each pair of integers $(x,y)$, \emph{exactly one} test $t_j$ is true for $X\coloneqq x$ and $Y\coloneqq y$.
\end{claim}

The contradiction will be proven as an application of the ``pigeonhole" principle.
By assumption, Step 2 of $M$ (addition step) returns the sum $x+y$ at time 
$\mu$, this means $A(\mu)=x+y$. 
Therefore, by Claim~\ref{claim:unary}, we obtain, for all $(x,y)\in[0,N[^2$,
\begin{eqnarray}\label{eq: PlusSucc}
x+y=A(\mu) =
\left\{
\begin{array}{l}
\mathtt{if}\; t_1(x,y)\;\mathtt{then}\;a_1\\
\ldots\ldots\ldots\ldots \\
\mathtt{if}\; t_q(x,y)\;\mathtt{then}\;a_q\\
\end{array}
\right .
\end{eqnarray}
where each test $t_j$ is a conjunction of (in)equalities 
$\bigwedge_{h} c_h \prec_h d_h$ with $\prec_h$  $\in \{=,\neq,<,\leq,>,\geq\}$
and each term $a_j,c_h,d_h$ is of the (unary) form $f(0)$, $f(x)$ or $f(y)$, in which $f$ is a composition 
$f_1\circ f_2 \circ \cdots \circ f_s$ of functions $f_i\in \mathtt{UnaryOp} \cup \{p\}$.

Without loss of generality, assume $N>q$. Let us define the $q$ parts $D_j$ of $[0,N[^2$ defined by the tests~$t_j$:
$D_j\coloneqq \{(x,y)\in [0,N[^2 \;\mid\; t_j(x,y)\}$. 
The partition property implies
\[ 
\sum_{j=1}^q \mathtt{card}(D_j)=N^2
\] 
from which it comes -- by the pigeonhole principle -- that at least one of the $q$ sets $D_j$ has at least $N^2/q$ elements. For example, assume $\mathtt{card}(D_1)\geq N^2/q$. 

Also, by~(\ref{eq: PlusSucc}), the implication $t_1(x,y) \Rightarrow x+y=a_1$  is valid.
Therefore, we get
\[
\mathtt{card}\{(x,y)\in [0,N[^2 \;\mid\; x+y=a_1\}\geq \mathtt{card}(D_1)\geq N^2/q>N
\]
using the assumption $N>q$.
Now, suppose that $a_1$ is of the form $f(x)$ (the case $f(y)$ is symmetric and the case $f(0)$ is simpler).
By a new application of the pigeonhole principle, we deduce that there is at least one integer $x_0\in [0,N[$ such that 
\[
 \mathtt{card} \{y\in [0,N[\; \;\mid\; x_0+y=f(x_0)\}>1.
 \]
 Thus, there are at least two distinct integers $y_1,y_2\in[0,N[$ such that $x_0+y_1=f(x_0)=x_0+y_2$, a contradiction.
 This concludes the proof of Theorem~\ref{th: AddNotCTonSuccRAM}.
\end{proof}
 
\begin{remark}
Using the pigeonhole principle again, we also easily see that not only addition but also 
each \emph{binary} operation whose result  \emph{really} depends on the \emph{two} operands  \emph{cannot} be computed in constant time on a RAM which \emph{only uses unary operations}. 
Intuitively, a constant number of comparisons (branch instructions) generates a constant number of cases. But to compute a binary operation, it is not enough to use a constant number of cases defined by (in)equalities between “unary” terms, each of which depends on only one operand.
\end{remark}

\begin{notation}
To keep it short and simple, we often write that an operation is \emph{constant-time computable} to express that it belongs to the class $\cpp$.
\end{notation}

The arguments of Theorem~\ref{th: AddNotCTonSuccRAM} concerning $\cpp$ are not sufficient to be applied to the complexity classes $\textsc{LinTime}$ and $\cdlin$ (defined in Subsection~\ref{subsec:ComplexClasses}).
\begin{openpb}
Are the complexity classes $\textsc{LinTime}$ and $\cdlin$ strictly smaller if the set of allowed operations is 
$\{\mathtt{succ},\mathtt{pred}\}$ instead of $\{+\}$?
An enumeration problem that doesn't seem to be in $\cdlin$ when the set of allowed operations is 
$\{\mathtt{succ},\mathtt{pred}\}$ is \linebreak
$\Pi: X_N \mapsto \Pi(X_N)$, with input $X_N\coloneqq (0,1,\dots,N-1)$ and output  
$\Pi(X_N)\coloneqq \{a^2\mid a\in [0,N[\}$.
\end{openpb}

\section{Constant time reduced to an addition expression}\label{sect:expression} 
As mentioned in the introduction, a constant-time procedure for computing an operation is reminiscent of direct access in a table. 
Of course, computing an operation by a single access in an array of size $O(N)$ is not possible if the operation is at least binary on operands $<N$ or has operands $\Theta(N^d)$ for $d>1$.
However, we will now show that this is possible if instead of a single direct access, we use an \emph{addition expression}.

\begin{definition}
An \emph{addition expression} (or for short, \emph{expression}) $E(y_1,\dots,y_q)$ over the variables $y_1,\dots,y_q$ is defined recursively as follows:
each variable or explicit integer is a \emph{simple} addition expression;
if $\alpha$ and $\beta$ are simple addition expressions, then $\alpha+\beta$ and $R[\alpha]$ are simple addition expressions; if $\alpha_1,\dots,\alpha_r$ are simple addition expressions, then 
$(\alpha_1,\dots,\alpha_r)$ is an \emph{$r$-tuple} addition expression.
\end{definition}

\begin{remark}
More liberally, we also allow arrays $A[\;]$ in addition expressions because any term 
$A[\alpha]$ is implemented as the addition expression $R[r\alpha+s]$, for fixed integers $r,s$.
\end{remark}

Addition expressions \emph{well compose}:
\begin{definition}[composition of addition expressions] If $E(\vec{x})$ is an $r$-tuple addition expression 
$(\alpha_1,\dots,\alpha_r)$ and $E'(y_1,\dots,y_r)$ is an addition expression on the $r$ variables $y_i$,
then $E' \circ E(\vec{x})$ denotes the addition expression obtained from $E'(y_1,\dots,y_r)$ by replacing each variable $y_i$ by $\alpha_i$.
\end{definition}

The following Theorem~\ref{th:csttime->exp} means that constant time can be reduced to a single addition expression.
Basically, the idea of its proof is to express each instruction of a constant-time procedure, and then the entire procedure, as an \emph{arithmetic expression} involving pre-computed tables and the binary operations $+$, $\times$ and $\subop$, where $x \subop y \coloneqq \max(0,x-y)$. 
The $\subop$ and $\times$ operations are finally eliminated in the resulting expression using only addition and other pre-computed tables.
Note that here again we need to convert the operands from base $N$ to base $N_1\coloneqq \lc N^{1/2}\rc$.
This allows us to replace any expression $x\subop y$ (or $x\times y$) with $x,y<N_1$ by the element $A_{\mathtt{diff}}[x][y]$ (or $A_{\times}[x][y]$) of a pre-computed array $A_{\mathtt{diff}}$ (or $A_{\times}$) of size $(N_1)^2=O(N)$.

\begin{theorem}\label{th:csttime->exp}
Let $\op$ be an operation in $\cpp$ of arity $k$.
Then, for each integer $d\ge 1$, there exist a $\mathrm{RAM}(+)$ and an addition expression 
$E(\vec{x}_1,\dots,\vec{x}_k)$ such that, 
for each integer $N\ge 2$,
\begin{enumerate}
\item the RAM performs a computation in time $O(N)$;
\item at the end of the computation, we have, for all integers $x_1,\dots,x_k<N^d$,
$\op(x_1,\dots,x_k)= E(\vec{x}_1,\dots,\vec{x}_k)$,
where $\vec{x}_i=(x_{i,d-1},\dots,x_{i,0})$ is the radix-$N$ notation of $x_i$, for $i=1,\dots,k$.
\end{enumerate}
\end{theorem}

\begin{example}\label{ex:termMult}
There is an addition expression $E_{\times}(x,y)$ such that, for each $N\ge 2$, we have after a preprocessing of time $O(N)$, the equality $x\times y = E_{\times}(x,y)$, for all $x,y<N$. See Appendix~\ref{append:+expForMult}.
\end{example}

\begin{proof}[Proof of Theorem~\ref{th:csttime->exp}]
To simplify the notation, let $\op$ be a unary operation.
Let $M$ be a RAM(+) which computes $\op$ in constant time with linear preprocessing.
Remember our convention that an input $x<N^d$ is read by $M$ in the form of its $d$ digits 
$x_{d-1},\dots,x_0$ in base $N$, the reference integer, and that in the same way the output $y\coloneqq\op(x)$ is produced by $M$ in base $N$.
The main idea of the simulation of $M$ is the one already used in the proof of 
Theorem~\ref{th: N->N^1/c}: 
instead of $N$, take as the reference integer a root $\lc N^{1/c}\rc$ of $N$; 
here we choose $N_1\coloneqq\lc N^{1/2}\rc$ ($c=2$).
Let us simulate the operation procedure of $M$ by a sequence of five constant-time transformations:

\begin{enumerate}

\item transform the input $x=(x_{d-1},\dots,x_0)_N$ into the tuple 
$(u_{2d-1},\dots,u_1,u_0)$ where $(u_{2i+1},u_{2i})$ is the radix-$N_1$ notation of $x_i$:
set $u_{2i+1}\gets x_i \divop  N_1$ and $u_{2i}\gets x_i \modop N_1$, for $0\le i<d$;

\item transform the tuple $(u_{2d-1},\dots,u_1,u_0)$ into the notation $(x'_{2d-1},\dots,x'_0)$ of $x$ in base $N_1$ 
by a procedure whose reference integer is $N_1$, i.e. whose registers only contain integers~$O(N_1)$,
which computes the radix-$N_1$ notation of the expression 
$\sum_{i=0}^{d-1} (u_{2i+1} N_1 +u_{2i})\times (n_1N_1+n_0)^i$ of~$x$ 
(the values $n_1\coloneqq N \divop N_1$ and $n_0 \coloneqq N \modop N_1$ can be pre-computed);

\item execute on this radix-$N_1$ representation of $x<(N_1)^{2d}$ the operation procedure of $M$, with reference integer $N_1$, the registers of which therefore only contain integers $O(N_1)$; 
this gives the output $y \coloneqq \op(x)$ presented in base~$N_1$,
i.e. $y=(y_{2d'-1},\dots,y_0)_{N_1}$, for a fixed integer $d'=O(d)$;

\item convert the result $y=(y_{2d'-1},\dots,y_0)_{N_1}$ from base $N_1$ to base~$N$: 
according to Item~2 of Lemma~\ref{lemma: N to cRootN} (where we take $c=2$), the radix-$N$ representation $(y'_{d'},\dots,y'_0)_{N}$ of $y$ is produced in the form
of a list of $2d'+2$ integers $(z_{2d'+1},z_{2d'},\dots,z_1,z_0)$  
contained in $2d'+2$ registers, so that $(z_{2i+1},z_{2i})_{N_1}$ is the $N_1$-radix representation of $y'_i$, for $0\le i \le d'$; 

\item therefore, set $y'_i \gets z_{2i }+ \textsc{Mult}N_1[z_{2i+1}]$, for each $i=0,\dots,d'$.

\end{enumerate}

Here are the important points:
\begin{itemize}
\item The first and last transformations are immediate thanks to the pre-computed arrays \linebreak
$\textsc{Mod}N_1[0..N]$, $\textsc{Div}N_1[0..N]$ and $\textsc{Mult}N_1[0..N_1]$: \\
-- transformation~1 is expressed by the assignment
$(u_i)_{0\le i<2d}\gets E_0((x_i)_{0\le i<d})$
with the addition expression 
$E_0((x_i)_{0\le i<d})\coloneqq (\textsc{Mod}N_1[x_i],\; \textsc{Div}N_1[x_i])_{0\le i<d}$;\\
-- similarly, transformation~5 is expressed by 
$(y'_i)_{0\le i \le d'} \gets E_2((z_j)_{0\le j\le 2d'+1})$
with the addition expression 
$E_2((z_j)_{0\le j\le 2d'+1}) \coloneqq (z_{2i}+\textsc{Mult}N_1[z_{2i+1}])_{0\le i \le d'}$.
\item The reference integer for intermediate transformations~2, 3, 4 is $N_1=\lc N^{1/2} \rc$: considering the composition of these three transformations as a single procedure, called $P$, whose only primitive operation is addition (thanks to pre-computed arrays) and whose input is 
$(u_{i})_{0\le i < 2d}$ and output is $(z_{j})_{0\le j\le 2d'+1}$, 
we now plan to simulate $P$ by a single addition expression  $E_1((u_{i})_{0\le i < 2d})$, i.e. by the assignment $(z_{j})_{0\le j\le 2d'+1}\gets E_1((u_{i})_{0\le i < 2d})$.
When we have the addition expression $E_1$, we will 
simulate the operation procedure of $M$ by the composition of $E_0$, $E_1$, $E_2$, i.e. by the assignment $(y'_i)_{0\le i \le d'} \gets E_2\circ E_1\circ E_0((x_i)_{0\le i<d})$.

\end{itemize}

\noindent
Let us construct the expression $E_1$ from the $P$ procedure.
The only points we take into account are that $P$ only uses addition, is constant time and only uses integers $O(N_1)$, including its inputs $u_i$, $0\le i <2d$, and its outputs 
$z_j$, $0\le j \le 2d'+1$.
Without loss of generality, the following conditions, similar to those assumed in the proof of Theorem~\ref{th: AddNotCTonSuccRAM}, can be assumed.

\medskip
The $P$ procedure does not modify the internal memory registers $R[0],R[1],\ldots$ (which are read-only when $P$ executes and only contain integers $O(N_1)$) but uses a fixed number of read/write registers, denoted~$A$ (the accumulator) and $B_0,\ldots,B_{m}$ (the buffers, with $m\ge 2d'+1$), all initialized to 0, and $2d$ read-only registers $U_{i}$ which contain inputs $u_{i}$, for $0\le i <2d$.
$P$ is a sequence $I_0,\ldots, I_{r}$ of numbered instructions of the following eight forms:
\begin{itemize}
\item $A\gets U_{i}$, for a fixed integer $i\in\{0,\dots,2d-1\}$;
\item $A\gets 1$ ; 
$A\gets A+B_0$ ; $A\gets R[A]$; 
\item $B_i\gets A$ ; $A\gets B_i$, for a fixed integer $i\in\{0,\ldots,m\}$;
\item $\mathtt{if}$ $A = B_0$ $\mathtt{then}$ $\mathtt{goto}$ $\ell_0$ 
$\mathtt{else}$ $\mathtt{goto}$ $\ell_1$; 
\item $\mathtt{return}(B_{2d'+1},\dots,B_1,B_0)$ 
(this is the last instruction $I_{r}$; it returns $(z_j)_{0\le j \le 2d'+1}$).
\end{itemize}

Let $\tau$ be a constant upper bound of the time of $P$.
Let $(A(t),B_0(t),\ldots,B_m(t), L(t))$, for $0\leq t \leq \tau$, be the sequence of contents of the 
read/write registers $A,B_0,\ldots,B_m$ and the value $L(t)$ of the ordinal instruction counter at time~$t$ of the execution of $P$. 
We have $(A(0),B_0(0),\ldots,B_m(0), L(0))=(0,\ldots,0)$. 
When $L(t)=r$ for some $t<\tau$, we keep the final configuration unchanged until time $\tau$:
$(A(t+1),B_0(t+1),\ldots,B_m(t+1),L(t+1))\coloneqq (A(t),B_0(t),\ldots,B_m(t), r)$.
The following claim similar to Claim~\ref{claim:unary} can be similarly proven by routine induction, see~\cite{GrandjeanJachiet22} for details.

\begin{claim}\label{claim:const+}

At each instant $t\le \tau$ of the execution of $P$ on an input $(u_i)_{0\le i <2d}$, an equality of the following form is true:
\begin{eqnarray}\label{eq: affine}
(A(t),B_0(t),\ldots,B_m(t), L(t)) =
\left\{
\begin{array}{l}
\mathtt{if}\; \theta_1\;\mathtt{then}\;(a_1,b_{0,1},\ldots,b_{m,1},\lambda_1)\\
\ldots\ldots\ldots\ldots\ldots\ldots\ldots\ldots\ldots\ldots \\
\mathtt{if}\;  \theta_q\;\mathtt{then}\;(a_q,b_{0,q},\ldots,b_{m,q},\lambda_q)\\
\end{array}
\right .
\end{eqnarray}

\noindent
Here, $0\le \lambda_k \le r$ for $k=1,\dots,q$, each test $ \theta_k$ is a conjunction of (in)equalities $\bigwedge_{h} c_h \prec_h d_h$ \linebreak
with $\prec_h$  $\in \{=,\neq\}$
and each term $a_h,b_{h,j},c_h,d_h$ is an addition expression on variables $u_i$.
Moreover, the family of tests $(\theta_k)_{1\leq k \leq q}$ satisfies the \emph{partition property}: 
\emph{exactly one} test $ \theta_k$ is true. 

\end{claim}

From equality~(\ref{eq: affine}) of Claim~\ref{claim:const+}, the following expression of each output $z_j$, $0\le j \le 2d'+1$, of the $P$ procedure is immediately deduced:
\begin{eqnarray}\label{eq:output zj}
z_j =B_j(\tau)=
\left\{
\begin{array}{l}
\mathtt{if}\; \theta_1\;\mathtt{then}\;b_{j,1}\\
\ldots\ldots\ldots\ldots\dots \\
\mathtt{if}\;  \theta_q\;\mathtt{then}\;b_{j,q}\\
\end{array}
\right .
\end{eqnarray}

Our last task is to express the branch expression~(\ref{eq:output zj}) as an addition expression.
To do this, let's arithmetize each test with $+, \subop$ and $\times$. 
Define $v(x=y)\coloneqq 1\subop ((x \subop y)+(y \subop x))$.
Note that if $x=y$ then $v(x=y)=1$, and $v(x=y)=0$ otherwise.
Now, define $v(x\neq y)\coloneqq 1\subop v(x=y)$ and, 
for $\theta\coloneqq \bigwedge_h c_h=d_h \land \bigwedge_h e_h \neq f_h$, define 
$v(\theta) \coloneqq \prod_h v(c_h=d_h) \times \prod_h v(e_h \neq f_h)$.
Clearly, if $\theta$ is true then $v(\theta)=1$, otherwise $v(\theta)=0$.
This allows us to reformulate equality~(\ref{eq:output zj}) as
\begin{eqnarray}\label{eqn:+,-,times}
z_j=\sum_{k=1}^q v(\theta_k)\times b_{j,k} 
\end{eqnarray}
using the fact that \emph{exactly one} test $ \theta_k$ is true.
Finally, we transform this expression of $z_j$ into an addition expression by removing the subtraction and multiplication operations as follows: replace each subexpression $x\subop y$ by
$A_{\mathtt{diff}}[x\times N_1+y]$ where $A_{\mathtt{diff}}[0..(N_1)^2-1]$ is a pre-computed array of size 
$(N_1)^2=O(N)$ defined as $A_{\mathtt{diff}}[x\times N_1+y]\coloneqq x\subop y$ for $x,y<N_1$. 
Then, we similarly replace each subexpression $x\times y$ by
$A_{\times}[\textsc{Mult}N_1[x]+y]$ with the pre-computed arrays 
$\textsc{Mult}N_1[\;]$ and $A_{\times}[\;]$. 
This is possible because the values of all subexpressions in~(\ref{eqn:+,-,times}) are $O(N_1)$ as the reader 
can verify.

This concludes the proof of Theorem~\ref{th:csttime->exp}
\end{proof}

As we have shown that the $\cpp$ class is closed under composition (Theorem~\ref{th:closedCompose}), the following section establishes another closure property of this class.

\section{A closure property of $\cpp$ for fast growing operations}\label{sec:logf}

In this section, we will show that a rapidly growing function $f$ belongs to CstPP if and only if its \emph{inverse} also belongs to this class.  
This has many applications: 
e.g, we are able to compute in this class, for any fixed $c$, the functions 
$x\mapsto \lf \log_c(x)\rf$ and $x\mapsto \lc \log_c(x)\rc$
(the inverses of $x\mapsto c^x$), and also the inverse of the Fibonacci sequence or the inverse of the factorial.

\begin{definition}[exponentially growing function] 
  We define an \emph{exponentially growing function} as a function
  $f:\mathbb{N} \rightarrow \mathbb{N}$ for which there exist a real number $c>1$ and an integer~$K$ 
  such that for all $x \geq K$, we have $f(x+1) > c \times f(x)$.
\end{definition}

\begin{example}
Next to the exponential functions $x\mapsto c^x$, the usual Fibonacci function $x\mapsto \mathtt{Fib}(x)$ and the factorial function are exponentially growing functions because $\mathtt{Fib}(x+1)>1.5 \times \mathtt{Fib}(x)$, 
for $x\ge 4$, and $(x+1)!>2\times x!$ for $x\ge 2$.
\end{example}

\begin{definition}[inverse function]\label{def:inv}
The \emph{inverse} of an unbounded function $f:\mathbb{N}\to\mathbb{N}$ is the non-decreasing function 
$f^{-1}:\mathbb{N}\to\mathbb{N}$ defined by $f^{-1}(x)\coloneqq\min\{y \mid f(y) \geq x\}$.
\end{definition}

Here are straightforward consequences of Definition~\ref{def:inv}.

\begin{lemma}\label{lem:f from f-1}
Let $f:\mathbb{N}\to\mathbb{N}$ be a non-decreasing function which is ultimately strictly increasing, 
i.e. such that there exists an integer~$K>0$ such that for all $x \ge K$, we have $f(x)> f(x-1)$.

1. For each $y\in\mathbb{N}$, we have $f(y)=\max\{x\mid f^{-1}(x)=y\}$.

2. For each $y\ge K$,  
we have $f(f^{-1}(y-1)+1)=y$. 

3. If $f^{-1}(x)\ge K$ then we have $f^{-1}(x+1)\le f^{-1}(x)+1$.

4. For each $y\ge K$, we have the equivalence 
$f(y)=x \iff f^{-1}(x)=y \land f^{-1}(x+1)=y+1$.

\end{lemma}

\begin{proof}
1. For $y>0$ (resp. $y=0$), $f^{-1}(x)=y$ means $f(y-1)< x \le f(y)$ (resp. $x\le f(y)$).
The largest integer $x$ satisfying the required inequality is therefore~$f(y)$.

2. Define $z\coloneqq f^{-1}( f(y-1)+1)$ for $y\ge K$. Then we have $f(z-1)<f(y-1)+1\le f(z)$, which means 
$f(z-1)\le f(y-1)< f(z)$.
If we had $z-1>y-1$ then we would have $z-1\ge y > y-1$ which implies $f(z-1)\ge f(y) > f(y-1)$ because $y\ge K$: a contradiction. 
Therefore, we have $z\le y$.
Similarly, if we had $y-1\ge z$ then we would have $f(y-1)\ge f(z)$, which is a contradiction. 
So we get $y-1<z$, which means $y\le z$.
So we have proven that $z=y$.

3. $y=f^{-1}(x)\ge K$ implies $x+1\le f(y)+1\le f(y+1)$ from which it comes 
$f^{-1}(x+1)\le y+1$. 

4. We can easily convince ourselves of the following equivalences for $y \ge K$:  
$f(y)=x \iff$
 
$f(y-1)<x = f(y)<x+1= f(y)+1\le f(y+1) \iff$
$f^{-1}(x)=y$ $\;\land\; f^{-1}(x+1)=y+1$. 
\end{proof}

The image set $\{f(y) \mid f(y)<N^d\}$ of an \emph{exponentially growing} function $f$ is \emph{exponentially scattered}.
Therefore, for each element $i\in\{\lf\log_2(f(y))\rf \mid f(y)<N^d\}$, the number of $y$ such that
$\lf\log_2(f(y))\rf=i$ is \emph{bounded by a constant}.
For each $x$, this will make it possible to look for $y=f^{-1}(x)$ among a \emph{constant number of candidates}.
To implement this idea, we need to establish that a logarithm function belongs to $\cpp$.
Indeed, we proved in~\cite{GrandjeanJachiet22} the following general result which we admit here.

\begin{proposition}\label{prop:log} 
The exponential function $(x,y)\mapsto x^y$ and its inverse function $(x,y)\mapsto \lf \log_x(y)\rf$ 
belong to $\cpp$.
\end{proposition} 

To prove the implication $\Rightarrow$ of the following theorem, we only need a special case of Proposition~\ref{prop:log}. 
We choose $x\mapsto \lf \log_2(x)\rf$.

\begin{theorem}\label{th:log_exp_fun}
Let $f$ be a non-decreasing and exponentially growing function.
Then we have the equivalence $f\in\cpp \iff f^{-1}\in\cpp$.
\end{theorem}

\begin{proof}
First, assume $f\in\cpp$ and let us deduce $f^{-1}\in\cpp$.
We want to compute in constant time, for each given integer $x<N^d$, the integer $y>0$ such that
$f(y-1)<x \le f(y)$, or $y=0$ if $x \le f(0)$.
Without loss of generality, we assume that $f$ is never zero\footnote{Otherwise, replace $f$ by the non-zero function $f'$ defined by $f'(y)\coloneqq f(y)+1$ and also belonging to $\cpp$. 
This is justified by the identity $f^{-1}(x)=(f')^{-1}(x+1)$.}.
 
 \medskip \noindent
 {\bf Preprocessing.}
  For $q\coloneqq\lf \log_2(N^d-1)\rf$, we create $q+1$  lists $L_0, \dots, L_q$ all initially empty. 
 We then iterate the following procedure for each $y=0,1,\dots$ until $f(y)\ge N^d$ 
 (recall that $f(y)$ can be computed in constant time):
 store $y$ into the list $L_{\lf \log_2(f(y))\rf}$.
  
Note that the number of integers~$y$ to be stored (and the total time to store them) is $O(\log N)$. 
This is due to the inequalities $N^d> f(y) > f(K) c^{y-K} \ge c^{y-K}$, for $y\ge K$, from which it comes 
$y-K <\log_c(N^d)$ and therefore $y<K+d \log_c N$.
  
\medskip
Since the function $f$ is non-decreasing, each non-empty list $L_i$ consists of an interval of integers and the family $(L_i)_{i \le q}$ is ``increasing'' in the following sense:
if $i < j \le q$, then for all $y\in L_i$ and $y'\in L_j$, we have $y\le y'$ 
(because $\lf \log_2(f(y))\rf < \lf \log_2(f(y'))\rf$ implies $y<y'$).

\medskip
The second part of the preprocessing consists in completing each empty list $L_i$ by distinguishing two cases:  
\begin{itemize}
\item if there exists $j>i$ such that $L_j$ is not empty, then compute the smallest element $y$ of the non-empty list $L_j$ of smallest index $j>i$, and set $L_i\gets (y)$;
\item otherwise, there is an integer $y$ satisfying the following inequalities \\
$\lf \log_2(f(y-1)) \rf < i \le \lf \log_2(N^d-1)\rf < \lf \log_2(f(y)) \rf$ and we set $L_i\gets (y)$. 
\end{itemize}

\noindent
In both cases, we have $\lf \log_2(f(y-1)) \rf < i < \lf \log_2(f(y)) \rf$
and we say that $L_i$ has been \emph{completed}. 
Observe that after the completion of all empty lists, the family $(L_i)_{i \le q}$ is still ``increasing''.

\medskip
Here is a crucial point:

\begin{claim}
The size of each list $L_i$ is bounded by a constant $K'$.
\end{claim}

\begin{proof}[Proof of the claim]
Two elements $y_1,y_2$  belong to the same list $L_i$ if and only if 
$\lf \log_2(f(y_1))\rf =\lf \log_2(f(y_2))\rf=i$. 
This means $2^{i} \le f(y_1) < 2^{i+1}$ and $2^{i} \le f(y_2) < 2^{i+1}$.
But we also have $f(y_2)>c^{y_2-y_1}f(y_1)$ for $y_2>y_1\ge K$.
These inequalities imply $2^{i+1}> f(y_2)>c^{y_2-y_1}f(y_1) \ge c^{y_2-y_1}2^{i}$
from which it comes $c^{y_2-y_1}<2$ and finally $y_2-y_1 < 1/\log_2 c$.
Thus, the size of each list $L_i$ is bounded by the constant $K'\coloneqq \max(K,1/\log_2 c)$.
\end{proof}

\noindent
\emph{Preprocessing time:} Recall that the first part of the preprocessing takes time $O(\log N)$.
The reader can easily verify that its second part -- the completion of the empty lists -- can be carried out in time 
$O(q)=O(\log N)$.

\medskip \noindent
 {\bf Constant time procedure.}
First, note that if $x\le f(0)$ then by definition $f^{-1}(x) = 0$.
and that case is easy to handle.
Given an $x>f(0)$ for which we want to find $f^{-1}(x) = \min\{y \mid f(y) \geq x \}$, 
we look for $y$ in the list $L_{\lf \log_2 x\rf}$. 
We distinguish two cases depending on whether the list $L_{\lf\log_2 x\rf}$ has been completed (= was initially empty) or not:
 \begin{enumerate}
\item $L_{\lf\log_2 x\rf}$ has been completed and by construction, the
  list contains a single element $y$, and $x,y$ satisfy the
  inequalities $\lf \log_2(f(y-1))\rf<\lf\log_2 x\rf< \lf
  \log_2(f(y))\rf$.
\item $L_{\lf\log_2 x\rf}$ has not been completed and it is of the
  form $(y,\dots,y+k-1)$ with $1\le k\le K'$, and $x,y$ satisfy the
  (in)equalities \\ $\lf \log_2(f(y-1))\rf<\lf\log_2 x\rf = \lf
  \log_2(f(y))\rf = \dots = \lf \log_2(f(y+k-1))\rf < \lf
  \log_2(f(y+k))\rf$.
\end{enumerate}
 In case~1, we obtain $f(y-1)<x<f(y)$, which implies $f^{-1}(x)=y$.
 In case~2, we obtain $f(y-1)<x<f(y+k)$, which means $f(y+\delta-1)<x \le f(y+\delta)$ for 
 $\delta \in \{0,\dots,k\}$, or equivalently $f^{-1}(x)=y+\delta$.
 In all cases, $f^{-1}(x)$ is computed in constant time.
 Thus, we have proven the implication $f\in\cpp \Rightarrow f^{-1}\in\cpp$.
 
 \medskip
Let us prove the converse implication\footnote{Roughly, the guiding idea of the proof of the implication $f^{-1}\in\cpp \Rightarrow f\in\cpp$ is the following.
Since the number of integers $y$ such that $f(y)<N^d$ is very small, we can store all the values $f(y)<N^d$ in a very small array $\textsc{F}[0..\ell]$, $\ell=O(\log N)$, such that $\textsc{F}[y] \coloneqq f(y)$ for all $y\le\ell$.
This array can be easily computed from $f^{-1}$ by $\ell+1$ dichotomous searches.}. 
Suppose $f^{-1}\in\cpp$. 
We are therefore allowed to prove $f\in\cpp$ using $f^{-1}$ as a primitive operation.
Let us define $\ell \coloneqq \max \{y\mid f(y) < N^d\}$. 
Since $f(\ell)<N^d\le f(\ell+1)$, we obtain $\ell+1=f^{-1}(N^d)$ and therefore $\ell=f^{-1}(N^d)-1$.
The interval $[0,\ell]$ of elements $y$ for which $f(y)<N^d$ is very small, namely $\ell=O(\log N)$.
Therefore, to prove $f\in\cpp$ it will suffice to compute in $O(N)$ time  an array $\textsc{F}[0..\ell]$ 
such that $\textsc{F}[y]=f(y)$ for all $y\le\ell$.
Without loss of generality, we may assume that the $K$ values $f(0),\dots,f(K-1)$ are known and initially stored as $\textsc{F}[0],\dots,\textsc{F}[K-1]$.
The following algorithm computes (and stores) the values  $f(K),\dots,f(\ell)$ by $\ell-K+1$ dichotomous searches using Lemma~\ref{lem:f from f-1} and mainly its Item~4, the equivalence
$f(y)=x \iff f^{-1}(x)=y \land f^{-1}(x+1)=y+1$ for $y\ge K$.

\begin{minipage}{\almosttextwidth}
 \begin{algorithm}[H]\label{algo:f}
 \caption{Pre-computation of $\ell$ and the array $\textsc{F}$ using $f^{-1}$ as a primitive operation}
 \begin{algorithmic}[1]
  \State $\ell \gets f^{-1}(N^d)-1$ \Comment{as justified above}
  \State $\textsc{F}\gets $ an array of indices $0..\ell$ 
  \State $\textsc{F}[0]\gets f(0)$ ; $\dots$ ; $\textsc{F}[K-1]\gets f(K-1)$
  \Comment{$f(0),\dots,f(K-1)$ are explicit constants}
  \For{$y \From K \To \ell$} \Comment{Compute $\textsc{F}[y]\coloneqq f(y)$}
     \State $b_0 \gets \textsc{F}[y-1]+1$ ; $b_1 \gets N^d$ 
     \State \Comment{$f^{-1}(b_0) = f^{-1}(f(y-1)+1) = y < \ell+1=f^{-1}(N^d)= f^{-1}(b_1)$ by Item~2 of Lemma~\ref{lem:f from f-1}}
     \While{$b_0+1<b_1$} \Comment{loop invariant: $y = f^{-1}(b_0) < f^{-1}(b_1)$}
       \State $m\gets (b_0+b_1)/2$ 
       \If{$y= f^{-1}(m)$} $b_0\gets m$ \Comment{$y=f^{-1}(b_0) = f^{-1}(m) < f^{-1}(b_1)$}
       \Else $\;b_1\gets m$ \Comment{$y=f^{-1}(b_0) < f^{-1}(m)=f^{-1}(b_1)$}   
       \EndIf
     \EndWhile
     \State $\textsc{F}[y]\gets b_0$  \Comment{justified here below}
  \EndFor
 \end{algorithmic}
 \end{algorithm}
\end{minipage}

\medskip \noindent
\emph{Justification of line~11}:  
Here, $b_0+1 = b_1$ and 
$f^{-1}(b_0) = y < f^{-1}(b_0+1) \le f^{-1}(b_0)+1$ 
(by the loop invariant and Item~3 of Lemma~\ref{lem:f from f-1}).
Thus we obtain
$y=f^{-1}(b_0)$ and the circuit of inequalities $f^{-1}(b_0)+1=y+1\le f^{-1}(b_0+1)\le f^{-1}(b_0)+1$, which gives $y+1 = f^{-1}(b_0+1)$ and therefore $f(y)=b_0$ by Item~4 of Lemma~\ref{lem:f from f-1}.

\medskip
All in all, the array $\textsc{F}[0..\ell]$ is pre-computed by our algorithm in time 
$O(\ell\times\log(N^d))=O((\log N)^2)=O(N)$ and, for all integers $N\ge 2$ and $d\ge 1$, the image $f_{N,d}(x)$ of any $x<N^d$ is computed in constant time: 
it is $\textsc{F}[x]$ if $x\le \ell$ and $\mathtt{overflow}$ otherwise.

\smallskip
This completes the proof of Theorem~\ref{th:log_exp_fun}.
\end{proof}

\begin{corollary}\label{cor:FibFact}
The inverse $\mathtt{Fib}^{-1}$ of the Fibonacci function $\mathtt{Fib}$ and the inverse 
$\mathtt{Fact}^{-1}$ of the factorial function, i.e. the functions defined by 
$\mathtt{Fib}^{-1}(x)\coloneqq \min \{y \mid \mathtt{Fib}(y)\ge x\}$ and \linebreak
$\mathtt{Fact}^{-1}(x)\coloneqq \min \{y \mid y!\ge x\}$ belong to $\cpp$.
\end{corollary}

\begin{proof}
This is an immediate consequence of Theorem~\ref{th:log_exp_fun} because the exponentially growing functions~$\mathtt{Fib}$ and factorial belong to $\cpp$ as proven in Appendix~\ref{appendix:Fib fact}.
\end{proof}

Theorem~\ref{th:log_exp_fun} applies only to unary operations.
\begin{openpb}
Can Theorem~\ref{th:log_exp_fun} be extended to operations of any arity? 
\end{openpb}
For example, this would give an alternative proof that the function $(x,y)\mapsto \log_x y$ belongs to 
$\cpp$.

\smallskip
Another natural question arises: can the first implication of Theorem~\ref{th:log_exp_fun} be extended to functions $f$ whose growth is faster than that of any polynomial, such as the function \linebreak
$f:x\mapsto x^{\lf\log_2 x\rf}$ (which belongs to $\cpp$ as a composition of functions in $\cpp$)?
\begin{openpb}
Let $f:\mathbb{N}\to \mathbb{N}$ be a function in $\cpp$ such that
$\lim_{x\to\infty}f(x)/x^k=\infty$, \linebreak
for each $k$.
Does the inverse function $f^{-1}$ also belong to $\cpp$?
\end{openpb}

\begin{remark}
By a variant of the proof of the first implication of Theorem~\ref{th:log_exp_fun}, we can prove the following variant of this implication\footnote{This implication can be applied to reprove that the inverses of the Fibonacci and factorial functions belong to $\cpp$ because they satisfy its hypothesis.}: 
If $f:\mathbb{N}\to \mathbb{N}$ is a non-decreasing function in $\cpp$ such that for each integer~$p$, there exist integers $K$ and $C$ such that for every integer $x\ge K$, we have $f^{-1}((x+1)^p)\le f^{-1}(x^p)+C$,
then $f^{-1}\in\cpp$.

Here is a sketch of proof:
the hypothesis implies  $f^{-1}(N^d)\le f^{-1}(K^d)+C(N-K)=O(N)$; 
each integer $y$ such that $f(y)<N^d$ is put in the list $L_{i}$ for $i\coloneqq \lf(f(y))^{1/d}\rf<N$;
this implies $i^d\le f(y) < (i+1)^d$; thus, applying $f^{-1}$ and using the hypothesis we obtain 
$f^{-1}(i^d)\le y \le f^{-1}((i+1)^d)\le f^{-1}(i^d)+C$, for $i\ge K$; 
therefore, if $i\ge K$ then the list $L_i$ contains at most $C+1$ elements which are 
$f^{-1}(i^d)+j$, for $0\le j \le C$.
\end{remark}

Of course, any operation defined as a polynomial belongs to $\cpp$.
But what about its inverse?
\begin{openpb}
Let $f:\mathbb{N}\to \mathbb{N}$ be a function defined as $f(x)\coloneqq P(x)$ for a polynomial 
$P\in\mathbb{N}[X]$ such that $\lim_{x\to\infty}P(x)=\infty$. 
Does the inverse function $f^{-1}$ belong to $\cpp$?
\end{openpb}
The answer is positive for polynomials of degree at most 4 because the roots of such a polynomial are computed by radicals which can be computed in constant time using Newton's approximation method as demonstrated in~\cite{GrandjeanJachiet22}.
We hope that the same positive answer can be obtained for the general case by an extension of this proof.

\section{Conclusion and other open problems} 
In this paper, we have proven that the $\cpp$ class is very robust:
\begin{itemize}
\item it is not sensitive to the set of primitive operations of the RAM, provided it includes addition  
(Theorems~\ref{th:transitivity2}, \ref{th:+=+...mod} and~\ref{th:robustT}, also proved in~\cite{GrandjeanJachiet22});
\item it is not sensitive to the preprocessing time, which can be \emph{reduced} from $O(N)$ to~$N^{\varepsilon}$, for all $\varepsilon>0$, if $+$, $\mathtt{div}$, $\mathtt{mod}$ are among the primitive operations (Theorem~\ref{th: N->N^1/c} and Corollary~\ref{cor:N->N1/c}) or \emph{increased} to $N^c$, for all $c>1$ (Theorem~\ref{th:N^c->N});
\item any constant-time procedure can be reduced to an expression (a return instruction) involving only additions and direct memory accesses (Theorem~\ref{th:csttime->exp}).
\end{itemize}

Note the transitive effect of Theorem~\ref{th:transitivity_base}: 
e.g., the division operation is in $\cpp$ because $+,-,\times\in\cpp$, which makes it possible to show that the function $x\mapsto \lf \log_2(x)\rf$ (or more generally the function $(x,y)\mapsto \lf \log_x y \rf$, see~\cite{GrandjeanJachiet22}) also belongs to $\cpp$, which we in turn used crucially to prove Theorem~\ref{th:log_exp_fun} whose Corollary~\ref{cor:FibFact} states that the inverses of the Fibonacci function and the factorial function belong to $\cpp$. 

\medskip
As a side effect, recall that Lemma~\ref{lemma:fund} also establishes the invariance of the ``minimal'' complexity classes $\textsc{LinTime}$ and $\cdlin$ -- and a fortiori, all classes of higher complexity -- as a function of the set of primitive RAM operations, provided it includes addition (Theorems~\ref{th:transitivity2}, \ref{th:+=+...mod} and~\ref{th:robustT}).

\medskip
We believe the main outstanding question concerns the boundaries of the $\cpp$ class:
we think that if an operation is piecewise monotonic or is local (or is composed of such operations), then we can reasonably conjecture and try to prove that it is in $\cpp$, as the following examples suggest. 
\begin{itemize}
\item All operations inspired by arithmetic and proven to be in $\cpp$ in~\cite{GrandjeanJachiet22} and/or in this paper are monotonic: $+,-,\times,\mathtt{div},\log$. 
Note that the $\mathtt{mod}$ operation is defined by an expression on $\mathtt{div}, \times$ and~$-$.
The function $x\mapsto \lf x^{p/q}\rf$, for a fixed rational exponent $p/q$, is also in $\cpp$ because it is composed of the $\cpp$ functions $x\mapsto x^p$ and $y\mapsto \lf y^{1/q}\rf$. 
\item The division algorithm of Pope and Stein~\cite{PopeS60} presented in Knuth's book~\cite{Knuth69} and recalled here is local. 
Similarly, the proof given in~\cite{GrandjeanJachiet22} that the operations computable in linear time on Turing machines or cellular automata belong to $\cpp$ is essentially based on the locality of these computation models.
\end{itemize}
It would be interesting to generalize and classify the methods used to prove that problems belong to $\cpp$, mainly Newton's approximation method (used for square root or more generally $c$th~root, for fixed $c$, see~\cite{GrandjeanJachiet22}) and local methods, used e.g. for the class of operations computable in linear time on cellular automata~\cite{GrandjeanJachiet22}.
We believe this could significantly expand – and provide a better understanding of – the class of operations known to be in $\cpp$.

\medskip
As we hope to have convinced the reader in this paper, the $\cpp$ class is flexible and broad.
Therefore, we believe that, as with other robust and broad complexity classes such as 
$\textsc{PTime}$ and $\textsc{LinTime}$, there are \emph{countless} algorithmic strategies for computing operations in $\cpp$. 
In other words, it is impossible to group them into a single model or a fixed number of models.
For the same reason, we believe that with our current knowledge it is impossible to \emph{prove} that a given ``natural'' operation is \emph{not} in $\cpp$, 
as it is currently still impossible for $\textsc{PTime}$ and even $\textsc{LinTime}$.

\medskip 
Let us list some operations whose memberships in the $\cpp$ class are open problems. 

\medskip 
{\bf Generalized root.} 
We strongly conjecture but cannot prove that the function $(x,y)\mapsto \lf x^{1/y} \rf$ belongs to $\cpp$. 
Indeed, we proved in~\cite{GrandjeanJachiet22} that its restriction for $y$ in the union of intervals $[1,(\log_2 N)/(12 \log_2 \log_2 N)]\cup[\log_2 N ,N^d[$, which seems very close to the desired interval $[1,N^d[$, is in $\cpp$.

\medskip 
{\bf Exponential, logarithm and beyond.}
In~\cite{GrandjeanJachiet22}, we proved that the function $(x,y)\mapsto x^y$ and its inverse 
$(x,y)\mapsto \lf \log_x y \rf$ are in $\cpp$.
We strongly conjecture that the iterated exponential $(x,y)\mapsto x^{(y)}$, defined by
$x^{(0)}\coloneqq 1$ and $x^{(y+1)}\coloneqq x^{x^{(y)}}$, or equivalently its inverse, the iterated logarithm
$(x,y)\mapsto \log^*_x(y)$ defined by
$\log^*_x(y)\coloneqq 0$ if $y\le 1$ and $\log^*_x(y)\coloneqq 1+ \log^*_x(\log_x y)$ otherwise, 
belongs to $\cpp$.
Similarly, we believe that the $\cpp$ class contains the various versions of the Ackermann function.
On the other hand, we think that the belonging to $\cpp$ of the integer part of the natural exponential 
$x\mapsto \lf e^x \rf$ is a completely open problem.

\medskip
{\bf Modular exponential and GCD.}
Even though it is well-known that the ternary operation $(x,y,z)\mapsto x^y \modop z$  and the binary operation $(x,y)\mapsto \mathtt{gcd}(x,y)$ are computable in logarithmic time (close to constant time), we think that it is very unlikely that they belong to $\cpp$\footnote{However, $\mathtt{gcd}(x,y)$ can be computed in constant time for $x<N^d$ and $y<N^{1/2-\varepsilon}$ for any $\varepsilon>0$ because we have 
$\mathtt{gcd}(x,y)=\mathtt{gcd}(y, x \modop y)$ and 
$x \modop y<y<N^{1/2-\varepsilon}$ and $\mathtt{gcd}(y, x \modop y)$ is obtained by a direct access to a table $GCD[a][b]$ for $a,b<N^{1/2-\varepsilon}$, which can be pre-computed in time $O(N^{1-2\varepsilon}\times \log N)=O(N)$ by Euclid's algorithm.}.

\medskip 
{\bf Operations computable by local models.}
As we proved in~\cite{GrandjeanJachiet22} the general result that each operation computable by (one-dimensional) cellular automata in linear time (such as $+,-,\times$, bitwise operations, etc.) is in $\cpp$, can this result be generalized again to more general ``local'' computation models such as e.g. \emph{multi-dimensional} cellular automata?
More importantly, can such a local computation model be general enough to incorporate the local division algorithm?

\medskip 
{\bf Operations computable by circuits.}
What about the complexity classes of operations defined by circuits such as $\textsc{AC}^0$?
This is an important question because several authors~\cite{AnderssonMRT96, AnderssonMT99, Thorup00} consider the $\textsc{AC}^0$ operations as the standard primitive operations of RAMs: 
an operation is $\textsc{AC}^0$ if it can be computed by a constant depth Boolean circuit of size polynomial in the length of operands.
Even though it is well-known that the $+$ and $-$ operations are $\textsc{AC}^0$ (see e.g.~\cite{Vollmer99}) so that $\cpp(+)\subseteq \cpp(\textsc{AC}^0)$, we think it is unlikely that all $\textsc{AC}^0$ operations are in $\cpp$. 
Any way, $\cpp$ contains all the operations mentioned by Thorup~\cite{Thorup00}, ``shifts, comparisons, bit-wise Boolean operations..., addition, subtraction, multiplication, and division... which, except for multiplication and division, are in $\textsc{AC}^0$...'' and called by Thorup~\cite{Thorup00} ``\emph{standard $\textsc{AC}^0$ operations} to emphasize the fact that they are available in most imperative programming languages.''
The paper~\cite{AnderssonMT99} argues: 
``...$\textsc{AC}^0$ is the class of functions which can be computed in constant time in a certain, reasonable, model of hardware. 
In this model, it is therefore the class of functions for which we can reasonably assume unit cost evaluation.''

\paragraph{Beyond $\cpp$: computations interspersed with inputs/outputs.}
If the notion of \emph{computation with preprocessing} and, more generally, the notion of \emph{computation interspersed with inputs/outputs} are very common in computer science and specially in database management systems, the study of the associated complexity is still in its infancy, with the exception of that of enumeration problems~\cite{AmarilliBJM17, BaganDG07, BerkholzGS20, Durand20, DurandG07, Segoufin14, Strozecki19, Uno15}, complexity of the $j$th solution problems, also called ``direct access'' problems~\cite{BaganDGO08, BringmannCM22, BringmannCM22arxiv, CarmeliTGKR23}, and complexity of dynamic problems~\cite{AmarilliJP21, BerkholzKS17, BerkholzKS18}.

\smallskip
As we did for the $\cpp$ class, we can justify the robustness and the large extent of all these classes of problems, as does for example Theorem~\ref{th:transitivity2} for the classes $\textsc{LinTime}$ and $\cdlin$. 
This is thanks to the ``fundamental lemma'' (Lemma~\ref{lemma:fund}) which demonstrates that all these classes are invariant if we authorize all $\cpp$ operations as primitive operations.
The statement of this fundamental lemma crucially requires the notion of \emph{faithful (or lock-step) simulation}, a very precise concept designed by Yuri Gurevich~\cite{DexterDG97, Gurevich93}.
In the end, let us recall that the lemma combines this notion with that of linear-time preprocessing to make the complexity classes concerned robust.


\medskip \noindent
\emph{Acknowledgments:} We thank the reviewers for their remarks and comments, and in particular for their suggestion to give ``some more intuition on how a proof works beforehand'' and ``the idea on the various constructions, maybe informally'', which helped us to improve the writing of the paper.
The first author would like to thank Yuri Gurevich for his numerous discussions and especially for his personal encouragement and support, especially in the 1980s and 1990s.
He also thanks Yuri for sharing with him his concern to identify "robust" complexity classes (and logic classes) with "good" computational models. In this paper, we also use Yuri's notion of ``lock-step'' simulation between programs.

\section{Appendices} 

\subsection{Computing $B\coloneqq \lc N^{1/2}\rc$ in linear time using only addition}\label{app:B}

The following algorithm, which only uses $+$ and the variables $\var{x}, \var{x}_2, \var{x}_{-1}, \var{y}$, runs in $O(N)$ time: 

\begin{minipage}{0.1\textwidth}
  ~
\end{minipage}
\begin{minipage}{0.80\textwidth}
\begin{algorithm}[H]
  \caption{Computation of $B \coloneqq \lc N^{1/2} \rc$}
  \begin{algorithmic}[1]
    \State $\var{x} \gets 1$ ; $\var{x}_2 \gets 1$ ; $\var{y}\gets 1$ 
    \While{$\var{y}\neq N+1$} \Comment{loop invariant: $(x-1)^2<y\le x^2$ and $x_2=x^2$}
     \If{$\var{y}=\var{x}_2$}
      \State $\var{x}_2 \gets \var{x}_2 +\var{x}+\var{x}+1$ \Comment{$x_2=(x+1)^2$}
      \State $\var{x}_{-1} \gets \var{x}$
      \State $\var{x}\gets \var{x}+1$ 
      \State $\var{y}\gets \var{y}+1$
      \EndIf
      \EndWhile
     \State $B\gets \var{x}_{-1}$ 
  \end{algorithmic}
\end{algorithm}
\end{minipage}

\begin{proof}[Correctness of the algorithm]
At line 8, we have $\var{y}=N+1$ and  
$(\var{x}-1)^2 <N+1\le \var{x}^2$, which yields $(\var{x}-1)^2 \le N < \var{x}^2$.
Since we also have $B=x-1$, we obtain $B^2\le N <(B+1)^2$.
\end{proof}

\subsection{Proof of the Approximation Lemma for the leading digit of a quotient}\label{subsec:ApproxLemma}

Let us prove Lemma~\ref{lemma:tildeq} of Section~\ref{sec:timesDiv} which we rephrase in the following self-contained version.

\medskip \noindent 
{\bf Lemma~\ref{lemma:tildeq}.}
Let $u,y$ be integers such that $y\le u < B\times y$ with $y=(y_{m-1},\dots,y_0)_B$, $\;y_{m-1}\ge 1$, and $u=(u_m,u_{m-1},\dots,u_0)_B$.
Define $u_{(2)} \coloneqq (u_m,u_{m-1})_B$ (the two leading digits of $u$). 
Assume that the leading digit $v \coloneqq y_{m-1}$ of $y$ satisfies $v\ge \lf B/2 \rf$.
Then the quotient $q \coloneqq \lf u/y \rf$ satisfies the inequalities 
$\tilde{q}-2 \le q \le \tilde{q}$
where $\tilde{q} \coloneqq \min \left(\lf u_{(2)} / v \rf, \;B-1 \right)$.

\begin{proof}
This proof copies step by step the proof of Theorems A and B of~\cite{Knuth69} with our notations.

First, let us prove $q \le \tilde{q}$.
It is true if $\tilde{q} = B-1$ because $q<B$. 
Otherwise we have $\tilde{q}= \lf u_{(2)} / v \rf$, hence $\tilde{q}v > (u_{(2)}/v -1) \times v$ and therefore
$\tilde{q}v \ge u_{(2)}-v+1$.
It follows that 
\begin{eqnarray*}
\lefteqn{u-\tilde{q}y \le u-\tilde{q}vB^{m-1}\le u - (u_{(2)}+1-v)B^{m-1} = }\\
& & u - u_{(2)} B^{m-1} - B^{m-1} + vB^{m-1} = \\
& & u_{m-2}B^{m-2}+\dots +u_1B+u_0 - B^{m-1} + vB^{m-1}<vB^{m-1} \le y
\end{eqnarray*}
Therefore, we have $u-\tilde{q}y <y$, hence $u<(\tilde{q}+1)y$. 
Since we also have $qy\le u$, we obtain $qy<(\tilde{q}+1)y$ from which it comes $q\le \tilde{q}$.

\medskip
Now, let us prove that $\tilde{q}\le q+2$ by contradiction. 
Assume that $\tilde{q}\ge q+3$.
By definition of $\tilde{q}$, we have 
$\tilde{q}\le u_{(2)}/v=u_{(2)}B^{m-1}/(vB^{m-1})\le u/(vB^{m-1})< u/(y-B^{m-1})$
because $y<(v+1)B^{m-1}$.
From the hypothesis $\tilde{q}\ge q+3$ and the relations $\tilde{q}<u/(y-B^{m-1})$ and $q>(u/y)-1$, we deduce
\[
3\le \tilde{q}-q < \frac{u}{y-B^{m-1}}-\frac{u}{y}+1 = \frac{u}{y} \left(\frac{B^{m-1}}{y-B^{m-1}} \right)+1.
\]
Therefore, we obtain
$u/y > 2(y-B^{m-1})/B^{m-1}\ge 2(v-1)$ because $y\ge vB^{m-1}$.
Finally, since $B-4\ge \tilde{q}-3 \ge q =\lf u/y\rf \ge 2v-2$, we have $(B/2)-1\ge v$, which contradicts the hypothesis $v\ge \lf B/2\rf$.
\end{proof}

\subsection{An addition expression for multiplication}\label{append:+expForMult}

Let us exhibit and justify an addition expression $E_{\times}(x,y)$ such that, for each $N\ge 2$, we have after a preprocessing of time $O(N)$, the equality $x\times y = E_{\times}(x,y)$, for all $x,y<N$. 
This is Example~\ref{ex:termMult} of Theorem~\ref{th:csttime->exp}.

\begin{proof}[An addition expression for multiplication and its justification]
We use the pre-computed arrays $\mathtt{Div}B$, $\mathtt{Mod}B$, $\mathtt{Mult}B$ and $A_{\times}$, see Section~\ref{sec:timesDiv}.
(Recall: for $B\coloneqq\lc N^{1/2}\rc$, we have
$\textsc{Mod}B[x]\coloneqq x\modop B$ and
$\textsc{Div}B[x]\coloneqq x\divop B$, for each $x< 3N$; 
we also have
$\textsc{Mult}B[x]\coloneqq x\times B$ and $A_{\times}[\textsc{Mult}B[x]+y]\coloneqq x\times y$, for all $x,y<B$.)

Define the expressions $x_1\coloneqq \textsc{Div}B[x]$, $x_0\coloneqq \mathtt{Mod}B[x]$, 
$y_1\coloneqq \textsc{Div}B[y]$ and $y_0\coloneqq \mathtt{Mod}B[y]$,
for the $B$-digits of $x,y$. 
The notation $(z_3,z_2,z_1,z_0)$ of the product $x\times y$ in base $B$ is given by 
\begin{center}
$z_0=(x_0 \times y_0)\modop B$, and $z_1= u \modop B$,\\ 
where $u$ is the sum $(x_0 \times y_0)\divop B + x_1 \times y_0 + x_0 \times y_1< 3N$;\\
$z_2= u' \modop B$, and $z_3= u' \divop B$, where $u'$ is the sum $u \divop B +x_1 \times y_1$.
\end{center}
Formally, we define the addition expressions
\begin{center} 
$z_0 \coloneqq \textsc{Mod}B[A_{\times}[\textsc{Mult}B[x_0]+y_0]]$,
and $z_1 \coloneqq \textsc{Mod}B[u]$, \\ 
where $u \coloneqq \textsc{Div}B[A_{\times}[\textsc{Mult}B[x_0]+y_0]]
+ A_{\times}[\textsc{Mult}B[x_1]+y_0] + A_{\times}[\textsc{Mult}B[x_0]+y_1]$;\\
$z_2  \coloneqq \textsc{Mod}B[u']$ and $z_3  \coloneqq \textsc{Div}B[u']$,
where $u'  \coloneqq \textsc{Div}B[u] + A_{\times}[\textsc{Mult}B[x_1]+y_1]$.
\end{center}

From $(z_3,z_2,z_1,z_0)$, we now construct a $2$-tuple addition expression $(t_1,t_0)$ representing the product $x\times y$ in base $N$. 
We use for that the tables
$\textsc{Div}N[0..2N-1]$, $\textsc{Mod}N[0..2N-1]$, $T_0[0..2N-1]$ and $T_1[0..2N-1]$, defined by 
$\textsc{Div}N[v]\coloneqq v \divop N$, 
$\textsc{Mod}N[v]\coloneqq v \modop N$,
$T_0[v]\coloneqq (v\times B) \modop N$ and $T_1[v]\coloneqq (v\times B) \divop N$.
Note that these tables can be pre-computed in time $O(N)$.
By definition, we have, for each $v<2N$, the identity
\begin{eqnarray}\label{eqn:vB}
v\times B = T_1[v] \times N + T_0[v].
\end{eqnarray}
Observe the equality $x \times y = ((z_3 \times B + z_2 )\times B + z_1) \times B + z_0$ 
with\footnote{From $z_3,z_2<B$  and $N>(B-1)^2=B^2 - 2B +1$, we deduce 
$z_3 \times B + z_2<B^2 \le N +2B-2 \le 2N$.}
$z_3 \times B + z_3<2N$, and define the expression 
$v_1 \coloneqq \textsc{Mult}B[z_3]+z_2$ ($v_1$ equals $z_3\times B + z_2$). 
Using identity~(\ref{eqn:vB}), we obtain 
$x \times y = (v_1 \times B + z_1) \times B + z_0 = (T_1[v_1]\times N + T_0[v_1]+z_1)\times B + z_0$.
We define the expression $v_2 \coloneqq T_0[v_1]+z_1< 2N$ so that we have
$x \times y = (T_1[v_1]\times B) \times N + v_2 \times B + z_0$,
which can be rewritten using again identity~(\ref{eqn:vB}),
$x \times y = (T_1[v_1]\times B) \times N + T_1[v_2]\times N + T_0[v_2] + z_0$.
This justifies that we must have 
\begin{center}
$t_0 = (T_0[v_2] + z_0) \modop N$
and $t_1 = T_1[v_1]\times B + T_1[v_2] + (T_0[v_2] + z_0) \divop N$.
\end{center}
Formally,
$E_{\times}(x,y) \coloneqq (t_1,t_0)$, where
$t_1 \coloneqq \textsc{Mult}B[T_1[v_1]] + T_1[v_2] + \textsc{Div}N[T_0[v_2] + z_0]$ and 
$t_0 \coloneqq \textsc{Mod}N[T_0[v_2] + z_0]$.
\end{proof}
\subsection{The Fibonacci and factorial functions belong to $\cpp$}\label{appendix:Fib fact}

Let $\mathtt{Fib}:\mathbb{N}\to\mathbb{N}$ denote the usual Fibonacci function defined by
$\mathtt{Fib}(0)\coloneqq 0$, $\mathtt{Fib}(1)\coloneqq 1$ and 
$\mathtt{Fib}(x) \coloneqq \mathtt{Fib}(x-1)+\mathtt{Fib}(x-2)$ for $x\ge 2$.

\begin{proposition} The Fibonacci function belongs to $\cpp$.
\end{proposition}

\begin{proof}
We want to compute its ``polynomial'' restriction $\mathtt{Fib}_{N,d}$. 
We first pre-compute the value $\arr{BOUND}$ such that
$\arr{BOUND}\coloneqq \max\{x \mid \mathtt{Fib}(x)<N^d\}$ and the array $\arr{FIB}[0..N-1]$ such that
$\arr{FIB}[x] \coloneqq \mathtt{Fib}(x)$ for all integers $x\leq \arr{BOUND}$.
Note that $\mathtt{BOUND}$ is $O(\log N^d)$, which justifies the $\mathtt{FIB}$ array being of size $N$. 
The following code pre-computes $\arr{BOUND}$ and $\arr{FIB}[0..N-1]$ in time 
$O(\arr{BOUND})=O(\log N^d)=O(N)$, which makes it possible to compute $\mathtt{Fib}_{N,d}$ in constant time.
\end{proof}

\begin{minipage}{0.05\textwidth}
        ~
      \end{minipage}
      \begin{minipage}{0.85\textwidth}
        \begin{algorithm}[H]\label{algo:pred}
          \caption{Pre-computation of $\arr{FIB}$ and computation of $\mathtt{Fib}_{N,d}$}
      \begin{minipage}[t]{0.50\textwidth}
      \begin{algorithmic}[1]
    \State $\arr{FIB}[0] \gets 0$ 
    \State $\var{x}\gets 1$ 
    \State $\var{z}\gets 1$
   \While{$\var{z}<N^d$}\label{line:while}
      \State $\arr{FIB}[\var{x}] \gets \var{z}$\label{line:fib}
      \State $\var{z} \gets \arr{FIB}[\var{x}-1] + \arr{FIB}[\var{x}]$
      \State $\var{x}\gets \var{x}+1$
    \EndWhile
    \State $\arr{BOUND} \gets \var{x}-1$\label{line:bound}
      \end{algorithmic}
      \end{minipage}
      \begin{minipage}[t]{0.55\textwidth}
      \vspace{0.9em}
      \begin{algorithmic}[1]
\Procedure{fibonacci}{$\var{x}$}
  \If{$\var{x} \leq \arr{BOUND}$}
    \State \Return $\arr{FIB}[\var{x}]$ 
  \Else 
    \State \Return $\mathtt{overflow}$
  \EndIf
\EndProcedure
      \end{algorithmic}
      \end{minipage}
    \end{algorithm}
      \end{minipage}

\medskip
\begin{proposition} The factorial function belongs to $\cpp$.
\end{proposition}

\begin{proof}
Here again, we want to compute the polynomial restriction $\mathtt{Fact}_{N,d}$.  
The following code pre-computes the value 
$\arr{BOUND}\coloneqq\mathtt{max}\{x \mid x!<N^d\}$ and the array $\arr{FACT}[0..N-1]$ such that
$\arr{FACT}[x] \coloneqq x!$ for all $x\leq \arr{BOUND}$ in time $O(\arr{BOUND})=O(\log N^d)=O(N)$
and computes $\arr{Fact}_{N,d}$ in constant time.
\end{proof}

\begin{minipage}{0.05\textwidth}
        ~
      \end{minipage}
      \begin{minipage}{0.85\textwidth}
        \begin{algorithm}[H]\label{algo:pred}
          \caption{Pre-computation of $\arr{FACT}$ and computation of $\mathtt{Fact}_{N,d}$}
      \begin{minipage}[t]{0.50\textwidth}
      \begin{algorithmic}[1]
    \State $\var{x}\gets 0$ 
    \State $\var{z}\gets 1$
   \While{$\var{z}<N^d$}
      \State $\arr{FACT}[\var{x}] \gets \var{z}$ 
      \State $\var{z} \gets (\var{x}+1) \times \arr{FACT}[\var{x}]$
      \State $\var{x}\gets \var{x}+1$
    \EndWhile
    \State $\arr{BOUND} \gets \var{x}-1$  
      \end{algorithmic}
      \end{minipage}
      \begin{minipage}[t]{0.55\textwidth}
      \vspace{0.5em}
      \begin{algorithmic}[1]
\Procedure{factorial}{$\var{x}$}
  \If{$\var{x} \leq \arr{BOUND}$}
    \State \Return $\arr{FACT}[\var{x}]$ 
  \Else 
    \State \Return $\mathtt{overflow}$
  \EndIf
\EndProcedure
      \end{algorithmic}
      \end{minipage}
    \end{algorithm}
      \end{minipage}

\subsection{Using registers of length $O(\log N)$ does not change complexity classes}\label{app:O(log N)}
Let us now state and justify how, as mentioned in Subsection~\ref{subsec:RAMmodel}, the complexity classes we have studied do not change if the contents of the RAM registers are less than $N^d$, for a fixed integer $d\ge 2$, (equivalently, register contents have length $O(\log N)$) instead of $O(N)$.

\begin{lemma}\label{lem:addressesNdToLin}
Let $c,d\ge 2$ be fixed integers and let $M$ be a \emph{RAM($+$)} whose register contents belong to the interval $[0,N^d[$ and whose inputs, outputs and used addresses are in $[0,cN[$.
Then $M$ is faithfully simulated by a  \emph{RAM($+,\times,\mathtt{div},\mathtt{mod}$)} $M'$ whose register contents are $O(N)$.
\end{lemma}

\begin{proof}
The idea is to simulate each register $R[a]$ of $M$ by $d$ registers $R[a \times d],\dots,R[a\times d +d-1]$ of $M'$ so that if the contents of $R[a]$ is $x<N^d$ then $R[a\times d +i]$ contains the digit $x_i$ of the representation $(x_{d-1},\dots,x_0)_N$ of $x$ in base $N$. 
(Note that $a\times d +i$ is $O(N)$ like $a$.)
From this representation, it is quite easy to simulate each instruction of $M$ by a fixed finite sequence of instructions of $M'$.
As an example, let us describe in detail the simulation of three representative instructions (see Table~\ref{table:instRAM2}):

\smallskip \noindent
• The ``$\mathtt{Input}\;j$'' instruction $\mathrm{Read}\;R[j]$ is simulated by the subroutine 
\begin{center}
$\mathrm{Read}\;R[j\times d]$ ; $R[j\times d+1] \gets R[j\times d]\divop N$ ;
$R[j\times d] \gets R[j\times d]\modop N$ \\ 
$\mathtt{for}$ $i$ $\mathtt{from}$ $2$ $\mathtt{to}$ $d-1$ $\mathtt{do}$ $R[j\times d+i] \gets 0$
\end{center}
\emph{Explanation:} Each input integer $x$ belongs to $[0,cN[$ and we can assume $N \ge c$, so that 
$cN\le N^2$. 
Therefore, the registers $R[j\times d+1]$ and $R[j\times d]$  (resp. $R[j\times d+i]$) for $2\le i<d$) of $M'$ receive the respective values $x \divop N$ and $x \modop N$ (resp. $0$) because the representation of $x$ with $d$ digits in base $N$ is 
$(x_{d-1},\dots,x_2,x_1,x_0)_N=(0,\dots,0,x_1,x_0)_N$ with $(x_1,x_0)=(x \divop N, x \modop N)$.

\smallskip \noindent
• The ``$\mathtt{Store}\;j\;k$'' instruction $R[R[j]]\gets R[k]$ is simulated by the ``for'' instruction
\begin{center}
$\mathtt{for}$ $i$ $\mathtt{from}$ $0$ $\mathtt{to}$ $d-1$ $\mathtt{do}$ 
$R[(R[j\times d+1]\times N + R[j\times d])\times d+i]\gets R[k\times d +i]$
\end{center}
\emph{Explanation:} By hypothesis, the ``address'' register $R[j]$ of $M$ contains an integer $a<cN$, so that its notation in base $N$ is of the form $(a_1,a_0)_N$; $R[j]$ is therefore represented by the registers \linebreak
$R[j \times d],R[j \times d +1],\dots,R[j\times d + d-1]$ of $M'$, whose contents are $0$, 
except the registers $R[j \times d +1]$ and $R[j \times d]$, 
whose contents are $a_1$ and $a_0$ respectively, so that
$a=R[j\times d+1]\times N + R[j\times d]$.

\smallskip \noindent
• The addition instruction $R[0]\gets R[0]+R[1]$ is simulated as follows by the classical addition algorithm using the variables $a_0,\dots,a_{d-1}$ and $b_0,\dots,b_{d-1}$ (representing the $d$ digits in base $N$ of the initial contents of $R[0]$ and $R[1]$, respectively), $c_0,\dots,c_{d}$ (for the $d+1$ carries of the addition in base $N$) and 
$s_0,\dots,s_{d}$ (for the $d+1$ digits of the sum):
\begin{center}
$\mathtt{for}$ $i$ $\mathtt{from}$ $0$ $\mathtt{to}$ $d-1$ $\mathtt{do}$ $(a_i,b_i)\gets (R[i],R[d+i])$ ;
$c_0\gets 0$\\
$\mathtt{for}$ $i$ $\mathtt{from}$ $0$ $\mathtt{to}$ $d-1$ $\mathtt{do}$ 
$(c_{i+1},s_i) \gets ((a_i+b_i+c_i)\divop N,(a_i+b_i+c_i)\modop N)$ \\
$s_d \gets a_d+b_d+c_d$ ;
$\mathtt{for}$ $i$ $\mathtt{from}$ $0$ $\mathtt{to}$ $d$ $\mathtt{do}$ $R[i]\gets s_i$
\end{center}
Each of the other instructions in Table~\ref{table:instRAM2} is simulated in the same way. 
Lemma~\ref{lem:addressesNdToLin} is proven.
\end{proof}

The following is a direct consequence of Lemma~\ref{lem:addressesNdToLin}.
\begin{proposition}\label{prop:addressesNdToLin}
Let $c,d\ge 2$ be fixed integers and let $M$ be a \emph{RAM($+$)} whose register contents belong to the interval $[0,N^d[$ and whose inputs, outputs and addresses used are in $[0,cN[$.

Then $M$ can be faithfully simulated by a  \emph{RAM($+$)} $M''$ using linear-time initialization so that the computation of $M''$, including initialization, only uses integers $O(N)$.
\end{proposition}

\begin{proof}
The program of $M''$ is that of the RAM($+,\times,\mathtt{div},\mathtt{mod}$) $M'$ of Lemma~\ref{lem:addressesNdToLin} (which faithfully simulates $M$), preceded by the linear preprocessing necessary to compute $\times$, $\mathtt{div}$ and $\mathtt{mod}$ from~$+$ in constant time, so that each expression involving $\times$, $\mathtt{div}$ or $\mathtt{mod}$ is replaced by an equivalent expression involving $+$ as the only operation, both in the preprocessing and in the computation itself.
\end{proof}

Here is our invariance result of complexity classes depending on the contents of RAM registers.

\begin{corollary}
The complexity classes $\cpp$, $\textsc{LinTime}$ and $\cdlin$ are invariant, whether the RAMs use only integers less than $cN$ \emph{or} less than $N^d$ as register contents, for fixed integers \linebreak 
$c,d\ge 2$, provided that the inputs, outputs\footnote{Of course, for the comparison to make sense, the inputs and outputs must be the same in both cases.} and addresses used are less than $cN$.
\end{corollary}

\begin{proof}
The transformation of a RAM($+$) with register contents less than $N^d$ into a RAM($+$) with register contents less than $cN$ is justified by Proposition~\ref{prop:addressesNdToLin}. 
The inverse implication is justified by the inequality $cN\le N^d$, for $N\ge c$ and $d\ge 2$. 
\end{proof}


\subsection{How to define time complexity classes beyond linear time on RAMs?}\label{app:nonlinear}

This cannot be achieved by maintaining the constraint that allowed integers are $O(N)$ because this implies that RAM uses linear space. 
The upper limit of allowed integers and allowed addresses should be increased according to the chosen time limit.

\begin{definition}[General RAM time complexity classes]\label{def:timeBeyondLinear}
For a function $T:\mathbb{N}\to\mathbb{N}$, $T(N)\ge N$, 
we denote by $\textsc{Time}(T(N))$ the class of problems $\Pi: X \mapsto \Pi(X)$, with input 
$X=(x_1,\dots,x_N)$ of size $N$, and $\;0\le x_i<cN$ (where $c$ is a constant integer), which are computed in time $O(T(N))$ by RAMs whose registers can contain integers $O(T(N))$, instead of $O(N)$, and which can also use addresses $O(T(N))$.
\end{definition}

\begin{remark}
In addition to $N$, a RAM of time complexity $O(T(N))$ has a second reference integer, $N_1\coloneqq T(N)$. Note that for usual time functions $T$, for example $T(N)=N^2$ or $T(N)=2^N$, \linebreak
$N_1$ can be easily computed from $N$.
\end{remark}

The general complexity class $\textsc{Time}(T(N))$ is robust like its linear version 
$\textsc{LinTime}\coloneqq \textsc{Time}(N)$ as expressed for example by the following theorem which generalizes the part of Theorem~\ref{th:+=+...mod} concerning $\textsc{LinTime}$.
\begin{theorem}\label{th:robustT}
The complexity class $\textsc{Time}(T(N))$ is invariant depending on whether the set of primitive operations is $\{+\}$ or $\{+,-,\times,\mathtt{div},\mathtt{mod}\}$ or is a set of $\cpp$ operations including $+$.
\end{theorem}

\begin{proof}[Sketch of proof]
The result follows from Lemma~\ref{lemma:fund} (Fundamental Lemma) and the first item of Theorem~\ref{th:+=+...mod} by replacing the first reference integer $N$ by the second reference integer $N_1\coloneqq T(N)$.
Thus, the initialization time $O(N)$ mentioned in Lemma~\ref{lemma:fund} is replaced by $O(N_1)=O(T(N))$.
\end{proof}


\subsection{Why non-polynomial preprocessing time expands $\cpp$?}\label{app:NonPol-CstPP}

Let us now recall and prove Theorem~\ref{th:NonPol-CstPP} 
of Subsection~\ref{subsec:polPP}, which contrasts with Theorem~\ref{th:N^c->N} 
(invariance of the $\cpp$ class for polynomial-time preprocessing).

\medskip \noindent 
{\bf Theorem~\ref{th:NonPol-CstPP}.}
Let $T(N)$ be a time function greater than any polynomial, i.e. such that for every~$c$, there exists an integer $N_0$ such that $T(N)>N^c$ for each $N\ge N_0$.

Then the inclusion $\cpp_N \subset \cpp_{T(N)}$ is strict.

\medskip\noindent
The theorem is essentially deduced from a variant of the time hierarchy theorem~\cite{CookR73} for RAMs.
\begin{theorem}[Time Hierarchy Theorem]\label{hierarchyTh}
Let $T(N)$ be a time function such that $T(N)>N$ and $T(N)$ is time
constructible, i.e. there is a RAM which, for each reference integer
$N\ge 2$, computes $T(N)$ in time $O(T(N))$. Then there exists a
function $f:\mathbb{N}\to\{0,1\}$ computable by a RAM in time
$O(T(N))$ but computable by no RAM in time $o(T(N))$.
\end{theorem}

\begin{proof}[Sketch of proof of Theorem~\ref{hierarchyTh}]
Our proof follows the “classical” proof for Time Hierarchy theorems
with a diagonalization argument.  For that purpose, we suppose the
existence of a function that associates to each RAM program $P$ a
binary code, $\mathtt{code}(P)$, identified to the integer~$N$ it
represents.  We also suppose the existence of its reverse
$\mathtt{decode}(\mathtt{code}(P)) = P$ for any RAM~$P$ and we turn
that function into $\mathtt{decode}$-$\mathtt{padding}(w)$ that takes
any binary word $w$ and returns $\mathtt{decode}(v)$ when~$w$ is of
the form $1^k0v$ for any $k\ge 1$ and $\mathtt{decode}(w)$ otherwise.
This $\mathtt{decode}$-$\mathtt{padding}$ function ensures that for
any RAM $P$ and for any integer $N_0$ we can find a code $N>N_0$ such
that $P = \mathtt{decode}$-$\mathtt{padding}(N)$.  The existence of
functions $\mathtt{code}$ and $\mathtt{decode}$ is pretty
straightforward, therefore we don’t include it in this proof.  We will
simply assume that the functions $\mathtt{code}$ and $\mathtt{decode}$
are computed in $O(N)$ time.

\medskip
Let us now define $f$ as the function from $\mathbb{N}$ to $\{0, 1\}$
computed by the following RAM program called~$M$:

\begin{itemize}
\item From the reference integer $N$, compute $t = T(N)$ in $O(T(N))$
  time;
\item Compute $P = \mathtt{decode}$-$\mathtt{padding}(N)$ in time
  $O(N) = O(T(N))$;
\item Execute $t$ steps of program $P$ acting on the reference integer
  $N$ (or less steps if $P$ stops), this takes $O(1)$ time per step,
  hence total time $O(T(N))$;
\item Set $V$ to the value returned by $P$ if it stopped or $0$ if it
  did not stop;
\item Return $1 - V$.
\end{itemize}

By construction, $M$ computes $f$ in $O(T(N))$ time.  Let us now
suppose there is a RAM program $P$ computing $f$ in $o(T(N))$ time.
By definition of $o(T(N))$, we know that there exists an integer $N_0$
such that $P$ runs in less than $T(N)$ steps for each input integer $N
> N_0$.  By construction of $\mathtt{decode}$-$\mathtt{padding}$, we
know that there exists an $N > N_0$ such that
$\mathtt{decode}$-$\mathtt{padding}(N) = P$, but then we can look at
the value $V = f(N)$. The value $V$ should be the value returned by
$P$ on $N$ as $P$ terminates in less than $T(N)$ steps on $N$, but by
construction $f(N)$ is $1-V$, which means we have a contradiction.
Therefore, no RAM can compute $f$ in time~$o(T(N))$.

\end{proof}

The intuition of the proof of Theorem~\ref{th:NonPol-CstPP} is that the $\cpp_{t(N)}$ classes are framed by classical time complexity classes to which the Time Hierarchy Theorem~\ref{hierarchyTh} applies. We are now ready to give the proof of \ref{th:NonPol-CstPP}.

\begin{proof}[Proof of Theorem~\ref{th:NonPol-CstPP}]
By Theorem~\ref{hierarchyTh}, there exist functions
$\mathbb{N}\to\{0,1\}$ computable in time $O(x^2)$ but not in time
$o(x^2)$ because the function $x\mapsto x^2$ is time constructible,
let $f$ be such a function.

It is clear that $f\not\in \cpp_N$, as otherwise we would have a RAM
computing $f(N)$ in $O(N)=o(N^2)$ (using a linear preprocessing and a
constant-time evaluation).

To finish the proof, we need to show that $f\in \cpp_{T(N)}$, which we
do by examining the following RAM program:
\begin{itemize}
  \item \emph{Preprocessing:} 
From the reference integer $N\ge 2$, construct and store an array $A[0..N^d-1]$ by computing $f(x)$ for each index $x<N^d$ and setting $A[x] \gets f(x)$ ; 
\item \emph{Computation:} Read an integer $x<N^d$ and return the value $A[x]$.
\end{itemize}

This preprocessing computes $N^d$ values for $f$, each in time $O(
(N^d)^2)$, therefore the total preprocessing time is $O(N^{3d})$ but
$T(N)$ grows faster than any polynomial function and therefore we have
that the preprocessing is indeed in time $O(T(N))$ which proves that
$f\in\cpp_{T(N)}$.
\end{proof}


\bibliographystyle{plain}
\bibliography{biblioConstTime}

\begin{thebibliography}{10}

\bibitem{AhoHU74}
Alfred~V. Aho, John~E. Hopcroft, and Jeffrey~D. Ullman.
\newblock {\em The Design and Analysis of Computer Algorithms}.
\newblock Addison-Wesley, 1974.

\bibitem{AmarilliBJM17}
Antoine Amarilli, Pierre Bourhis, Louis Jachiet, and Stefan Mengel.
\newblock A circuit-based approach to efficient enumeration.
\newblock In {\em 44th International Colloquium on Automata, Languages, and
  Programming, {ICALP} 2017, July 10-14, 2017, Warsaw, Poland}, pages
  111:1--111:15, 2017.

\bibitem{AmarilliJP21}
Antoine Amarilli, Louis Jachiet, and Charles Paperman.
\newblock Dynamic membership for regular languages.
\newblock {\em CoRR}, abs/2102.07728, 2021.

\bibitem{AnderssonMRT96}
Arne Andersson, Peter~Bro Miltersen, S{\o}ren Riis, and Mikkel Thorup.
\newblock Static dictionaries on ac\({}^{\mbox{0}}\) rams: Query time
  theta(sqrt(log n/log log n)) is necessary and sufficient.
\newblock In {\em 37th Annual Symposium on Foundations of Computer Science,
  {FOCS} '96, Burlington, Vermont, USA, 14-16 October, 1996}, pages 441--450.
  {IEEE} Computer Society, 1996.

\bibitem{AnderssonMT99}
Arne Andersson, Peter~Bro Miltersen, and Mikkel Thorup.
\newblock Fusion trees can be implemented with ac\({}^{\mbox{0}}\) instructions
  only.
\newblock {\em Theor. Comput. Sci.}, 215(1-2):337--344, 1999.

\bibitem{BaganDG07}
Guillaume Bagan, Arnaud Durand, and Etienne Grandjean.
\newblock On acyclic conjunctive queries and constant delay enumeration.
\newblock In Jacques Duparc and Thomas~A. Henzinger, editors, {\em Computer
  Science Logic, 21st International Workshop, {CSL} 2007, 16th Annual
  Conference of the EACSL, Lausanne, Switzerland, September 11-15, 2007,
  Proceedings}, volume 4646 of {\em Lecture Notes in Computer Science}, pages
  208--222. Springer, 2007.

\bibitem{BaganDGO08}
Guillaume Bagan, Arnaud Durand, Etienne Grandjean, and Fr{\'e}d{\'e}ric Olive.
\newblock Computing the jth solution of a first-order query.
\newblock {\em ITA}, 42(1):147--164, 2008.

\bibitem{BerkholzGS20}
Christoph Berkholz, Fabian Gerhardt, and Nicole Schweikardt.
\newblock Constant delay enumeration for conjunctive queries: a tutorial.
\newblock {\em {ACM} {SIGLOG} News}, 7(1):4--33, 2020.

\bibitem{BerkholzKS17}
Christoph Berkholz, Jens Keppeler, and Nicole Schweikardt.
\newblock Answering conjunctive queries under updates.
\newblock {\em CoRR}, abs/1702.06370, 2017.

\bibitem{BerkholzKS18}
Christoph Berkholz, Jens Keppeler, and Nicole Schweikardt.
\newblock Answering {FO+MOD} queries under updates on bounded degree databases.
\newblock {\em {ACM} Trans. Database Syst.}, 43(2):7:1--7:32, 2018.

\bibitem{BringmannCM22}
Karl Bringmann, Nofar Carmeli, and Stefan Mengel.
\newblock Tight fine-grained bounds for direct access on join queries.
\newblock In Leonid Libkin and Pablo Barcel{\'{o}}, editors, {\em {PODS} '22:
  International Conference on Management of Data, Philadelphia, PA, USA, June
  12 - 17, 2022}, pages 427--436. {ACM}, 2022.

\bibitem{BringmannCM22arxiv}
Karl Bringmann, Nofar Carmeli, and Stefan Mengel.
\newblock Tight fine-grained bounds for direct access on join queries.
\newblock {\em CoRR}, abs/2201.02401, 2022.

\bibitem{CarmeliK18}
Nofar Carmeli and Markus Kr{\"{o}}ll.
\newblock Enumeration complexity of conjunctive queries with functional
  dependencies.
\newblock In Benny Kimelfeld and Yael Amsterdamer, editors, {\em 21st
  International Conference on Database Theory, {ICDT} 2018, March 26-29, 2018,
  Vienna, Austria}, volume~98 of {\em LIPIcs}, pages 11:1--11:17. Schloss
  Dagstuhl - Leibniz-Zentrum f{\"{u}}r Informatik, 2018.

\bibitem{CarmeliTGKR23}
Nofar Carmeli, Nikolaos Tziavelis, Wolfgang Gatterbauer, Benny Kimelfeld, and
  Mirek Riedewald.
\newblock Tractable orders for direct access to ranked answers of conjunctive
  queries.
\newblock {\em {ACM} Trans. Database Syst.}, 48(1):1:1--1:45, 2023.

\bibitem{CookR73}
Stephen~A. Cook and Robert~A. Reckhow.
\newblock Time bounded random access machines.
\newblock {\em J. Comput. Syst. Sci.}, 7(4):354--375, 1973.

\bibitem{CormenLRS09}
Thomas~H. Cormen, Charles~E. Leiserson, Ronald~L. Rivest, and Clifford Stein.
\newblock {\em Introduction to Algorithms, 3rd Edition}.
\newblock {MIT} Press, 2009.

\bibitem{DarwicheM02}
Adnan Darwiche and Pierre Marquis.
\newblock A knowledge compilation map.
\newblock {\em J. Artif. Intell. Res.}, 17:229--264, 2002.

\bibitem{Date81}
C.~J. Date.
\newblock {\em An Introduction to Database Systems, 3rd Edition}.
\newblock Addison-Wesley, 1981.

\bibitem{DexterDG97}
Scott~D. Dexter, Patrick Doyle, and Yuri Gurevich.
\newblock Gurevich abstract state machines and schoenhage storage modification
  machines.
\newblock {\em J. Univers. Comput. Sci.}, 3(4):279--303, 1997.

\bibitem{Durand20}
Arnaud Durand.
\newblock Fine-grained complexity analysis of queries: From decision to
  counting and enumeration.
\newblock In Dan Suciu, Yufei Tao, and Zhewei Wei, editors, {\em Proceedings of
  the 39th {ACM} {SIGMOD-SIGACT-SIGAI} Symposium on Principles of Database
  Systems, {PODS} 2020, Portland, OR, USA, June 14-19, 2020}, pages 331--346.
  {ACM}, 2020.

\bibitem{DurandG07}
Arnaud Durand and Etienne Grandjean.
\newblock First-order queries on structures of bounded degree are computable
  with constant delay.
\newblock {\em {ACM} Trans. Comput. Log.}, 8(4):21, 2007.

\bibitem{EldarCK24}
Idan Eldar, Nofar Carmeli, and Benny Kimelfeld.
\newblock Direct access for answers to conjunctive queries with aggregation.
\newblock In Graham Cormode and Michael Shekelyan, editors, {\em 27th
  International Conference on Database Theory, {ICDT} 2024, March 25-28, 2024,
  Paestum, Italy}, volume 290 of {\em LIPIcs}, pages 4:1--4:20. Schloss
  Dagstuhl - Leibniz-Zentrum f{\"{u}}r Informatik, 2024.

\bibitem{GrandjeanRT12}
Ana{\"{e}}l Grandjean, Ga{\'{e}}tan Richard, and V{\'{e}}ronique Terrier.
\newblock Linear functional classes over cellular automata.
\newblock In Enrico Formenti, editor, {\em Proceedings 18th international
  workshop on Cellular Automata and Discrete Complex Systems and 3rd
  international symposium Journ{\'{e}}es Automates Cellulaires, {AUTOMATA} {\&}
  {JAC} 2012, La Marana, Corsica, September 19-21, 2012}, volume~90 of {\em
  {EPTCS}}, pages 177--193, 2012.

\bibitem{Grandjean96}
Etienne Grandjean.
\newblock Sorting, linear time and the satisfiability problem.
\newblock {\em Ann. Math. Artif. Intell.}, 16:183--236, 1996.

\bibitem{GrandjeanJachiet22}
Etienne Grandjean and Louis Jachiet.
\newblock Which arithmetic operations can be performed in constant time in the
  {RAM} model with addition?
\newblock {\em CoRR}, abs/2206.13851, 2022.

\bibitem{GrandjeanSchwentick02}
Etienne Grandjean and Thomas Schwentick.
\newblock Machine-independent characterizations and complete problems for
  deterministic linear time.
\newblock {\em SIAM J. Comput.}, 32(1):196--230, 2002.

\bibitem{Gurevich93}
Yuri Gurevich.
\newblock Evolving algebras: an attempt to discover semantics.
\newblock In Grzegorz Rozenberg and Arto Salomaa, editors, {\em Current Trends
  in Theoretical Computer Science - Essays and Tutorials}, volume~40 of {\em
  World Scientific Series in Computer Science}, pages 266--292. World
  Scientific, 1993.

\bibitem{HorowitzS84}
Ellis Horowitz and Sartaj Sahni.
\newblock {\em Fundamentals of data structures in Pascal}.
\newblock Pitman, 1984.

\bibitem{JohnsonPY88}
David~S. Johnson, Christos~H. Papadimitriou, and Mihalis Yannakakis.
\newblock On generating all maximal independent sets.
\newblock {\em Inf. Process. Lett.}, 27(3):119--123, 1988.

\bibitem{Knuth69}
Donald~E. Knuth.
\newblock {\em The Art of Computer Programming, Volume {II:} Seminumerical
  Algorithms}.
\newblock Addison-Wesley, 1969.

\bibitem{Marquis15}
Pierre Marquis.
\newblock Compile!
\newblock In Blai Bonet and Sven Koenig, editors, {\em Proceedings of the
  Twenty-Ninth {AAAI} Conference on Artificial Intelligence, January 25-30,
  2015, Austin, Texas, {USA}}, pages 4112--4118. {AAAI} Press, 2015.

\bibitem{Papadimitriou94}
Christos~H. Papadimitriou.
\newblock {\em Computational complexity}.
\newblock Addison-Wesley, 1994.

\bibitem{PopeS60}
David~A. Pope and Marvin~L. Stein.
\newblock Multiple precision arithmetic.
\newblock {\em Commun. {ACM}}, 3(12):652--654, 1960.

\bibitem{SandersMDD19}
Peter Sanders, Kurt Mehlhorn, Martin Dietzfelbinger, and Roman Dementiev.
\newblock {\em Sequential and Parallel Algorithms and Data Structures - The
  Basic Toolbox}.
\newblock Springer, 2019.

\bibitem{Schwentick97}
Thomas Schwentick.
\newblock Algebraic and logical characterizations of deterministic linear time
  classes.
\newblock In R{\"{u}}diger Reischuk and Michel Morvan, editors, {\em {STACS}
  97, 14th Annual Symposium on Theoretical Aspects of Computer Science,
  L{\"{u}}beck, Germany, February 27 - March 1, 1997, Proceedings}, volume 1200
  of {\em Lecture Notes in Computer Science}, pages 463--474. Springer, 1997.

\bibitem{Segoufin14}
Luc Segoufin.
\newblock A glimpse on constant delay enumeration (invited talk).
\newblock In {\em 31st International Symposium on Theoretical Aspects of
  Computer Science {(STACS} 2014), {STACS} 2014, March 5-8, 2014, Lyon,
  France}, pages 13--27, 2014.

\bibitem{ShepherdsonS63}
John~C. Shepherdson and Howard~E. Sturgis.
\newblock Computability of recursive functions.
\newblock {\em J. {ACM}}, 10(2):217--255, 1963.

\bibitem{Strozecki19}
Yann Strozecki.
\newblock Enumeration complexity.
\newblock {\em Bull. {EATCS}}, 129, 2019.

\bibitem{Terrier12}
V{\'{e}}ronique Terrier.
\newblock Language recognition by cellular automata.
\newblock In Grzegorz Rozenberg, Thomas B{\"{a}}ck, and Joost~N. Kok, editors,
  {\em Handbook of Natural Computing}, pages 123--158. Springer, 2012.

\bibitem{Thorup00}
Mikkel Thorup.
\newblock On {RAM} priority queues.
\newblock {\em {SIAM} J. Comput.}, 30(1):86--109, 2000.

\bibitem{UllmanW97}
Jeffrey~D. Ullman and Jennifer Widom.
\newblock {\em A First Course in Database Systems}.
\newblock Prentice-Hall, 1997.

\bibitem{Uno15}
Takeaki Uno.
\newblock Constant time enumeration by amortization.
\newblock In Frank Dehne, J{\"{o}}rg{-}R{\"{u}}diger Sack, and Ulrike Stege,
  editors, {\em Algorithms and Data Structures - 14th International Symposium,
  {WADS} 2015, Victoria, BC, Canada, August 5-7, 2015. Proceedings}, volume
  9214 of {\em Lecture Notes in Computer Science}, pages 593--605. Springer,
  2015.

\bibitem{Vollmer99}
Heribert Vollmer.
\newblock {\em Introduction to Circuit Complexity - A Uniform Approach}.
\newblock Texts in Theoretical Computer Science. An {EATCS} Series. Springer,
  1999.

\bibitem{vonzurGathenG2013}
Joachim von~zur Gathen and J{\"{u}}rgen Gerhard.
\newblock {\em Modern Computer Algebra {(3.} ed.)}.
\newblock Cambridge University Press, 2013.

\end{thebibliography}

\end{document}